\documentclass[a4paper,11pt]{article}
\usepackage[utf8]{inputenc}
\usepackage[british]{babel}
\usepackage[mono=false]{libertine}
\usepackage[T1]{fontenc}
\usepackage{amsthm}
\usepackage[top=2.5cm, bottom=2.5cm, left=2cm, right=2cm]{geometry}
\usepackage[font=small]{caption}
\usepackage{enumitem}
\usepackage{titlesec}
\usepackage[usenames, dvipsnames]{color}
\makeatletter
\newcommand{\setappendix}{Appendix~\thesection:~~}
\newcommand{\setsection}{\thesection~~}
\titleformat{\section}{\bfseries\LARGE}{%
	\ifnum\pdfstrcmp{\@currenvir}{appendices}=0
	\setappendix
	\else
	\setsection
\fi}{0em}{}
\makeatother
\usepackage[titletoc]{appendix}
\usepackage[pdftex]{graphicx}
\graphicspath{./plots}
\usepackage{array}
\usepackage{url}
\usepackage{mathtools}
\usepackage{amssymb,amsfonts,amsmath}
\usepackage{dsfont}
\usepackage{stmaryrd}
\usepackage{bm}
\usepackage{footnote}
\usepackage[utf8]{inputenc} 
\usepackage[T1]{fontenc}    
\usepackage{booktabs}       
\usepackage{nicefrac}       
\usepackage{microtype}      
\usepackage{tikz}
\usepackage{subcaption}
\usepackage{algorithm}
\usepackage[noend]{algpseudocode}
\usepackage{wrapfig}
\usepackage{xr}
\usepackage{xr-hyper}
\usepackage{braket}
\usepackage[pdfpagemode=UseNone,bookmarksopen=false,colorlinks=true,urlcolor=blue,citecolor=magenta,citebordercolor=blue,linkcolor=blue]{hyperref}
\newcommand{\<}{\langle}
\renewcommand{\>}{\rangle}

\newcommand{\sG}{\mathcal{G}}
\newcommand{\sB}{\mathcal{B}}
\newcommand{\sH}{\mathcal{H}}
\newcommand{\sR}{\mathcal{R}}
\newcommand{\sX}{\mathcal{X}}

\newcommand{\sfx}{\mathsf{x}}

\newcommand{\sfc}{\mathsf{c}}
\newcommand{\sfh}{\mathsf{h}}
\newcommand{\sfT}{\mathsf{T}}

\newtheorem{theorem}{Theorem}[section]
\newtheorem{lemma}[theorem]{\textbf{Lemma}}

\newtheorem{remark}[theorem]{\textbf{Remark}}
\newtheorem{proposition}[theorem]{\textbf{Proposition}}

\newtheorem{definition}[theorem]{\textbf{Definition}}

\DeclareMathAlphabet{\varmathbb}{U}{bbold}{m}{n}

\usepackage{setspace}

\begin{document}
\title{Adaptive path interpolation method for sparse systems:\\ Application to a censored block model}
\author{Jean Barbier$^{\star*}$, Chun Lam Chan$^{\dagger}$ and Nicolas Macris$^{\dagger}$}
\date{}
\maketitle
{\let\thefootnote\relax\footnote{
\hspace{-18.5pt}
$\star$ Quantitative Life Sciences, International Center for Theoretical Physics, Trieste, Italy.\\
$*$ Statistical Physics Laboratory, \'Ecole Normale Sup\'erieure, Paris, France.\\
$\dagger$ Communication Theory Laboratory, \'Ecole Polytechnique F\'ed\'erale de Lausanne, Switzerland.
}}\setcounter{footnote}{0}
\begin{abstract}
Recently a new adaptive path interpolation method has been developed as a simple and versatile scheme to calculate exactly the asymptotic 
mutual information  of Bayesian inference problems defined on {\it dense} factor graphs. These include random linear and generalized estimation, sparse superposition codes, 
or low-rank matrix and tensor estimation. For all these systems, the adaptive interpolation method directly proves that the replica symmetric prediction is exact, in a simple and unified manner. 
When the underlying factor graph of the inference problem is {\it sparse} the replica prediction is considerably 
more complicated, and rigorous results are often lacking or obtained by rather complicated methods. In this 
work we show how to extend the adaptive path interpolation method to sparse systems. We concentrate on a Censored Block Model, where hidden 
variables are measured through a binary erasure channel, for which we fully prove the replica prediction. 
\end{abstract}
{
	\singlespacing
	\hypersetup{linkcolor=black}
	\tableofcontents
}
\setcounter{tocdepth}{2}
\section{Introduction}
Much progress has been achieved recently in Bayesian inference of high dimensional problems. It has been possible to develop 
rigorous methods in order to derive exact ``single letter'' variational formulas for the mutual information 
in the asymptotic limit of the number of variables tending to infinity, when the prior and all hyperparameters of the problem are assumed to be known (this is referred to as the Bayes-optimal setting). 
Such formulas have often been first conjectured on the basis of the replica and cavity methods of statistical mechanics of disordered spin systems 
and are also known as ``replica symmetric'' formulas \cite{MezardParisi87b,MezardMontanari09}. Examples where {\it full} proofs have been achieved are 
random linear estimation and compressed sensing \cite{barbier_allerton_RLE,barbier_ieee_replicaCS,ReevesP16}, learning for single layer networks
\cite{barbier2017phase}, generalized estimation in multi-layer settings \cite{2018arXiv180509785G,2018arXiv180605451A}, or low-rank matrix and tensor estimation \cite{XXT,2017arXiv170200473M,2017arXiv170108010L,2016arXiv161103888L,2017arXiv170910368B}. Invariably, the Guerra-Toninelli interpolation method \cite{guerra2002thermodynamic} has been used to derive one-sided bounds (we note that \cite{ReevesP16} is exceptional and does not seem 
to rely directly on the same interpolation). For the converse bounds typically other ideas have usually been 
necessary, such as 
spatial coupling \cite{XXT,barbier_allerton_RLE,barbier_ieee_replicaCS} or the rigorous version of the cavity method  \cite{2017arXiv170200473M,2017arXiv170108010L,2016arXiv161103888L}. Recently two of us introduced 
a new interpolation scheme, called {\it adaptive path interpolation method}, that allows to derive the replica symmetric formulas in a more straightforward and unified manner 
\cite{BarbierM17a}, \cite{BarbierMacrisCW2019}. The new method is quite generic once the mean field solution has been identified and 
is directly applicable when the concentration 
of the ``overlap'' can be proved. Roughly speaking in Bayesian inference the ``overlap'' is the inner product between the random vector to be inferred and the ground truth vector. 
For models on dense graphs this concentration follows from variants of Ghirlanda-Guerra identities \cite{GG} adapted to Bayesian inference combined with the so-called Nishimori identities.

The successes of the adaptive interpolation method have so far been limited to inference models with a dense underlying factor graph. 
It is therefore desirable to see to what extent the method can be developed when the factor graph is instead sparse. 
Typical examples of such systems are Low-Density Parity-Check codes, Low-Density Generator-Matrix (LDGM) codes, or the Stochastic and Censored Block Models. 
It is fair to say that the replica symmetric formulas for the mutual information is much more complicated in such models. 
Indeed, besides the measurements (or channel outputs): (i) the sparse graph is also random; (ii) the single letter variational problem involves a functional 
over a set of probability distributions (instead of scalars as in the dense graph case). Existing rigorous 
derivations of the replica formulas have so far been achieved using a combination of the interpolation method 
(first developed by \cite{FranzLeone} for sparse models) and spatial coupling \cite{MacrisGiurgiuUrbanke2016} or the rigorous 
version of the cavity method \cite{aizenman2003extended,Coja-Oghlan:2017}. 

In this work we 
consider a simple version of the Censored Block Model \cite{6949658,abbe2013conditional,DBLP:journals/corr/ChinRV15,7282642}, for which we 
fully develop the adaptive interpolation method. We believe that this constitutes a first step towards an 
analysis of more complicated models via this method. A summary of the present analysis was presented in \cite{2018arXiv180605121B}.

In the Censored Block Model one has a set of $n$ hidden binary variables. One observes $\alpha n$ products of random $K$-tuples, where 
$\alpha > 0$ is called the fraction of measurements, through a noisy channel. The goal is to reconstruct an estimate of the hidden variables from the noisy observations. 
There are other interpretations of this model. For example, it can be interpreted as a Low-Density Generator-Matrix code ensemble, with design communication rate $1/\alpha$, on a 
factor graph with degree $K$ factor nodes and variable nodes with Poisson degrees (when $n\to +\infty$). Another possible interpretation is, 
as a model of statistical mechanics, namely an Ising model on a sparse random graph with $K$-spin interactions. 
The Censored Block Model has been discussed when the measurement channel is a Binary Symmetric Channel in \cite{6949658} and the 
replica formula proven in this case \cite{Coja-Oghlan:2017}. Here, we consider a simpler situation where the measurement channel is the Binary Erasure 
Channel (BEC) for which the adaptive interpolation method can be completely developed. As we will see this method requires concentration results for a whole set of suitable ``overlaps'' 
and requires new ideas in the case of sparse graphs.
Here, we solve this issue for the BEC, and it is currently the only aspect of the method that is missing for extending our analysis to other channels.

The paper is organized as follows. In Section \ref{sec:2} we give a precise formulation of the model and state the main result of this paper (Theorem \ref{RSconj}). 
In Section \ref{prelimtool} we review 
two important tools used throughout our analysis, namely the Nishimori identities and the Griffiths-Kelly-Sherman inequalities. The adaptive interpolation method for the sparse graph models
is formulated in Section \ref{interp} and the core of the proof of Theorem \ref{RSconj} is also developed. This section contains the main new technical ideas of this paper. Overlap concentration 
is proved in Section \ref{sec:overlap} and \ref{sec:thermal-and-quenched}. A series of more technical results are found in Section \ref{sec:proof} and in the appendices.

%
%

%
%
%
%
%
\section{Setting and main result}
\label{sec:2}
\subsection{Censored Block Model}
\label{sec:partI}

We shall denote binary variables by $\sigma_i\in \{-1, +1\}$, $i=1,\ldots, n$ and vectors of such variables by $\underline{\sigma} = (\sigma_1, \ldots, \sigma_n) \in \{ -1, +1\}^{n}$. 
Subsets $S\subset \{1, \ldots, n\}$ with at least two elements are always denoted by capital letters. For the product of binary variables 
in a subset $S$ we use the shorthand notation $\sigma_S \equiv \prod_{i\in S} \sigma_i$. If there is a possible confusion between small and capital 
letter subscripts we occasionally use more specific notations. Below, the integer $K \geq 2$ and the fraction $\alpha \in \mathbb{R}_+$ are fixed independent of $n$.

In the Censored Block Model considered in this paper $n$ {\it hidden} binary variables $\underline{\sigma}^0 = (\sigma_1^0, \ldots, \sigma_n^0)$ are i.i.d. uniform, i.e. drawn 
independently according to a ${\rm Ber}(1/2)$ prior $P_0(\sigma_i^0) = \frac12 \delta_{\sigma_i^0, +1} + \frac12 \delta_{\sigma_i^0, -1}$. A {\it noiseless measurement} consists
in a product $\sigma_{a_1}^0 \sigma_{a_2}^0 \ldots \sigma_{a_K}^0$ of a $K$-tuple
of variables drawn uniformly at random. The $K$-tuple is identified 
with a subset $A \equiv \{a_1, \ldots, a_K\}\subset \{1, \ldots, n\}$ and we set
$\sigma_A^0 \equiv \sigma_{a_1}^0 \sigma_{a_2}^0 \ldots \sigma_{a_K}^0$. Of course $\sigma_A^0 = \pm 1$.
The {\it true observations} $J_A\in \mathbb{R}$ are noisy versions of these products 
obtained through a binary input memoryless channel described by some transition probability
$Q(J_A\vert \sigma_A^0)$.
For large $n$ the total number of observations $m$ asymptotically follows a 
Poisson distribution with mean $\alpha n$, i.e., $m \sim \mathrm{Poi}(\alpha n)$. 
We shall also index the observations as $A= 1,\ldots, m$. 

Let us now describe the Bayesian setting used here to determine the information theoretic limits for reconstructing the hidden variables.
From the Bayes rule we have that the posterior given the observations is
\begin{align*}
P(\underline{\sigma}|\underline{J})
	= \frac{\prod_{i=1}^n P_0(\sigma_i) \prod_{A=1}^m Q(J_A | \sigma_A)}
	{\sum_{\underline{\sigma}\in \{-1, +1\}^n} \prod_{i=1}^n P_0(\sigma_i)\prod_{A=1}^m Q(J_A | \sigma_A)}.
\end{align*}
Dividing both the numerator and denominator by $\prod_{A=1}^m Q(J_A | \sigma_A=+1)$, the posterior $P(\underline{\sigma}|\underline{J})$ can be rewritten as
\begin{align}\label{eq:LDGM:Gibbs}
P(\underline{\sigma}|\underline{\tilde{J}}) = \frac{1}{\cal Z}  \exp \sum_{A=1}^{m} \tilde{J}_A (\sigma_A - 1) , 
\end{align}
where
\begin{align*}
\tilde{J}_A 
	& \equiv \frac{1}{2} \ln \frac{Q(J_A |  +1)}{Q(J_A | -1)}, \\
{\cal Z} 
	& \equiv \sum_{\underline{\sigma}\in \{-1, +1\}^n}\exp\sum_{A=1}^{m} \tilde{J}_A (\sigma_A - 1)
.
\end{align*}
We will use the language and notations of statistical mechanics. The normalization ${\cal Z}$ shall be called the partition function. The bipartite factor graph $\sG$ underlying \eqref{eq:LDGM:Gibbs} contains variable nodes $i = 1, \ldots, n$ and constraint (or factor) nodes $A = 1, \ldots, m$. 
Each variable node $i$ ``carries'' the binary variable $\sigma_i$ and each constraint node $A$ ``carries'' the half-log-likelihood ratio $\tilde{J}_A$ and uniformly 
connects to $K$ variable nodes $a_1, \ldots, a_K$. 
As said before, we identify $A\equiv\{a_1, \ldots, a_K\}$. Distribution \eqref{eq:LDGM:Gibbs} can be interpreted as the Gibbs distribution of a {\it random} spin system (or spin glass). The expectation of a quantity 
$A(\underline \sigma)$
with respect to the posterior \eqref{eq:LDGM:Gibbs} will 
be denoted by a Gibbs bracket
$$
\<A(\underline \sigma)\> \equiv \sum_{\sigma\in \{-1, +1\}^n} A(\underline \sigma) P(\underline{\sigma}|\underline{\tilde J})\,.
$$
The posteroior distribution as well as the expectations $\<A(\sigma)\>$  are random because of the randomness in: $i)$ the factor graph $\mathcal{G}$ ensemble; $ii)$ the observations $\underline{J}$ given the hidden vector
$\underline \sigma^0$; and $iii)$ the hidden vector $\underline \sigma^0$. 

It is equivalent to work in terms of observations $\underline{J}$ or associated 
half-log-likelihood ratios $\underline{\tilde J}$. The latter
are (formally) distributed according to
\begin{align}\label{induced}
 \prod_{i=1}^nP_0({\sigma}^0_i) \prod_{A=1}^m \sfc(\tilde J_A | \sigma_A^0) d\tilde J_A  \equiv \prod_{i=1}^nP_0({\sigma}^0_i) \prod_{A=1}^m Q(J_A | \sigma_A^0) dJ_A.
\end{align}
Most of the time it will be more convenient for us  to refer directly to half-log-likelihood ratios.
The graph, the observations and the hidden vector are called {\it quenched} random variables (r.v.) because {\it given instance} of the problem their realization is {\it fixed}.
In contrast the r.v. $\underline{\sigma}$ is sampled from the posterior \eqref{eq:LDGM:Gibbs}, and hence is often called an annealed variable. 
Expectations with respect 
to the quenched variables are denoted $\mathbb{E}_{\mathcal{G}}$ and 
$\mathbb{E}_{\underline \sigma^0}\mathbb{E}_{\underline{\tilde{J}}|\underline \sigma^0}$. To alleviate notations we shall often simply use $\mathbb{E}$ when the expectation 
is taken with respoect to {\it all} quenched r.v in the ensuing expression. The bracket $\< - \>$ is reserved for expectations with respect to the posterior \eqref{eq:LDGM:Gibbs}. 


%

Let $H(\underline\sigma | \underline{\tilde J})\equiv -\sum_{{\underline\sigma}\in \{-1, +1\}^n}
P(\underline\sigma |  \underline{\tilde J}) \ln P(\underline\sigma |  \underline{\tilde J})$ be the conditional entropy
of the hidden variables given fixed observations.
It is easy to see that the average conditional entropy (per variable) is given by the average {\it free entropy} (the r.h.s of the formula)
\begin{align}\label{cond-ent-free-energy}
\frac{1}{n}  \mathbb{E}_{\mathcal G}\mathbb{E}_{\underline \sigma^0}\mathbb{E} _{\underline{\tilde{J}}|\underline \sigma^0}
H(\underline{\sigma} | \underline{\tilde{J}}) = \frac{1}{n}  \mathbb{E}_{\mathcal G}\mathbb{E}_{\underline \sigma^0}\mathbb{E} _{\underline{\tilde{J}}|\underline \sigma^0} \ln {\cal Z}.
\end{align} 
We refer readers to 
\cite{KudM:2009} for details. The singularities, as a function of the measurement fraction $\alpha$, 
of this limiting quantity when $n\to +\infty$ give us the information theoretic thresholds, or the location of {\it static phase transitions} in physics language. 



\subsection{The replica symmetric formula for the average conditional entropy}

The cavity method \cite{MezardParisi87b} predicts that the asymptotic average conditional entropy per variable is accessible from the following ``replica symmetric'' functional. This functional is an ``average form''
of the Bethe free entropy expression. Details of the relationship between the replica symmetric functional and Bethe free entropy can be found 
in \cite[Appendix VII]{KumYMP:2014}.

\begin{definition}[The replica symmetric free entropy functional]
Let $V$ be a r.v. with distribution $\sfx$, and $V_i$, $i=1,\ldots, K$ i.i.d. copies of $V$. 
Let\footnote{Equation \eqref{BPfp} corresponds to one of the two {\it density evolution} fixed point equations associated 
with the belief propagation algorithm, see \cite{MezardMontanari09} for the links between this algorithm and the replica symmetric functional.} 
\begin{align}
 U & = \tanh^{-1}\Big(\tanh \tilde{J} \prod_{i=1}^{K-1} \tanh V_i\Big), \label{BPfp}
\end{align}
and $U_B$, $B=1,\ldots,l$ i.i.d. copies of $U$ where $l\sim\mathrm{Poi}(\alpha K)$ is a Poisson distributed integer. 
Let $\sigma^0\sim P_0$ and $\prod_{a=1}^K \sigma_a$ be the product of $K$ independent copies. Let $\tilde J\sim \sfc(\tilde J | \prod_{a=1}^K \sigma_a)$ (see equation \eqref{induced}). 
The replica symmetric free entropy functional is defined to be
\begin{align}
{h}_{\mathrm{RS}} ( \sfx )
	\equiv  \mathbb{E}_l \mathbb{E}_{\sigma_1,\cdots, \sigma_K}
	\mathbb{E}_{\tilde J |\prod_{a=1}^K \sigma_a} \mathbb{E}_{\underline{U}} \mathbb{E}_{\underline{V}}
	\Big [& \ln \Big ( \prod_{B=1}^{l} (1 + \tanh U_B ) + \prod_{B=1}^{l} (1 - \tanh U_B ) \Big ) 
	\nonumber \\ 
	&- \alpha(K-1) \ln \Big ( 1 + \tanh \tilde{J} \prod_{i=1}^{K}\tanh V_i \Big ) 
	 - \alpha \ln (1 + \tanh \tilde{J})\Big].
\label{eq:RSFreeEntropy1}
\end{align}
\end{definition}

\begin{remark}
 For uniform $P_0$ we can replace the product $\prod_{a=1}^K \sigma_a$ by a single binary variable $\sigma_0\sim P_0$.
\end{remark}



While a substantial part of our analysis holds for
general (symmetric) memoryless channels, our main result is fully proved for the BEC. This channel has 
transition probability
$$
Q(J_A\vert \sigma_A^0) = (1-q) \delta_{J_A, \sigma_A^0} + q\delta_{J_A, 0}\,,
$$
and from  \eqref{induced} we get in this case 
$$\sfc(\tilde J_A | \sigma_A^0) = (1-q)\delta_{ \sigma_A^0\tilde J_A, +\infty}
+ q \delta_{\tilde J_A, 0}\,.
$$

The set of distributions with point masses at $\{0, +\infty\}$ plays a special role and will be called $\mathcal{B}$. We adopt the 
notation (from coding theory) $\Delta_0$ and $\Delta_{\infty}$ for the two point masses at $0$ and $+\infty$. Any distribution 
$\mathsf{x}\in \mathcal{B}$ is of the from $\mathsf{x} = x\Delta_0 + (1-x) \Delta_{\infty}$, with $x\in [0, 1]$. In this case the replica symmetric free entropy functional becomes 
\eqref{eq:RSFreeEntropy1} becomes a function of $x\in [0,1]$.
A numerical illustration is found in Appendix~\ref{app:plot-BEC}.

Our main result is the proof, through the use of the adaptive interpolation method for sparse graphs, of the following theorem:

\begin{theorem}[The replica symmetric formula is exact for the BEC channel]\label{RSconj}
For a Censored Block Model with observations obtained through a Binary Erasure Channel as described above we have 
\begin{align}
 \lim_{n \rightarrow \infty} \frac{1}{n}  \mathbb{E}_{\mathcal G}\mathbb{E}_{\underline \sigma^0}\mathbb{E} _{\underline{\tilde{J}}|\underline \sigma^0}H(\underline{\sigma} | \underline{\tilde{J}}) = \sup_\mathsf{x \in \sB} \  h_{\mathrm{RS}}(\sfx).
\end{align}
\label{thm:RSconjecture-BEC}
\end{theorem}
This theorem is a direct consequence of two main Propositions \ref{prop:lowerbound} and \ref{prop:upperbound} proved in Sec. \ref{interp}. 
\section{Two preliminary tools}\label{prelimtool}
In this section we review standard material which is needed in our analysis. For more details the reader can consult  \cite{Macris2007,Mac:2007,richardsonurbanke2008}
\subsection{Nishimori identities}

\subsubsection{A consequence of Bayes rule} Consider the quantity $\prod_{S\in \mathcal{C}} \sigma_S^0 \prod_{S\in \mathcal{C}} \langle \sigma_S \rangle$ for a given graph and any collection $\mathcal{C}$ of subsets $S\subset \{1, \ldots, n\}$. The same subset can occur many times in a collection. From Bayes formula, for a given factor graph $\mathcal{G}$, 
\begin{align}
\mathbb{E}_{\underline \sigma^0}\mathbb{E}_{\underline{\tilde J} | \underline \sigma^0} \Big [ \prod_{S\in \mathcal{C}} \sigma_S^0 \prod_{S\in \mathcal{C}} \langle \sigma_S \rangle \Big ] 
& 
= \mathbb{E}_{\underline {\tilde J}}\mathbb{E}_{\underline\sigma^0 | \underline{\tilde J}}  \Big [\prod_{S\in \mathcal{C}} \sigma_S^0 \prod_{S\in \mathcal{C}} \langle \sigma_S \rangle \Big ]
\nonumber \\ &
= \mathbb{E}_{\underline {\tilde J}}\Big[\langle \prod_{S\in \mathcal{C}} \sigma_S \rangle \prod_{S\in \mathcal{C}} \langle \sigma_S \rangle \Big ]
\nonumber \\ &
= \mathbb{E}_{\underline \sigma^0} \mathbb{E}_{\underline{\tilde J}\vert \underline \sigma^0} 
\Big[\langle \prod_{S\in \mathcal{C}} \sigma_S \rangle \prod_{S\in \mathcal{C}} \langle \sigma_S \rangle \Big ].
\label{eq:tower}
\end{align}
The equality between the l.h.s and the last line on the r.h.s is a trivial but very important consequence of Bayes rule. This formula has been abusively called a ``Nishimori identity'' in the literature.
The ``true'' Nishimori identity is obtained when two extra features are present, namely a ``gauge invariance'' of the posterior and channel symmetry.

\subsubsection{Nishimori identities for symmetric channels} 
For symmetric channels this identity can be further specialized and yields the so-called Nishimori identities. 
This is specially important for us since the BEC is a symmetric channel. By definition, 
symmetric channels are those satisfying $Q(J_A | \sigma_A^0) = Q(- J_A | -\sigma_A^0)$ or equivalently 
$\sfc(\tilde J_A | \sigma_A^0) = \sfc(- \tilde J_A | - \sigma_A^0)$.

Given $\underline \sigma^0$ the Gibbs distribution \eqref{eq:LDGM:Gibbs} is {\it invariant} under the {\it gauge transformation} $\sigma_i \to \sigma_i^0\sigma_i$, $\tilde J_A \to \sigma_A^0 \tilde J_A$. Let us denote 
by $\underline{\sigma}^0\star\underline{\tilde J}$ the ``component-wise'' product $(\sigma_A^0 \tilde J_A)_{A=1}^m$. 
Now we perform a gauge transformation on both sides of \eqref{eq:tower}. For the left hand side we have
\begin{align}\label{intermediate}
\mathbb{E}_{\underline{\tilde J} | \underline \sigma^0} \Big [ \prod_{S\in \mathcal{C}} \sigma_S^0 \prod_{S\in \mathcal{C}} \langle \sigma_S \rangle \Big ]
=
\mathbb{E}_{\underline{\sigma}^0\star\underline{\tilde J} | \underline \sigma^0} \Big [ \prod_{S\in \mathcal{C}} \langle \sigma_S \rangle \Big ].
\end{align}
Moreover, from $\sfc(\tilde J_A | \sigma_A^0) = \sfc(- \tilde J_A | - \sigma_A^0)$ one can see that for a {\it symmetric channel} 
$\sfc(\sigma_A^0 \tilde J_A | \sigma_A^0) = \sfc(\tilde J_A | 1)$, and therefore in \eqref{intermediate} 
we can replace $\mathbb{E}_{\underline{\sigma}^0\star\underline{\tilde J} | \underline \sigma^0}$ by
$\mathbb{E}_{\underline{\tilde J} | \underline 1}$. We get
\begin{align}\label{lhs}
\mathbb{E}_{\underline{\tilde J} | \underline \sigma^0} \Big [ \prod_{S\in \mathcal{C}} \sigma_S^0 \prod_{S\in \mathcal{C}} \langle \sigma_S \rangle \Big ]
=
\mathbb{E}_{\underline{\tilde J} | \underline 1} \Big [\prod_{S\in \mathcal{C}} \langle \sigma_S \rangle \Big ].
\end{align}
The same steps show that the right hand side of \eqref{eq:tower} also satisfies
\begin{align}\label{rhs}
\mathbb{E}_{\underline {\tilde J} | \underline\sigma^0}\Big [ \langle \prod_{S\in \mathcal{C}} \sigma_S \big \rangle \prod_{S\in \mathcal{C}} \langle \sigma_S \rangle \Big ]
=
\mathbb{E}_{\underline {\tilde J} | \underline 1}\Big [ \big \langle \prod_{S\in \mathcal{C}} \sigma_S \big \rangle \prod_{S\in \mathcal{C}} \langle \sigma_S \rangle \Big ].
\end{align}
From \eqref{lhs}, \eqref{rhs}, \eqref{eq:tower} we get the final Nishimori identity
\begin{align}
\mathbb{E}_{\underline {\tilde J}|\underline 1}
\Big [\prod_{S\in \mathcal{C}} \langle \sigma_S \rangle \Big ] 
= \mathbb{E}_{\underline {\tilde J}|\underline 1}\Big [ \big \langle \prod_{S\in \mathcal{C}} \sigma_S \big \rangle \prod_{S\in \mathcal{C}} \langle \sigma_S \rangle \Big ].
\label{eq:Nishimori}
\end{align}

\subsubsection{A special Nishimori identity for symmetric distributions}

An important role is played by the space $\mathcal{X}$ of {\it symmetric distributions} which we define as follows. Take a transition probability (a ``channel'') satisfying 
 $q(J | \sigma^0) = q(-J | -\sigma^0)$, $\sigma^0\in \{-1, +1\}$, $J\in \mathbb{R}$. The associated half-log-likelihood variable is 
$h = \frac{1}{2}\ln \frac{q(J | +1)}{q(J | -1)}$. The space $\mathcal{X}$ is the space of symmetric distributions over the half-log-likelihood variable is formally defined 
by $\mathsf{x}(dh) = q(J| +1) dJ$.  It is easy to deduce from $q(J | \sigma^0) = q(-J | -\sigma^0)$
 that a symmetric distribution $\mathsf{x}\in \mathcal{X}$ satisfies $\mathsf{x}(-dh) = e^{-2h} \mathsf{x}(dh)$.
 We note that $\mathcal{B}\subset \mathcal{X}$ (recall $\mathcal{B}$ is the set of convex combinations of point masses at $0$ and $+\infty$).

There is an important special case of the Nishimori identity \eqref{eq:Nishimori}. Namely the one satisfied by the system constituted by a single uniform hidden variable $\sigma^0\sim P_0$, 
observed through a noisy ``channel'' $q(J\vert \sigma^0)$.  The Gibbs distribution is simply in this case
$e^{h \sigma}/(2\cosh h)$ where $h$ is the half-log-likelihood of the ``channel''. 
Since $\langle \sigma \rangle = \tanh h$, an 
application of \eqref{eq:Nishimori} (where the singleton set is taken 
$k$ times) yields
\begin{align}
\int (\tanh h)^{2k-1} \,\sfx(dh) = \int (\tanh h)^{2k} \,\sfx(dh), \qquad k\in \mathbb{N}^*.
\label{eq:channelSymmetry}
\end{align}
In Appendix~\ref{proof:equi_channel_symmetry} we show in an independent and direct way that any $\mathsf{x}\in \mathcal{X}$ satisfies
\eqref{eq:channelSymmetry}.

\subsubsection{Conditional entropy for symmetric channels} 
Since the Gibbs distribution is invariant under a gauge transformation, the partition function 
${\cal Z}$ also is, and therefore, for a given graph $\mathcal{G}$ and hidden vector $\underline \sigma^0$, we have 
\begin{align}
\mathbb{E}_{\underline{\tilde J} | \underline\sigma^0} \ln {\cal Z}
& 
= \mathbb{E}_{\underline{\sigma}^0\star\underline{\tilde J} | \underline\sigma^0} \ln {\cal Z} .
\end{align}
For symmetric channels the r.h.s. equals $\mathbb{E}_{\underline{\tilde J} | \underline 1} \ln {\cal Z}$ and thus the average conditional entropy \eqref{cond-ent-free-energy} becomes 
\begin{align}\label{cond-ent-free-energy-sym}
\frac{1}{n} \mathbb{E}_{\mathcal{G}}\mathbb{E}_{\underline \sigma^0}\mathbb{E}_{\underline{\tilde J} | \underline \sigma^0}
H(\underline{\sigma} | \underline{\tilde{J}}) 
=  \frac{1}{n} \mathbb{E}_{\mathcal{G}} \mathbb{E}_{\underline{\tilde J} | \underline 1}\ln {\cal Z}.
\end{align} 

\subsubsection{Summary} 
When one is dealing with symmetric measurement channels, in order to compute the 
average conditional entropy, or {\it certain} averages,  one may assume that $\sigma_i^0 =1$, $i=1, \ldots, n$ and 
that the quenched variables $\underline{\tilde J}$ have distribution $\sfc(\tilde J_A | 1)$, $A=1, \ldots, m$. From now on this is 
understood unless explicitly specified otherwise. 


\subsection{Griffiths-Kelly-Sherman inequalities for the BEC}

The BEC is a symmetric channel so as shown before, without loss of generality for analysis purposes, we assume $\sigma_i^0 =1$, $i=1, \ldots, n$ and that
$\underline{\tilde J}$ have distribution $\sfc(\tilde J_A | 1)$, $A=1, \ldots, m$.
Since $\sfc(\tilde J_A | 1) = (1-q) \Delta_{\infty} + q\Delta_0$ the Gibbs 
distribution \eqref{eq:LDGM:Gibbs} has non-negative coupling constants $\tilde J_A$, $A=1,\ldots,m$. Therefore 
the Gibbs distribution satisfies the Griffiths-Kelly-Sherman (GKS) inequalities \cite{20017,Macris2007,Mac:2007}: For any subsets of variable indices $S, T \subset \{1 \ldots n\}$ we have
\begin{align}
& \< \sigma_S \> \geq 0, \label{eq:GKS1} \\ &
\< \sigma_S \sigma_T \> - \< \sigma_S \> \< \sigma_T \> \geq 0.\label{eq:GKS2}
\end{align}
These two inequalities play an important role in the
proof of Theorem~\ref{thm:RSconjecture-BEC}.
\section{The adaptive path interpolation method}\label{interp}
For $t=1,\dots, T$ let $V^{(t)}_i$ be i.i.d. r.v. distributed according to $\sfx^{(t)}\in \mathcal{X}$.
Consider the r.v.
\begin{align}
U^{(t)} = \tanh^{-1} \Big( \tanh \tilde{J} \prod_{i=1}^{K-1} V_i^{(t)} \Big ) \label{eq:BP-update}
\end{align}
and independent copies denoted $U^{(t)}_B$ where $B$ is a subscript which runs over $l^{(t)} \sim \mathrm{Poi}(\frac{K}{RT})$ of these copies. 
Later on, we call $\tilde \sfx^{(t)}$ the distribution of $U^{(t)}$ (induced by $\sfx^{(t)}$ and $\sfc$).

Let also define two extra random variables, $H$ with distribution $\epsilon \Delta_{\infty} + (1-\epsilon) \Delta_{0} \in \mathcal{B}$, and $\tilde{H}$ with 
distribution $\delta n^{-\theta} \Delta_{\infty} + (1-\delta  n^{-\theta}) \Delta_0 \in \mathcal{B}$, where $\epsilon, \delta \in (0,1)$ and $\theta \in (0, 1)$ (eventually we will have to take $\theta\in (0, 1/5]$ in the final estimates). 

Let us set $\underline \sfx = (\sfx^{(1)}, \dots, \sfx^{(T)})$.
We define the {\it generalized free entropy functional}:
\begin{align}
\tilde{h}_{\epsilon,\delta}( \underline{\sfx} ) 
	& \equiv\mathbb{E} \Big[\ln \Big( \prod_{t=1}^{T} \prod_{B=1}^{l^{(t)}} (1 + \tanh U_B^{(t)} ) + e^{-2(H+\tilde{H})} \prod_{t=1}^{T} \prod_{B=1}^{l^{(t)}} (1 - \tanh U_B^{(t)} ) \Big ) \nonumber \\
	& \hspace{2cm} - \frac{\alpha(K-1)}{T} \sum_{t=1}^{T} \ln \Big (1 + \tanh \tilde{J} \prod_{i=1}^{K} \tanh V_i^{(t)} \Big ) 
	-\alpha \ln (1 + \tanh \tilde{J}) \Big].
\label{eq:RSFreeEntropy_coupled1}
\end{align}
One can easily check that if $\sfx^{(t)} = \sfx$ for all $t$, then $\tilde{h}_{\epsilon = 0, \delta = 0}(\underline{\sfx}) = h_{\mathrm{RS}}(\sfx)$. 
More is true as the following lemma shows:
\begin{lemma}\label{equivalence}
Let $\sX^T = \mathcal{X}\times \mathcal{X}\times \ldots \times \mathcal{X}$. We have for $\underline{\sfx} \in \mathcal{X}^T$
\begin{align}
\sup_{\underline{\sfx} \in \sX^T} \tilde{h}_{\epsilon = 0, \delta = 0}(\underline{\sfx}) = \sup_{\sfx \in \sX} h_{\mathrm{RS}}(\sfx).
\label{eq:h_RS_coupled:equivalence}
\end{align}
\end{lemma}
\begin{remark}
 We prove this lemma in Sec.~\ref{proof:h_RS_coupled:equivalence}. For distributions in $\underline \sfx\in \mathcal{B}^T$ the supremum carries over $(x^{(1)}, \cdots, x^{(T)})\in[0, 1]^T$ and the proof only requires real analysis.  
\end{remark}

%


\subsection{The $(t,s)$--interpolating model}\label{subsec:tsmodel}
Consider the construction of an interpolating factor graph ensemble $\sG_{t,s}$ involving discrete and a continuous interpolation parameters, $t \in \{1, 2, \dots, T \}$ and $s \in [0,1]$. 
This is the sparse graph counterpart of the interpolating ensemble initialy developed for dense graphs in \cite{BarbierM17a} (and the simplified in \cite{BarbierMacrisCW2019}).
%
\begin{algorithm}
\caption{Construction of $\sG_{t,s}$}
\begin{algorithmic}
\label{algo:G(t,s)}
	\For {$i = 1, \dots, n$}
	\For {$t' = 1, \dots, t-1$}
		\State {\_draw a random number $e_i^{(t')} \sim \mathrm{Poi}\big( \frac{\alpha K}{T} \big)$}
		\For {$B = 1, \dots, e_i^{(t')}$}
			\State{\_connect variable node $i$ with a half edge and assign a weight $U_{B \rightarrow i}^{(t')} \sim \tilde{\sfx}^{(t')}$ to this half-edge}
		\EndFor
	\EndFor
	\State {\_draw a random number $e_{i,s}^{(t)} \sim \mathrm{Poi}\big ( \frac{\alpha Ks}{T} \big )$}
	\For {$C = 1, \dots, e_{i,s}^{(t)}$}
		\State{\_connect variable node $i$ with a half edge and assign a weight $U_{C \rightarrow i}^{(t)} \sim \tilde{\sfx}^{(t)}$ to this half-edge}
	\EndFor
\EndFor\\
%
%
%
\State {\_draw a random number $m_s^{(t)} \sim \mathrm{Poi}\big ( \frac{\alpha n(T-t+1-s)}{T} \big )$}
\For {$A = 1, \dots, m_s^{(t)}$}
		\State {\_assign to factor node $A$ a r.v $\tilde{J}_{A} \sim \sfc$}
		\State {\_uniformly and randomly connect factor node $A$ to $K$ variable nodes (this subset of variable nodes is also denoted $A$ by a slight abuse of notation)}
	\EndFor
\end{algorithmic}
\end{algorithm}

The interpolating graph is designed such that $\sG_{t,1}$ is statistically equivalent to $\sG_{t+1,0}$; 
in addition, $\sG_{t,s}$ maintains the degree distribution of variable nodes invariant: For any $(t,s)$ the degree of each variable node is an independent ${\rm Poi}(K/R)$ random variable. The Hamiltonian associated with $\sG_{t,s}$ is 
%
\begin{align}
\sH_{t,s}(\underline{\sigma}, \underline{\tilde{J}}, \underline{U}, \underline{m}, \underline{e}) 
	= &
	-  \sum_{i=1}^{n}\Big\{\sum_{t' =1}^{t-1} \sum_{B=1}^{e_{i}^{(t')}} U_{B \rightarrow i}^{(t')} + \sum_{C=1}^{e_{i,s}^{(t)}} U_{C \rightarrow i}^{(t)}\Big\} ( \sigma_i - 1 )
	- \sum_{A=1}^{m_s^{(t)}} \tilde{J}_{A} ( \sigma_{A} - 1 ).
\label{eq:Hamiltonian_ts}
\end{align}
We further consider a generalized version of \eqref{eq:Hamiltonian_ts} by adding two kinds of perturbations that can be interpreted as small additional observations from side-channels for each node $i=1,\cdots, n$. 
These perturbations are then removed at the end of the analysis. let $H_i$ and $\tilde{H}_i$ be half-log-likelihood variables, where $H_i$ and $\tilde{H}_i$ have the same distribution as $H$ and $\tilde{H}$ 
defined at the beginning of this section. Our final {\it interpolating Hamiltonian} is
\begin{align}
\sH_{t,s; \epsilon, \delta}(\underline{\sigma}, \underline{\tilde{J}}, \underline{U}, \underline{m}, \underline{e}, \underline{H}, \underline{\tilde{H}}) 
	& \equiv \sH_{t,s}(\underline{\sigma}, \underline{\tilde{J}}, \underline{U}, \underline{m}, \underline{e}) - \sum_{i=1}^{n} (H_i + \tilde{H}_i) ( \sigma_i - 1 ).
	\label{eq:Hamiltonian_tse}
\end{align}
The associated interpolating partition function, Gibbs expectation and free entropy are:
\begin{align}
{\cal Z}_{t,s; \epsilon,\delta}
	& \equiv \sum_{\underline{\sigma}\in\{-1,+1\}^n} e^{-{\sH_{t,s; \epsilon,\delta}(\underline{\sigma}, \underline{\tilde{J}}, \underline{U}, \underline{m}, \underline{e}, \underline{H}, \underline{\tilde{H}})}},
	\label{eq:partitionFunction_tse} \\
\< A(\underline{\sigma}) \>_{t,s; \epsilon, \delta} 
	& \equiv \frac{1}{{\cal Z}_{t,s; \epsilon,\delta}}\sum_{\underline{\sigma}\in\{-1,+1\}^n} A(\underline{\sigma})\, 
	e^{-{\sH_{t,s; \epsilon,\delta}(\underline{\sigma}, \underline{\tilde{J}}, \underline{U}, \underline{m}, \underline{e}, \underline{H}, \underline{\tilde{H}})}}, 
	\label{eq:GibbsExp_tse} \\
h_{t,s; \epsilon, \delta} & \equiv \frac{1}{n} \mathbb{E}\ln {\cal Z}_{t,s; \epsilon,\delta}, \label{eq:htse} \\
H_{t,s; \epsilon,\delta} & \equiv \frac{1}{n} \mathbb{E}_{\underline{\tilde{H}}} \ln \mathcal{Z}_{t,s; \epsilon,\delta}\,. \label{eq:Htse}
\end{align}
Recall our notation: The expectation $\mathbb{E}$ here carries over {\it all} quenched variables entering in the interpolating system, thus $\underline{\tilde{J}}, \underline{U}, \underline{m}, \underline{e}$, $\underline{H}$ and $\underline{\tilde{H}}$.
Note that Nishimori's identity \eqref{eq:Nishimori} and GKS inequalities \eqref{eq:GKS1}, \eqref{eq:GKS2} still apply to the Gibbs 
expectation $\< - \>_{t,s; \epsilon,\delta}$. 

One may check that, at the initial point of the interpolating path, $t=1$, $s=0$ the free entropy $h_{1,0;\epsilon=0,\delta=0}$ is equal to the averaged conditional entropy of the original model
(see formula \eqref{h10ish} below), and at the end-point $t=T$, $s=1$ the free entropy $h_{T,1;\epsilon,\delta}$ is given by a part
of the generalized entropy functional \eqref{eq:RSFreeEntropy_coupled1} (see formula \eqref{eq:h_T1}).

The connection between the unperturbed and perturbed free entropies is given by (see Sec.~\ref{sec:proofepsilon-free})
\begin{lemma}\label{pert-free}
 Let $\sfc\in \mathcal{B}$ and $\underline \sfx\in \mathcal{B}^T$. We have
 \begin{align}
  & \vert h_{t, s; \epsilon,\delta} - h_{t, s; \epsilon=0,\delta=0} \vert \leq  ( \epsilon + \frac{\delta}{n^{\theta}} ) \ln 2\, ,\label{eq:pert-free:1} \\
  & \vert \tilde{h}_{\epsilon,\delta}(\underline{\sfx}) - \tilde{h}_{\epsilon=0,\delta=0}(\underline{\sfx}) \vert \leq  ( \epsilon + \frac{\delta}{n^{\theta}} ) \ln 2\, . \label{eq:pert-free:2}
  \end{align}
\end{lemma}
%
\subsection{Evaluating the free entropy change along the $(t,s)$--interpolation}
By interpolating $h_{t,s; \epsilon, \delta}$ from the initial state $(t=1, s=0)$ to the final one $(t=T, s=1)$, we have
\begin{align}
h_{1,0; \epsilon, \delta} = h_{T,1; \epsilon, \delta} + \sum_{t=1}^{T} ( h_{t,0; \epsilon, \delta} - h_{t,1; \epsilon, \delta}) = h_{T,1; \epsilon, \delta} - \sum_{t=1}^T \int_0^1 ds \frac{dh_{t,s; \epsilon, \delta}}{ds}.
\label{eq:interpolation:main1}
\end{align}
We have $m_{s=0}^{(t=1)}=m\sim{\rm Poi}(\alpha n)$ and thus
\begin{align*}
\sH_{1,0; \epsilon, \delta} =- \sum_{A=1}^{m} \tilde{J}_{A} (\sigma_{A}-1) - \sum_{i=1}^{n} (H_i + \tilde{H}_i) (\sigma_i - 1) .
\end{align*}
Therefore the initial interpolating free entropy without perturbation equals the average conditional entropy per variable:
\begin{align}
h_{1,0; \epsilon=0, \delta = 0}=\frac{1}{n}\mathbb{E} H(\underline \sigma | \underline{\tilde J}). \label{h10ish}	
\end{align}
On the other hand $h_{T,1; \epsilon, \delta}$ corresponds to a part of the generalized free entropy functional \eqref{eq:RSFreeEntropy_coupled1}. 
A subsequent computation (see Sec.~\ref{proof:interpolation:main2}) on \eqref{eq:interpolation:main1} leads to the {\it fundamental sum rule}
\begin{align}
h_{1,0; \epsilon, \delta} = \tilde{h}_{\epsilon, \delta}(\underline{\sfx}) + \frac{\alpha }{T} \sum_{t=1}^T \int_0^1 ds \,\sR_{t,s; \epsilon, \delta}
\label{eq:interpolation:main2}
\end{align}
where 
\begin{align}
\sR_{t,s; \epsilon, \delta} = \sum_{p=1}^{\infty} \frac{\mathbb{E} [ (\tanh \tilde{J})^{2p} ]}{2p(2p-1)} \ \mathbb{E} \big\langle Q_{2p}^{K} - K  (q_{2p}^{(t)})^{K-1} ( Q_{2p} - q_{2p}^{(t)}) - 
(q_{2p}^{(t)})^K \big\rangle_{t,s; \epsilon, \delta} 
\label{eq:interpolation:remainder}
\end{align}
with $$Q_{p} \equiv \frac{1}{n} \sum_{i=1}^{n} \sigma_i^{(1)} \cdots \sigma_i^{(p)}$$ the {\it overlap} 
of $p$ independent {\it replicas} $\underline{\sigma}^{(1)}, \dots , \underline{\sigma}^{(p)}$ 
and 
$$
q_p^{(t)} \equiv \mathbb{E} [(\tanh V^{(t)})^p].
$$
In \eqref{eq:interpolation:remainder} the Gibbs average $\langle - \rangle_{t, s;\epsilon, \delta}$ over a polynomial of $Q_p$ must be understood as an
average over the product measure
\begin{align*}
 \prod_{\alpha=1}^p \frac{1}{{\cal Z}_{t,s; \epsilon,\delta}}
	e^{-{\sH_{t,s; \epsilon,\delta}(\underline{\sigma}^{(\alpha)}, \underline{\tilde{J}}, \underline{U}, \underline{m}, \underline{e}, \underline{H}, \underline{\tilde{H}})}}
\end{align*}
where the quenched variables have the same realization for all replicas. We still denote this Gibbs average by $\langle - \rangle_{t, s;\epsilon, \delta}$ for simplicity.

\subsection{Lower bound}
In order to show the lower bound we need the following important concentration lemma (proven in Sec.~\ref{sec:overlap}), which is at the core of the ``replica symmetric'' behavior of the model:
\begin{lemma}[Concentration of $Q_p^K$ on $\< Q_p \>_{t,s; \epsilon; \theta}^K$]
For any $\sfc \in \sB$, $\underline\sfx\in \mathcal{B}^T$ we have
\begin{align}
\int_{\varepsilon_0}^{\varepsilon_1} d\epsilon \ \mathbb{E}  \big\< \big| Q_p^{K} - \< Q_p \>_{t,s; \epsilon,\delta}^K \big| \big\>_{t,s; \epsilon,\delta}  \leq K \Big ( \frac{3p(\varepsilon_1-\varepsilon_0)}{n} \Big )^{1/2}.
\end{align}
\label{thm:concentration1}
\end{lemma}

\begin{proposition}[Lower bound]\label{prop:lowerbound}
For $\sfc \in \mathcal{B}$ we have
\begin{align}\label{lower-bound-guerra-style}
\liminf_{n \rightarrow \infty} \frac{1}{n} \mathbb{E}H(\underline{\sigma} | \underline{\tilde{J}}) 
\geq \sup_{\sfx \in \mathcal{B}} h_{\mathrm{RS}}(\sfx).
\end{align}
\end{proposition}

\begin{remark}
 The methods of this paper can be extended to show this proposition for $\mathrm{c}, \sfx\in \mathcal{X}$.
\end{remark}

\begin{proof}
Eq.~\eqref{eq:interpolation:main2} implies
\begin{align}
h_{1,0;\epsilon=0, \delta=0} 
	& = \tilde{h}_{\epsilon=0,\delta=0}(\underline{\sfx}) + \frac{\alpha }{T} \sum_{t=1}^{T} \int_0^1 ds \sR_{t,s;\epsilon,\delta}(\underline{\sfx}) \nonumber \\
	& + \big ( \tilde{h}_{\epsilon,\delta}(\underline{\sfx}) - \tilde{h}_{\epsilon=0,\delta=0}(\underline{\sfx}) \big ) 
	- \big ( h_{1,0;\epsilon,\delta} - h_{1,0;\epsilon=0,\delta=0} \big ). \label{eq:bd-start}
\end{align}
We fix $\delta=0$. From \eqref{eq:interpolation:remainder} we have $\sR_{t,s;\epsilon,\delta=0}$ equal to
\begin{align*}
\sum_{p=1}^{\infty} \frac{\mathbb{E}[ (\tanh \tilde{J})^{2p} ] }{2p(2p-1)} 
\Big ( \mathbb{E}
\Big[ \langle Q_{2p} \rangle_{t,s; \epsilon,0}^{K} - K (q_{2p}^{(t)})^{K-1} (\langle Q_{2p} \rangle_{t,s; \epsilon,0} - 
q_{2p}^{(t)}) - (q_{2p}^{(t)})^K \Big ] - \mathbb{E} \< ( Q_p^K - \< Q_p \>_{t,s;\epsilon,0}^K ) \>_{t,s;\epsilon,0} \Big ).
\end{align*}
 Note that the convexity $x\mapsto x^K$ for $x \in \mathbb{R}_{+}$ implies 
$x^K - y^K \geq K y^{K-1} (x-y)\geq 0$ for any $x,y\in \mathbb{R}_+$. As $\<Q_{2p}\>_{t,s; \epsilon,0}=\frac{1}{n} \sum_{i=1}^n \<\sigma_i\>_{t,s; \epsilon,0}^{2p} \geq 0$ and $q_{2p}^{(t)} \geq 0$,
\begin{align*}
 \langle Q_{2p} \rangle_{t,s; \epsilon,0}^{K} - K (q_{2p}^{(t)})^{K-1} (\langle Q_{2p} \rangle_{t,s; \epsilon,0} - 
q_{2p}^{(t)}) - (q_{2p}^{(t)})^K \geq 0.
\end{align*}
Thus with Lemma~\ref{thm:concentration1} we obtain
\begin{align}
\frac{1}{\epsilon_n} \int_{\epsilon_n}^{2\epsilon_n} d\epsilon\, \sR_{t,s;\epsilon,0} \geq - (\ln 2) K \Big ( \frac{3p}{\epsilon_n n} \Big )^{1/2}.
\label{eq:almost-positivity}
\end{align}
Now we average both side of \eqref{eq:bd-start} over $\epsilon \in [\epsilon_n, 2 \epsilon_n]$ for some sequence $\epsilon_n$ specified at the end. Using \eqref{eq:almost-positivity} and Lemma~\ref{pert-free}
\begin{align*}
h_{1,0;\epsilon=0,\delta=0} \geq \tilde{h}_{\epsilon=0,\delta=0}(\underline{\sfx}) + \mathcal{O} \big ( \frac{1}{\sqrt{\epsilon_n n}} \big ) + \mathcal{O}(\epsilon_n).
\end{align*}
Choosing $\epsilon_n = n^{-\gamma}$ with $0<\gamma<1$, we conclude
\begin{align*}
\liminf_{n \rightarrow \infty} h_{1,0;\epsilon=0,\delta=0} \geq \tilde{h}_{\epsilon=0,\delta=0}(\underline{\sfx}).
\end{align*}
Finally one can take the supremum of the right hand side and use \eqref{eq:h_RS_coupled:equivalence} as well as \eqref{h10ish} to obtain \eqref{lower-bound-guerra-style}.
\end{proof}

\subsection{Upper bound}
In this paragraph we crucially use the specificities of the BEC. We take interpolating paths $\underline{\sfx}=(\sfx^{(1)},\cdots, \sfx^{(T)})\in \mathcal{B}^T$, where 
$\sfx^{(t)} = x^{(t)}\Delta_0 + (1- x^{(t)})\Delta_{\infty}$.
In particular we use the following lemma (proven in Sec.~\ref{sec:proof:BEC:overlap}):
\begin{lemma}
For any $\sfc \in \sB$, and any interpolation path $\underline\sfx(\epsilon) \in \mathcal{B}^T$ depending on $\epsilon$, and any $A \subseteq \{ 1, \dots, n \}$ we have $\< \sigma_A \>_{t,s; \epsilon,\delta}\in\{0,1\}$.
\label{thm:BEC:overlap}
\end{lemma}

Notice that $\< Q_p^k \>_{t,s; \epsilon,\delta} = \frac{1}{n^k}\sum_{i_1,\dots,i_k=1}^n \langle\sigma_{i_1} \cdots \sigma_{i_k}\rangle_{t,s; \epsilon,\delta}^p$ for any $k \in \mathbb{N}$. 
Lemma~\ref{thm:BEC:overlap} then implies 
$$
\< Q_p^k \>_{t,s; \epsilon,\delta} = \< Q_1^k \>_{t,s; \epsilon,\delta}
$$
for all $p \in \mathbb{N}^{*}$. We also have $\tanh V^{(t)}\in\{0,1\}$ because $\sfx^{(t)}\in \mathcal{B}$ and thus $$q_p^{(t)}= q_1^{(t)} \ \forall \ p \in \mathbb{N}^{*}.$$ 
Finally recall that $\sfc(\tilde J|1) = (1-q)\Delta_\infty + q \Delta_0$, therefore $\mathbb{E}[(\tanh \tilde{J})^{2p}] = 1- q$. 
These facts reduce \eqref{eq:interpolation:remainder} to
\begin{align}
\sR_{t,s; \epsilon, \delta} = (1-q)(\ln 2) \mathbb{E}
\Big [ \langle Q_{1}^K \rangle_{t,s; \epsilon, \delta} - K (q_{1}^{(t)})^{K-1} (\langle Q_{1} \rangle_{t,s; \epsilon, \delta} - q_{1}^{(t)}) - 
(q_{1}^{(t)})^K \Big ].
\label{eq:interpolation:remainder:simplify}
\end{align}
We then split the remainder as follows:
\begin{align}
\sR_{t,s; \epsilon, \delta} 
& = (\sR_{t,s; \epsilon, \delta} - \sR_{t,0; \epsilon, \delta}) + \sR_{t,0; \epsilon, \delta} \nonumber \\
& = (\sR_{t,s; \epsilon, \delta} - \sR_{t,0; \epsilon, \delta})
+ (1-q)(\ln 2) \Big ( \mathbb{E} [ \< Q_{1} \>_{t,0; \epsilon, 0} ]^K - K (q_{1}^{(t)})^{K-1} (\mathbb{E} \< Q_{1} \>_{t,0; \epsilon, 0} - q_{1}^{(t)}) - 
(q_{1}^{(t)})^K \Big ) \nonumber \\
& + (1-q) (\ln 2) \big ( \mathbb{E} \<Q_1^K\>_{t,0;\epsilon,\delta} - \mathbb{E}[\<Q_1\>_{t,0;\epsilon,\delta}]^K \big ) \nonumber \\
& + (1-q) (\ln 2) \big ( \mathbb{E}[\<Q_1\>_{t,0;\epsilon,\delta}]^K - \mathbb{E}[\<Q_1\>_{t,0;\epsilon,0}]^K
- K (q_1^{(t)})^{K-1} ( \mathbb{E}[\<Q_1\>_{t,0;\epsilon,\delta}] - \mathbb{E}[\<Q_1\>_{t,0;\epsilon,0}] ) \big )
\label{eq:interpolation:remainder:decompose}
\end{align}
and treat each part thanks to the three following lemmas.
Lemma \ref{thm:bd_fixed_t} is proven in Sec.~\ref{sec:proof:bd_fixed_t}, lemma \ref{thm:concentration2} in Sec.~\ref{sec:overlap}, 
and lemma~\ref{thm:rmv-delta} in Sec.~\ref{sec:proof:rmv-delta}).

\begin{lemma}[Weak $s$-dependence at fixed $t$]\label{thm:bd_fixed_t}
For any $k \in \mathbb{N}$ and $s\in [0,1]$ we have
\begin{align}
\big| \mathbb{E}\< Q_1^k \>_{t,s; \epsilon,\delta}  - \mathbb{E} \< Q_1^k \>_{t,0; \epsilon,\delta}  \big| \leq  \frac{2\alpha (K+1)n}{T}.
\end{align}
\end{lemma}
\begin{lemma}[Concentration of $\< Q_1 \>_{t,s; \epsilon,\delta}^K$ on $\mathbb{E}{[}\langle Q_1 \rangle_{t,s; \epsilon,\delta}{]}^K$ ]\label{thm:concentration2}

For any $\sfc \in \sB$ and $\underline\sfx(\epsilon) \in \mathcal{B}^T$ such that every component satisfies $dx^{(t)}/d\epsilon \geq 0$ we have for $\theta\in (0, 1/5]$
\begin{align}
\int_{\delta_0}^{\delta_1} d\delta \int_{\varepsilon_0}^{\varepsilon_1} d\epsilon \mathbb{E}\big[ \big| \< Q_1 \>_{t,s; \epsilon,\delta}^{K}  - \mathbb{E} [\< Q_1 \>_{t,s; \epsilon,\delta}]^K \big| \big ] 
	= \mathcal{O} \bigg ( \bigg ( (\delta_1-\delta_0)(\varepsilon_1-\varepsilon_0) \frac{\delta_1-\delta_0 + \varepsilon_1-\varepsilon_0}{n^{\theta}} \bigg )^{1/2} \bigg ).
\label{eq:concentration2}
\end{align}
\end{lemma}
\begin{lemma}
For any $\sfc \in \sB$ and $\underline\sfx(\epsilon) \in \mathcal{B}^T$ such that every component satisfies $dx^{(t)}/d\epsilon \geq 0$ we have
\begin{align}
& \bigg \vert \int_{\delta_0}^{\delta_1} d\delta \int_{\varepsilon_0}^{\varepsilon_1} d\epsilon \Big \{ \mathbb{E}[\<Q_1\>_{t,0;\epsilon,\delta}]^K - \mathbb{E}[\<Q_1\>_{t,0;\epsilon,0}]^K \nonumber \\
& \hspace{1cm} - K (q_1^{(t)})^{K-1} ( \mathbb{E}[\<Q_1\>_{t,0;\epsilon,\delta}] - \mathbb{E}[\<Q_1\>_{t,0;\epsilon,0}] ) \Big \} \bigg \vert \leq \frac{3K(\delta_1^2-\delta_0^2)}{n^{\theta}(1-\delta_1/n^{\theta})}\,.
\end{align}
where $\theta\in (0, 1/5]$.
\label{thm:rmv-delta}
\end{lemma}

Now we look into each term of \eqref{eq:interpolation:remainder:decompose}. Lemma~\ref{thm:bd_fixed_t} and \eqref{eq:interpolation:remainder:simplify} imply 
$$
|\sR_{t,s; \epsilon, \delta} - \sR_{t,0; \epsilon, \delta}| \leq \frac{2 (\ln 2) \alpha (K+1)^2 n}{T} = \mathcal{O}(\frac{n}{T}).
$$
Since $T$ is a free parameter (controlling the mean of $e_{i,s}^{(t)}$ and $m_s^{(t)}$) we can set it 
significantly larger than $n$. The first term of \eqref{eq:interpolation:remainder:decompose} thus can be neglected and it is sufficient to work with $\sR_{t,0; \epsilon, \delta}$. 
This separation is important because we use that $\mathbb{E} \langle Q_{1} \rangle_{t,0; \epsilon,0}$ 
(in the second term of \eqref{eq:interpolation:remainder:decompose}) is independent of $\{ \sfx^{(t')} \}_{t' \geq t}$. Also recall $q_1^{(t)}\equiv \mathbb{E} \tanh V^{(t)}$. 
This allows us to sequentially choose a distribution 
$\hat{\sfx}_n^{(t)}$ for $V^{(t)}$ along our interpolation from $t=1$ to $T$ such that the following equation is satisfied:
\begin{align}
q_1^{(t)} = \mathbb{E} \langle Q_1 \rangle_{t,0; \epsilon, 0} .
\label{eq:momentMatching}
\end{align}
In other words the interpolation path is {\it adapted} so that \eqref{eq:momentMatching} holds, which then cancels the second term in \eqref{eq:interpolation:remainder:decompose}. This 
path is also independent of $\delta$ because we have set $\delta=0$ in the Gibbs expectation \eqref{eq:momentMatching} as well as in the second term of \eqref{eq:interpolation:remainder:decompose}. 
We must still check that equation \eqref{eq:momentMatching} possesses a (unique) solution, see Sec.~\ref{sec:proof:momentMatching:solution} for the proof:
\begin{lemma}[Existence of the optimal interpolation path]\label{thm:momentMatching:solution}
Eq. \eqref{eq:momentMatching} has 
a unique solution $\underline{\hat{\sfx}}_n(\epsilon) = \{ \hat{\sfx}^{(t)}_n \}_{t=1}^T \in \sB^T$. The solution $\hat{\sfx}_n \equiv \hat{x}_n^{(t)} \Delta_\infty + (1-\hat{x}_n^{(t)}) \Delta_0$ satisfies $d \hat{x}_n^{(t)}/d\epsilon \geq 0$.
\end{lemma}

Fixing $\underline{\sfx} = \underline{\hat{\sfx}}_n(\epsilon)$, lemmas~\ref{thm:concentration2} and \ref{thm:rmv-delta} 
are used to upper bound the last two terms of \eqref{eq:interpolation:remainder:decompose} upon integrating over $\delta,\epsilon$. 
The solution of \eqref{eq:momentMatching}, that eliminates $\sR_{t,s; \epsilon, \delta}$, therefore can be considered as the ``optimal interpolation path''.
In summary, using lemmas~\ref{thm:bd_fixed_t} to \ref{thm:momentMatching:solution} on \eqref{eq:interpolation:remainder:decompose} we have
\begin{align}
\bigg \vert \int_{0}^{1} d\delta \int_{\epsilon_n}^{2\epsilon_n} d\epsilon \sR_{t,s;\epsilon,\delta} \big ( \underline{\hat{\sfx}}_n(\epsilon) \big ) \bigg \vert 
	= \mathcal{O} \big (\frac{n\epsilon_n}{T} \big ) + \mathcal{O} \big ( \frac{\sqrt{\epsilon_n}}{n^{\theta/2}} \big ) + \mathcal{O} \big ( \frac{1}{n^{\theta}} \big )
	\label{eq:R-order}
\end{align}
for any sequence $\epsilon_n$.
We are now ready to prove the upper bound.
\begin{proposition}[Upper bound]\label{prop:upperbound}
For any $\sfc \in \sB$ we have
\begin{align}
\limsup_{n \rightarrow \infty} \frac{1}{n} \mathbb{E} H(\underline{\sigma} | \underline{\tilde{J}}) 
\leq \sup_{\sfx \in \mathcal{B}} h_{\mathrm{RS}}(\sfx).
\end{align}
\end{proposition}
\begin{proof}
We evaluate \eqref{eq:bd-start} at $\underline{\sfx} = \underline{\hat{\sfx}}_n(\epsilon)$ and average the equation over $\delta \in [0,1]$, $\epsilon \in [\epsilon_n, 2 \epsilon_n]$. Using \eqref{eq:R-order} and Lemma~\ref{pert-free}
\begin{align*}
h_{1,0;\epsilon=0,\delta=0}
	& = \tilde{h}_{\epsilon=0,\delta=0}(\underline{\hat{\sfx}}_n(\epsilon)) + \mathcal{O} \big (\frac{n}{T} \big ) + \mathcal{O} \big ( \frac{1}{\sqrt{\epsilon_n n^{\theta}}} \big ) + \mathcal{O} \big ( \frac{1}{\epsilon_n n^{\theta}} \big ) + \mathcal{O}(\epsilon_n + \frac{1}{n^{\theta}}).
\end{align*}
Choosing $\epsilon_n = n^{-\gamma}$ with $0<\gamma<1$, we conclude
\begin{align}
\limsup_{n \rightarrow \infty} h_{1,0;\epsilon=0,\delta=0} = \tilde{h}_{\epsilon=0,\delta=0}(\underline{\hat{\sfx}}_n(\epsilon)).
\label{eq:up-bd-eq}
\end{align}
A trivial upper bound together with Lemma~\ref{eq:h_RS_coupled:equivalence} gives
\begin{align}
\tilde{h}_{\epsilon=0,\delta=0}(\underline{\hat{\sfx}}_n(\epsilon))
	\leq \sup_{\underline{\sfx} \in \mathcal{B}^T} \tilde{h}_{\epsilon=0,\delta=0}(\underline{\sfx})
	= \sup_{\sfx \in \mathcal{B}} h_{\mathrm{RS}}(\sfx). \label{eq:trivial-ub}
\end{align}
The proof is ended by substituting \eqref{eq:trivial-ub} into \eqref{eq:up-bd-eq} and using \eqref{h10ish}.
\end{proof}


%
\section{Proof of the fundamental sum rule~\eqref{eq:interpolation:main2}--\eqref{eq:interpolation:remainder}}
\label{proof:interpolation:main2}

Similar computations go back to \cite{FranzLeone} and were applied in Nishimori symmetric situations in \cite{Mon:2005,Mac:2007a,KudM:2009}, so we will 
be relatively brief.
We compute $h_{T,1; \epsilon,\delta}$ and $\frac{dh_{t,s; \epsilon,\delta}}{ds}$ in \eqref{eq:interpolation:main1}.
From the definitions \eqref{eq:Hamiltonian_tse}, \eqref{eq:partitionFunction_tse}, \eqref{eq:htse}, and the identity $e^{\sigma x}=(1+\sigma\tanh x)\cosh x$ for $\sigma \in\{-1,+1\}$, we can expand $h_{T,1;\epsilon,\delta}$ as
\begin{align}
h_{T,1; \epsilon,\delta}
	= \mathbb{E}\Big[ & \ln\Big( \prod_{t'=1}^{T} \prod_{B=1}^{e_{i}^{(t')}} (1 + \tanh U_{B}^{(t')}) 
	+ e^{-2(H+\tilde{H})} \prod_{t'=1}^{T} \prod_{B=1}^{e_{i}^{(t')}} (1 - \tanh U_{B}^{(t')})\Big)
	\nonumber \\ &
	- \frac{\alpha K}{T} \sum_{t'=1}^{T} \ln(1+\tanh U^{(t')})\Big].
		\label{eq:h_T1}
\end{align}
Note that the first term is part of~\eqref{eq:RSFreeEntropy_coupled1}.
For $\frac{dh_{t,s; \epsilon,\delta}}{ds}$ we use the following property of the Poisson distribution: for any function $f(X)$ of a r.v. $X$ with Poisson distribution and mean $\nu$ we have
\begin{align}
\frac{d \,\mathbb{E}f(X)}{d \nu} = \mathbb{E} f(X+1) - \mathbb{E} f(X).
\label{eq:dPoissonMean}
\end{align}
This allows us to write
\begin{align}
\frac{dh_{t,s; \epsilon,\delta}}{ds} \-
	& = - \frac{\alpha }{T} \mathbb{E}_{t,s; \epsilon,\delta}\mathbb{E}_{B, \tilde{J}_B }  \ln \< e^{\tilde{J}_{B} (\sigma_{B}-1)} \>_{t,s; \epsilon,\delta}  \-
		+ \frac{\alpha K}{nT} \sum_{i=1}^{n} \mathbb{E}_{t,s; \epsilon,\delta}\mathbb{E}_{U_i^{(t)}} \ln \< e^{U_{i}^{(t)} (\sigma_{i}-1)} \>_{t,s; \epsilon,\delta} 
		\label{eq:dh_ts}
\end{align}
where we distinguish the expectation $\mathbb{E}_{t,s; \epsilon,\delta}$ with respect to the original interpolating model with Hamiltonian \eqref{eq:Hamiltonian_tse} 
and the expectation with respect to 
an ``extra measurement'' and its neighborood $\mathbb{E}_{B, \tilde{J}_B }$ and an ``extra field'' $\mathbb{E}_{U_i^{(t)}}$.
Standard algebra, using again the identity $e^{\pm x}=(1\pm\tanh x)\cosh x$, leads to 
\begin{align}
\mathbb{E}  \ln \< e^{\tilde{J}_{B} (\sigma_{B} -1)} \>_{t,s; \epsilon,\delta} 
	& = \mathbb{E}_{t,s; \epsilon,\delta}\mathbb{E}_{B,\tilde{J}_B}  
	 \ln \big( 1 + \langle \sigma_B \rangle_{t,s; \epsilon,\delta} \tanh \tilde{J}_B \big ) 
	- \mathbb{E}_{\tilde{J}_B} \ln(1+ \tanh \tilde J_B) \nonumber \\
	& =\sum_{p=1}^{\infty} \frac{(-1)^{p+1}}{p} \frac{1}{n^K} \sum_{i_1,\dots,i_K} 
	\mathbb{E}[\langle \sigma_{i_1} \cdots \sigma_{i_K} \rangle_{t,s; \epsilon,\delta}^p] 
	\mathbb{E}[(\tanh \tilde{J})^p] 
	-\mathbb{E} \ln(1+ \tanh \tilde J)
	\nonumber 
\end{align}
and similarly, using \eqref{eq:BP-update},
\begin{align}
	& \frac{1}{n} \sum_{i=1}^{n} \mathbb{E} \ln \< e^{U_{i}^{(t)} (\sigma_{i}-1)} \>_{t,s; \epsilon,\delta}  \nonumber \\
	 = ~&\frac{1}{n}\sum_{i=1}^{n} \mathbb{E}_{t,s; \epsilon,\delta}\mathbb{E}_{U_i^{(t)}} 
	\ln \big ( 1 + \langle \sigma_i \rangle_{t,s; \epsilon,\delta} \tanh U_i^{(t)}  \big ) 
	-
	\frac{1}{n}\sum_{i=1}^{n} \mathbb{E}\ln(1+ \tanh U_i^{(t)})
	\nonumber \\
	=~& \frac{1}{n}\sum_{i=1}^{n} \mathbb{E}_{t,s; \epsilon,\delta}\mathbb{E}_{\tilde{J}, \underline{V}^{(t)} } 
	 \ln\big( 1 + \langle \sigma_i \rangle_{t,s; \epsilon,\delta} \tanh \tilde{J} \prod_{j=1}^{K-1} \tanh V_j^{(t)} \big) 
	-  \mathbb{E}\ln(1+ \tanh U^{(t)}) \nonumber \\
	=~& \sum_{p=1}^{\infty} \frac{(-1)^{p+1}}{p}  \mathbb{E}[(\tanh \tilde{J})^p]
	\mathbb{E} [ (\tanh V^{(t)})^p]^{K-1}\frac{1}{n} 
	\sum_{i=1}^n \mathbb{E}[\langle \sigma_i \rangle_{t,s; \epsilon,\delta}^p] 
	- \mathbb{E} \ln(1+ \tanh U^{(t)}). \nonumber 
\end{align}
Recall $Q_{p} \equiv \frac{1}{n} \sum_{i} \sigma_i^{(1)} \cdots \sigma_i^{(p)}$ and thus 
$$
\< Q_p \>_{t,s; \epsilon,\delta} = \frac{1}{n} \sum_{i} \langle \sigma_{i}\rangle_{t,s; \epsilon,\delta}^p, 
\quad \< Q_p^K \>_{t,s; \epsilon,\delta} = \frac{1}{n^K} \sum_{i_1,\dots,i_K} \langle \sigma_{i_1} \ldots \sigma_{i_K} \rangle_{t,s; \epsilon,\delta}^p\,.
$$
Recall also $q_p^{(t)} \equiv \mathbb{E} [(\tanh V^{(t)})^p]$. Then \eqref{eq:dh_ts} becomes
\begin{align}
\frac{dh_{t,s; \epsilon,\delta}}{ds} \-
	 & =
	 - \frac{\alpha}{T} \sum_{p=1}^{\infty} \frac{(-1)^{p+1}}{p} 
	\mathbb{E}[(\tanh \tilde{J})^p] 
	\mathbb{E}[\langle Q_p^K \rangle_{t,s; \epsilon,\delta} - K (q_p^{(t)})^{K-1} \langle Q_p \rangle_{t,s; \epsilon,\delta}] 
	\nonumber \\
	& \ + \frac{\alpha}{T}\mathbb{E}\ln (1 + \tanh\tilde J) 
	- \frac{\alpha K}{T}\mathbb{E}\ln (1 + \tanh U^{(t)})  
	\nonumber \\
	& = - \frac{\alpha}{T} \sum_{p=1}^{\infty} \frac{(-1)^{p+1}}{p} \mathbb{E}[( \tanh \tilde{J})^p] 
	\mathbb{E} \big\langle Q_{p}^{K} - K (q_{p}^{(t)})^{K-1} ( Q_{p} - q_{p}^{(t)}) - (q_{p}^{(t)})^K \big\rangle_{t,s; \epsilon,\delta} \nonumber \\
	& \ + \frac{\alpha(K-1)}{T} \mathbb{E}\ln\big( 1 + \tanh \tilde{J} \prod_{j=1}^{K} \tanh V_j^{(t)}\big) 
	+ \frac{\alpha}{T} \mathbb{E}\ln(1+ \tanh \tilde J) 
	- \frac{\alpha K}{T} \mathbb{E}\ln(1+ \tanh U^{(t)}) .
	\label{eq:dh_st _final}
\end{align}
Substituting \eqref{eq:h_T1} and \eqref{eq:dh_st _final} into \eqref{eq:interpolation:main1} gives~\eqref{eq:interpolation:main2}, where
\begin{align}
\sR_{t,s; \epsilon,\delta} = \sum_{p=1}^{\infty} \frac{(-1)^{p+1}}{p} \mathbb{E}[ (\tanh \tilde{J})^{p} ]
\mathbb{E} \big\langle Q_{p}^{K} - K (q_{p}^{(t)})^{K-1} ( Q_{p} - q_{p}^{(t)}) - (q_{p}^{(t)})^K 
\big\rangle_{t,s; \epsilon,\delta}.
\label{eq:interpolation:remainder:unpaired}
\end{align}
An application of \eqref{eq:Nishimori} yields 
$$
\mathbb{E}\langle Q_{2p-1}^m\rangle_{t, s;\epsilon,\delta} = \mathbb{E}\langle Q_{2p}^m\rangle_{t, s;\epsilon,\delta}
$$ 
for all $m\in \mathbb{N}$ and $p\geq 1$.
Similarly 
an application of \eqref{eq:channelSymmetry} yields $$q_{2p-1}^{(t)} = q_{2p}^{(t)} \quad \text{as well as} \quad \mathbb{E}[ (\tanh \tilde{J})^{2p-1} ]= \mathbb{E}[ (\tanh \tilde{J})^{2p} ]$$ for $p\geq 1$. Therefore 
combining the odd and even terms of \eqref{eq:interpolation:remainder:unpaired} we obtain the form in \eqref{eq:interpolation:remainder}.

\section{Concentration of overlaps I: Proof of lemmas~\ref{thm:concentration1} and \ref{thm:concentration2}}
\label{sec:overlap}
In this section we prove lemmas \ref{thm:concentration1} and \ref{thm:concentration2}. We need the following lemmas proved in the next section \ref{sec:thermal-and-quenched}. For lemma \ref{thm:concentration1}
it suffices to take an interpolation path $\underline{\sfx}\in \mathcal{B}^T$ independent of $\epsilon$ and $\delta$. However for $\ref{thm:concentration2}$ we need 
to take $\underline{\sfx}(\epsilon) \in \mathcal{B}^{T}$ dependent 
on $\epsilon$ (and independent of $\delta$). We therefore formulate the lemmas below for an $\epsilon$-dependent interpolation path. 

\begin{lemma}[Concentration of $Q_p$ on $\< Q_{p} \>_{t,s;\epsilon, \delta}$]\label{thm:overlap:thermal:concentration}
For any $\sfc \in \mathcal{B}$ and any choice of interpolating path $\underline{\sfx}(\epsilon) \in \mathcal{B}^{T}$ such that every component satisfies $dx^{(t)} / d\epsilon \geq 0$, we have
\begin{align}
\int_{\varepsilon_0}^{\varepsilon_1} d\epsilon \, \mathbb{E} \big\< ( Q_{p} - \< Q_{p} \>_{t,s;\epsilon, \delta} )^2\big \>_{t,s;\epsilon, \delta} 
	& \leq \frac{3p}{n} \label{eq:overlap:thermal:concentration}
\end{align}
uniformly in $t, s, \delta$.
\end{lemma}

\begin{lemma}[Concentration of $\<Q_1\>$ on $\mathbb{E}_{\underline{\tilde{H}}} \< Q_{1} \>_{t,s;\epsilon, \delta}$]\label{thm:overlap:pert:concentration}
For any $\sfc \in \mathcal{B}$ and $\underline{\sfx}(\epsilon) \in \mathcal{B}^{T}$ and any choice of interpolating path such that every component satisfies $dx^{(t)} / d\epsilon \geq 0$, we have
\begin{align}
\int_{\varepsilon_0}^{\varepsilon_1} d\epsilon \, \mathbb{E} \big [( \< Q_{1} \>_{t,s;\epsilon, \delta} - \mathbb{E}_{\underline{\tilde{H}}} \< Q_{1} \>_{t,s;\epsilon, \delta} )^2 \big ]
	& \leq \frac{3\delta}{n^{\theta}} \label{eq:overlap:pert:concentration}
\end{align}
for any $\theta\in (0, 1]$, uniformly in $t, s, \delta$.
\end{lemma}


\begin{lemma}[Concentration of $\< Q_{1} \>_{t,s;\epsilon, \delta}$ on $\mathbb{E}\< Q_{1} \>_{t,s;\epsilon, \delta}$]\label{thm:overlap:graph:concentration}
For $\sfc \in \sB$ and any choice of interpolating path $\underline{\sfx}(\epsilon) \in \sB^{T}$, we have
\begin{align}
\int_{\delta_0}^{\delta_1} d\delta\, \mathbb{E} \big [ ( \mathbb{E}_{\underline{\tilde{H}}} \< Q_1 \>_{t,s; \epsilon, \delta} - \mathbb{E} \< Q_1 \>_{t,s; \epsilon, \delta} )^2 \big ]
	\leq \bigg ( \frac{15{C}(\delta_1-\delta_0)}{(\ln 2)^2} + 4 \bigg ) n^{-(1-2\theta)/3} \label{eq:overlap:graph:concentration}
\end{align}
for any $\theta\in (0, 1/2)$,
uniformly in $t, s, \epsilon$, with $C >0$ a constant (this constant is obtained from Lemma~\ref{thm:concentration:free_energy}).
\end{lemma}

\begin{remark}
 We already saw that Lemma \ref{thm:BEC:overlap} implies for the BEC $\langle Q_1\rangle_{t,s;\epsilon,\delta} = \langle Q_p\rangle_{t,s;\epsilon,\delta}$ and 
 therefore the last two concentration lemmas are valid for all overlaps. 
\end{remark}

\subsection{Proof of lemma \ref{thm:concentration1}}

In lemma \ref{thm:concentration1} we take $\underline\sfx$ independent of $\epsilon$, thus $dx^{(t)}/d\epsilon = 0$. We have
\begin{align}
\mathbb{E}  \big\< \big| Q_{p}^{K} - \< Q_{p} \>_{t,s;\epsilon, \delta}^K \big| \big\>_{t,s;\epsilon, \delta} 
	& = \mathbb{E}\Big \< \Big | ( Q_{p} - \< Q_{p} \>_{t,s;\epsilon, \delta} )\sum_{k=0}^{K-1} Q_{p}^{K-k-1} \< Q_{p} \>_{t,s;\epsilon, \delta}^{k} \Big | \Big \>_{t,s;\epsilon, \delta}  \nonumber \\
	& \leq K\, \mathbb{E}\big \< \big| Q_{p} - \< Q_{p}\>_{t,s;\epsilon, \delta}\big| \big\>_{t,s;\epsilon, \delta} \, .\label{eq:overlap:thermal:Kexpansion}
\end{align}
We can apply the Cauchy-Schwarz inequality to get
\begin{align}
 \int_{\varepsilon_0}^{\varepsilon_1} d\epsilon \ \mathbb{E}  \big\< \big| Q_{p}^{K} - \< Q_{p} \>_{t,s;\epsilon, \delta}^K \big| \big\>_{t,s;\epsilon, \delta} 
	& \leq K \Big\{ (\varepsilon_1 - \varepsilon_0) \int_{\varepsilon_0}^{\varepsilon_1} d\epsilon \ \mathbb{E}  \big\< ( Q_{p} - \< Q_{p} \>_{t,s;\epsilon, \delta} )^2 \big\>_{t,s;\epsilon, \delta}  \Big\}^{1/2}.
	\label{eq:overlap:thermal:CSonK}
\end{align}
Thanks to \eqref{eq:overlap:thermal:concentration} we obtain
\begin{align*}
\int_{\varepsilon_0}^{\varepsilon_1} d\epsilon \ \mathbb{E} \big\< \big| Q_{p}^{K} - \< Q_{p} \>_{t,s;\epsilon, \delta}^K \big| \big\>_{t,s;\epsilon, \delta} 
	\leq K \Big( \frac{3p(\varepsilon_1-\varepsilon_0)}{n} \Big )^{1/2}.
\end{align*}
This proves Lemma~\ref{thm:concentration1}.

\subsection{Proof of lemma \ref{thm:concentration2}}

Similar to \eqref{eq:overlap:thermal:Kexpansion} and \eqref{eq:overlap:thermal:CSonK}, it is easy to show
\begin{align}
& \int_{\delta_0}^{\delta_1} d\delta \int_{\varepsilon_0}^{\varepsilon_1} d\epsilon \mathbb{E} \big[ \big| \< Q_{1} \>_{t,s;\epsilon, \delta}^K  - \mathbb{E} [\< Q_{1} \>_{t,s;\epsilon, \delta}]^K \big| \big ] 
\nonumber \\ &
\leq K \, \big \{ (\delta_1-\delta_0) (\varepsilon_1-\varepsilon_0) \int_{\delta_0}^{\delta_1} d\delta \int_{\varepsilon_0}^{\varepsilon_1} \mathbb{E} \big[ ( \<Q_1\>_{t,s;\epsilon, \delta} - \mathbb{E}\<Q_1\>_{t,s;\epsilon, \delta} )^2 \big] \big \}^{1/2}. \label{eq:overlap:full1}
\end{align}
We decompose $\mathbb{E} \big[ ( \<Q_1\>_{t,s;\epsilon, \delta} - \mathbb{E}\<Q_1\>_{t,s;\epsilon, \delta} )^2 \big]$ in three parts:
\begin{align*}
\mathbb{E} \big\< ( Q_{1} - \< Q_{1} \>_{t,s;\epsilon, \delta} )^2\big \>_{t,s;\epsilon, \delta}
	+ \mathbb{E} \big [( \< Q_{p} \>_{t,s;\epsilon, \delta} - \mathbb{E}_{\underline{\tilde{H}}} \< Q_{p} \>_{t,s;\epsilon, \delta} )^2 \big ]
	+ \mathbb{E} \big [ ( \mathbb{E}_{\underline{\tilde{H}}} \< Q_1 \>_{t,s; \epsilon, \delta} - \mathbb{E} \< Q_1 \>_{t,s; \epsilon, \delta} )^2 \big ].
\end{align*}
With Fubini's theorem we are free to switch the $\delta$ and $\epsilon$ integrals.  Lemmas~\ref{thm:overlap:thermal:concentration}, \ref{thm:overlap:pert:concentration} and \ref{thm:overlap:graph:concentration} then imply 
\begin{align}
\int_{\delta_0}^{\delta_1} d\delta \int_{\varepsilon_0}^{\varepsilon_1} d\epsilon \mathbb{E} \big[ ( \<Q_1\>_{t,s;\epsilon, \delta} - \mathbb{E}\<Q_1\>_{t,s;\epsilon, \delta} )^2 \big]
& \leq \frac{3(\delta_1-\delta_0)}{n} + \frac{3(\delta_1^2-\delta_0^2)}{n^{\theta}} + \big ( \frac{15{C} (\delta_1-\delta_0)}{(\ln 2)^2} + 4 \big ) \frac{\varepsilon_1-\varepsilon_0}{n^{(1-2\theta)/3}} \nonumber \\
& \leq \frac{3(\delta_1-\delta_0)(1+\delta_0+\delta_1)}{n^{\theta}} + \big ( \frac{15{C} (\delta_1-\delta_0)}{(\ln 2)^2} + 4 \big ) \frac{\varepsilon_1-\varepsilon_0}{n^{(1-2\theta)/3}}. 
\label{eq:overlap:full2}
\end{align}
The bound \eqref{eq:overlap:full2} is optimal for $\theta = (1-2\theta)/3$, i.e., $\theta=1/5$. But any $\theta\in (0, 1/2)$ will do. Lemma~\ref{thm:concentration2} is 
then obtained by substituting \eqref{eq:overlap:full2} into \eqref{eq:overlap:full1}:
\begin{align*}
& \int_{\delta_0}^{\delta_1} d\delta \int_{\varepsilon_0}^{\varepsilon_1} d\epsilon \mathbb{E} \big[ \big| \< Q_{1} \>_{t,s;\epsilon, \delta}^K  - \mathbb{E} [\< Q_{1} \>_{t,s;\epsilon, \delta}]^K \big| \big ] \nonumber \\
	& \leq K \bigg \{ (\delta_1-\delta_0) (\varepsilon_1-\varepsilon_0) \Big ( 3(\delta_1-\delta_0)(1+\delta_0+\delta_1) + (\varepsilon_1-\varepsilon_0) \big ( \frac{15{C} (\delta_1-\delta_0)}{(\ln 2)^2} + 4 \big ) \Big ) \bigg \}^{1/2} n^{-\theta/2} \\
	& = \mathcal{O} \bigg ( \big ( (\delta_1-\delta_0)(\varepsilon_1-\varepsilon_0) \frac{\delta_1-\delta_0 + \varepsilon_1-\varepsilon_0}{n^{\theta}} \big )^{1/2} \bigg ).
\end{align*}

\section{Concentration of overlaps II: Proof of lemmas~\ref{thm:overlap:thermal:concentration}, \ref{thm:overlap:pert:concentration}, \ref{thm:overlap:graph:concentration}}
\label{sec:thermal-and-quenched}

We start with useful preliminary results on the derivatives of the free entropy of the interpolated system, and then prove the three concentration lemmas.

\subsection{Useful derivative formulas}\label{sec:usefulderivatives}

We first remark that, according to \eqref{eq:BP-update}, the distribution $\tilde{\sfx}^{(t)}$ is a function of $\sfx^{(t)}$ and $\sfc$. 
Therefore $\tilde{\sfx}^{(t)} \in \sB$ when $\sfx^{(t)}, \sfc \in \sB$. 
Let $\sfx^{(t)} = x^{(t)} \Delta_\infty + (1-x^{(t)}) \Delta_0$ and $\tilde{\sfx}^{(t)} = \tilde{x}^{(t)} \Delta_\infty + (1-\tilde{x}^{(t)}) \Delta_0$ \footnote{We prefer to write $1-x^{(t)}$ and $1-\tilde{x}^{(t)}$ as the erasure probability in this interpolation to align with the way we define the distribution of $H$ and $\tilde{H}$.}.
From \eqref{eq:BP-update} we have the relation
\begin{align}
\tilde{x}^{(t)} = (1-q) x^{(t) K-1}. \label{eq:BP-update:BEC}
\end{align}

We now provide another 
view of the interpolating Hamiltonian \eqref{eq:Hamiltonian_ts}. Consider an ``effective half-edge'' fed into node $i$ with random half-log-likelihood variable
\begin{align*}
\bar{H}_i^{(t,s)} \equiv \sum_{t' =1}^{t-1} \sum_{B=1}^{e_{i}^{(t')}} U_{B \rightarrow i}^{(t')} + \sum_{C=1}^{e_{i,s}^{(t)}} U_{C \rightarrow i}^{(t)} + H_i + \tilde{H}_i,
\end{align*}
equal to $\infty$ with probability
\begin{align*}
\bar{\epsilon}_i^{(t,s)} 
	\equiv 1 - (1-\epsilon)(1-\frac{\delta}{n^{\theta}}) (1-\tilde{x}^{(t)})^{e_i^{(t,s)}} \prod_{t'=1}^{t-1}(1-\tilde{x}^{(t')})^{e_i^{(t')}}.
\end{align*}
and equal to $0$ complementary probability. Set $\bar{\mathcal{E}}=(\bar{\epsilon}_1^{(t,s)} ,\cdots, \bar{\epsilon}_n^{(t,s)})$. 
The Hamiltonian \eqref{eq:Hamiltonian_ts} is equal in distribution to 
\begin{align}
{\bar{\sH}}_{t,s; \bar{\mathcal{E}}}(\underline{\sigma}, \underline{\tilde{J}}, \underline{\bar{H}})
	\equiv - \sum_{i=1}^{n} \bar{H}_i^{(t,s)} ( \sigma_i - 1 )
	- \sum_{A=1}^{m_s^{(t)}} \tilde{J}_{A} ( \sigma_{A} - 1 ). \label{eq:Htse-abstract}
\end{align}
Let $n^{-1} \mathbb{E} \ln \bar{\mathcal{Z}}_{t,s; \bar{\mathcal{E}}}$ the associated averaged free entropy. Clearly this is a function 
of $(\mathbb{E}[\bar{\epsilon}_1^{(t,s)}] ,\cdots, \mathbb{E}[\bar{\epsilon}_n^{(t,s)}])$ where for all $i=1,\cdots,n$
\begin{align}
\mathbb{E}[\bar{\epsilon}_i^{(t,s)}] 
	& \equiv \mathbb{E}_{e_i^{(1)}, \dots, e_i^{(t-1)}, e_i^{(t,s)}}[\bar{\epsilon}_i^{(t,s)}] \nonumber \\
	& = 1 - (1-\epsilon)(1-\frac{\delta}{n^{\theta}}) e^{-\frac{K}{RT}(s \tilde{x}^{(t)} + \sum_{t'=1}^{t-1} \tilde{x}^{(t')})}. \label{eq:E-epsilon-bar}
\end{align}
Moreover it is clear that $h_{t,s;\epsilon,\delta} = n^{-1} \mathbb{E} \ln \bar{\mathcal{Z}}_{t,s;\bar{\mathcal{E}}}$. Therefore we see that the dependence in $\epsilon$ and $\delta$ effectively comes through the combination 
\eqref{eq:E-epsilon-bar}. Since this is independent of $i$ we denote it by $\mathbb{E}[\bar{\epsilon}^{(t,s)}]$. The reader should keep in mind that in this combination 
there is always an explicit $(\epsilon, \delta)$, and that there may also be an implicit 
one through the choice of the interpolating path $(x^{(t)}, \tilde{x}^{(t)})$. 

We are now ready to state 
 derivative formulas playing an important role. Their detailed derivation is provided in Appendix~\ref{app:entropy-derivatives}:
\begin{align}
\frac{d}{d \mathbb{E} [\bar{\epsilon}^{(t,s)}]} h_{t,s; \epsilon, \delta} 
	& = - \frac{\ln 2}{n} \sum_{i=1}^{n}  ( 1 - \mathbb{E} \< \sigma_i \>_{t, s;\epsilon,\delta;\sim \bar{H}_i^{(t,s)}} )
	\- = - \frac{\ln 2}{n(1-\mathbb{E}[\bar{\epsilon}^{(t,s)}])} \sum_{i=1}^n (1 - \mathbb{E}\<\sigma_i\>_{t,s;\epsilon,\delta}),
\label{eq:perturbed_f:derivative1a} \\
\frac{d^2}{d \mathbb{E} [\bar{\epsilon}^{(t,s)}]^2} h_{t,s; \epsilon,\delta} 
	& = \frac{\ln 2}{n(1-\mathbb{E} [\bar{\epsilon}^{(t,s)}])^2} \sum_{i \neq j} \mathbb{E} [\< \sigma_i \sigma_j \>_{t,s;\epsilon,\delta} 
- \< \sigma_i \>_{t,s;\epsilon,\delta} \< \sigma_j \>_{t, s;\epsilon,\delta}],
\label{eq:perturbed_f:derivative1b} \\
\frac{d}{d\delta} H_{t,s; \epsilon,\delta}
	& = - \frac{\ln 2}{n^{1+\theta}} \sum_{i=1}^{n}  ( 1 - \mathbb{E}_{\underline{\tilde{H}}} \< \sigma_i \>_{t, s;\epsilon,\delta;\sim \tilde{H}_i} )
	\- = - \frac{\ln 2}{n^{1+\theta}(1-\delta/ n^{\theta})} \sum_{i=1}^n (1 - \mathbb{E}_{\underline{\tilde{H}}}\<\sigma_i\>_{t,s;\epsilon,\delta}), \label{eq:perturbed_f:derivative4a} \\
\frac{d^2}{d\delta^2} H_{t,s; \epsilon,\delta}
	& = \frac{1}{n^{2\theta}(1-\delta/n^{\theta})^2} \sum_{i \neq j} \mathbb{E}_{\underline{\tilde{H}}} \ln \biggl\{\frac{1 + \< \sigma_i \>_{\sim \tilde{H}_i, \tilde{H}_j} + \< \sigma_j \>_{\sim \tilde{H}_i, \tilde{H}_j} + \< \sigma_i \sigma_j \>_{\sim \tilde{H}_i, \tilde{H}_j}}{ 1 + \< \sigma_i \>_{\sim \tilde{H}_i, \tilde{H}_j} + \< \sigma_j \>_{\sim \tilde{H}_i, \tilde{H}_j} + \< \sigma_i \>_{\sim \tilde{H}_i, \tilde{H}_j} \< \sigma_j \>_{\sim \tilde{H}_i, \tilde{H}_j}}\biggr\}, \label{eq:perturbed_f:derivative4b}
\end{align}
where $\< \sigma_i \>_{t, s;\epsilon,\epsilon;\sim \bar{H}_i^{(t,s)}}$ is the Gibbs expectation with fixed $\bar{H}_i^{(t,s)}=0$ 
and $\< \sigma_i \>_{\sim \tilde{H}_i, \tilde{H}_j} \equiv \< \sigma_i \>_{t,s;\epsilon, \delta; \sim \tilde{H}_i, \tilde{H}_j}$ is the Gibbs expecetaion with fixed $\tilde H_i=\tilde H_j=0$.
%
%
If we {\it choose} $\underline{\sfx} = \underline{\sfx}(\epsilon)$ {\it independent of} $\delta$ we have furthermore
\begin{align}
\frac{d}{d \delta} h_{t,s; \epsilon, \delta} 
	& = - \frac{\ln 2}{n^{1+\theta}} \sum_{i=1}^{n}  ( 1 - \mathbb{E} \< \sigma_i \>_{t, s;\epsilon,\delta;\sim \tilde{H}_i} ) 
	\- = - \frac{\ln 2}{n^{1+\theta}(1-\delta/ n^{\theta})} \sum_{i=1}^n (1 - \mathbb{E}\<\sigma_i\>_{t,s;\epsilon,\delta}), \label{eq:perturbed_f:derivative3a} \\
\frac{d^2}{d \delta^2} h_{t,s; \epsilon,\delta} 
	& = \frac{\ln 2}{n^{1+2\theta}(1-\delta/n^{\theta})^2} \sum_{i \neq j} \mathbb{E} [\< \sigma_i \sigma_j \>_{t,s;\epsilon,\delta} 
- \< \sigma_i \>_{t,s;\epsilon,\delta} \< \sigma_j \>_{t, s;\epsilon,\delta}], \label{eq:perturbed_f:derivative3b} \\
\frac{d}{d\delta} \big ( \frac{1}{n} \sum_{i=1}^{n} \mathbb{E}\<\sigma_i\>_{t,s;\epsilon,\delta} \big )
	& = \frac{1}{n^{1+\theta}(1-\delta/n^{\theta})} \sum_{i, j=1}^{n} \mathbb{E} [\< \sigma_i \sigma_j \>_{t,s;\epsilon,\delta} 
- \< \sigma_i \>_{t,s;\epsilon,\delta} \< \sigma_j \>_{t, s;\epsilon,\delta}]. \label{eq:Q1:delta-der}
\end{align}
The first equalities of \eqref{eq:perturbed_f:derivative1a}, \eqref{eq:perturbed_f:derivative4a}, \eqref{eq:perturbed_f:derivative3a}, together with \eqref{eq:GKS1}, tell us that
\begin{align}
\Big\vert \frac{d}{d\mathbb{E}[\bar{\epsilon}^{(t,s)}]} h_{t,s; \epsilon} \Big\vert \leq \ln 2, \quad \Big\vert \frac{d}{d\delta} H_{t,s; \epsilon,\delta} \Big\vert\leq \frac{\ln 2}{n^\theta}, \quad
\Big\vert \frac{d}{d\delta} h_{t,s; \epsilon,\delta} \Big\vert\leq \frac{\ln 2}{n^\theta}\,.
\label{littlebounds}
\end{align}
Moreover from \eqref{eq:perturbed_f:derivative4b}, \eqref{eq:perturbed_f:derivative3b} and the second GKS inequality \eqref{eq:GKS2}, we see that $h_{t,s; \epsilon, \delta}$ and $H_{t,s;\epsilon,\delta}$ are 
convex in $\delta$. 


\subsection{Proof of Lemma~\ref{thm:overlap:thermal:concentration}}
\label{sec:overlap:thermal}
From the definition of $Q_p$ we have
\begin{align}
\mathbb{E}\big\< ( Q_{p} -& \< Q_{p} \>_{t,s;\epsilon, \delta} )^2 \big\>_{t,s;\epsilon, \delta}  
	 = \frac{1}{n^2} \sum_{i,j=1}^{n} \mathbb{E} \big[ \< \sigma_i \sigma_j \>_{t,s;\epsilon, \delta}^p - \< \sigma_i \>_{t,s;\epsilon, \delta}^p \< \sigma_j \>_{t,s;\epsilon, \delta}^p \big] \nonumber \\
	& = \frac{1}{n^2} \sum_{i,j=1}^{n} \mathbb{E} \Big [ ( \< \sigma_i \sigma_j \>_{t,s;\epsilon, \delta} - \< \sigma_i \>_{t,s;\epsilon, \delta} \< \sigma_j \>_{t,s;\epsilon, \delta} ) 
	\sum_{l=0}^{p-1}  \< \sigma_i \sigma_j \>_{t,s;\epsilon, \delta}^{p-1-l} \< \sigma_i \>_{t,s;\epsilon, \delta}^l \< \sigma_j \>_{t,s;\epsilon, \delta}^l\Big ].
	\label{eq:overlap:thermal:expansion}
\end{align}
By \eqref{eq:GKS2} we have $0 \leq \< \sigma_i \sigma_j \>_{t,s;\epsilon, \delta} - \< \sigma_i \>_{t,s;\epsilon, \delta} \< \sigma_j \>_{t,s;\epsilon, \delta}$. This allows us to upper bound \eqref{eq:overlap:thermal:expansion} as
\begin{align}
\mathbb{E}\big\< ( Q_{p} - \< Q_{p} \>_{t,s;\epsilon, \delta} )^2 \big\>_{t,s;\epsilon, \delta}
	& \leq \frac{1}{n^2} \sum_{i,j=1}^{n} \mathbb{E} \Big [  (\< \sigma_i \sigma_j \>_{t,s;\epsilon, \delta} - \< \sigma_i \>_{t,s;\epsilon, \delta} \< \sigma_j \>_{t,s;\epsilon, \delta})
		\sum_{l=0}^{p-1}  \big | \< \sigma_i \sigma_j \>_{t,s;\epsilon, \delta}^{p-1-l} \< \sigma_i \>_{t,s;\epsilon, \delta}^l \< \sigma_j \>_{t,s;\epsilon, \delta}^l \big | \Big ] \nonumber \\
	& \leq \frac{p}{n^2} \sum_{i,j=1}^{n} \mathbb{E} \Big[ \< \sigma_i \sigma_j \>_{t,s;\epsilon, \delta} - \< \sigma_i \>_{t,s;\epsilon, \delta} \< \sigma_j \>_{t,s;\epsilon, \delta} \Big ]. \label{eq:overlap:thermal:expansion2}
\end{align} 
Hence integrating \eqref{eq:overlap:thermal:expansion2} over $\epsilon \in [\varepsilon_0, \varepsilon_1]$ and recalling the formula \eqref{eq:perturbed_f:derivative1b}, we obtain
\begin{align}
\int_{\varepsilon_0}^{\varepsilon_1} d\epsilon \, \mathbb{E} \big\< ( Q_{p} - \< Q_{p} \>_{t,s;\epsilon, \delta} )^2 \big\>_{t,s;\epsilon, \delta}
	& \leq p \int_{\varepsilon_0}^{\varepsilon_1} d\epsilon \, \frac{1}{n^{2}} \sum_{i,j=1}^{n} \mathbb{E} \big[ \< \sigma_i \sigma_j \>_{t,s;\epsilon, \delta} - \< \sigma_i \>_{t,s;\epsilon, \delta} \< \sigma_j \>_{t,s;\epsilon, \delta} \big ] \nonumber \\
	& \leq \frac{p}{n} + \frac{p}{n\ln 2} \int_{\varepsilon_0}^{\varepsilon_1} d\epsilon (1-\mathbb{E}[\bar{\epsilon}^{(t,s)}]) \frac{d^2}{d\mathbb{E}[\bar{\epsilon}^{(t,s)}]^2} h_{t,s; \epsilon,\delta}\,.
	\label{eq:overlap:thermal:bd1}
\end{align}
Recall \eqref{eq:E-epsilon-bar} for the expression of $\mathbb{E}[\bar{\epsilon}^{(t,s)}]$. Under the hypothesis $d x^{(t)} / d\epsilon \geq 0$ 
for all $t \in \{1,\dots, T\}$ and using \eqref{eq:BP-update:BEC}, we have $d \tilde{x}^{(t)} / d\epsilon \geq 0$. This gives
\begin{align*}
\frac{d\mathbb{E}[\bar{\epsilon}^{(t,s)}]}{d\epsilon}
	& = \frac{1-\mathbb{E}[\bar{\epsilon}^{(t,s)}]}{1-\epsilon} + (1-\mathbb{E}[\bar{\epsilon}^{(t,s)}]) \frac{\alpha K}{T} \big ( s \frac{d\tilde{x}^{(t)}}{d\epsilon} + \sum_{t'=1}^{t-1} \frac{d\tilde{x}^{(t')}}{d\epsilon} \big )
	\geq 1-\mathbb{E}[\bar{\epsilon}^{(t,s)}]
\end{align*}
and allows us to relax the second term of \eqref{eq:overlap:thermal:bd1}:
\begin{align*}
\int_{\varepsilon_0}^{\varepsilon_1} d\epsilon \, \mathbb{E} \big\< ( Q_{p} - \< Q_{p} \>_{t,s;\epsilon, \delta} )^2 \big\>_{t,s;\epsilon, \delta}
	& \leq \frac{p}{n} + \frac{p}{n\ln 2} \int_{\varepsilon_0}^{\varepsilon_1} d\epsilon \frac{d\mathbb{E}[\bar{\epsilon}^{(t,s)}]}{d\epsilon} \frac{d^2}{d\mathbb{E}[\bar{\epsilon}^{(t,s)}]^2} h_{t,s; \epsilon,\delta} \nonumber \\
	& = \frac{p}{n} + \frac{p}{n\ln 2} \int_{\varepsilon_0}^{\varepsilon_1} d\epsilon \frac{d}{d\epsilon} \Big ( \frac{d}{d\mathbb{E}[\bar{\epsilon}^{(t,s)}]} h_{t,s; \epsilon,\delta} \Big ) \nonumber \\
	& = \frac{p}{n} + \frac{p}{n\ln 2} \Big [ \frac{d}{d\mathbb{E}[\bar{\epsilon}^{(t,s)}]} h_{t,s; \epsilon,\delta} \Big ]_{\epsilon = \varepsilon_0}^{\epsilon = \varepsilon_1} \\
	&  \leq \frac{3p}{n} ,
\end{align*}
using the first bound in \eqref{littlebounds} for the last inequality.

\subsection{Proof of Lemma~\ref{thm:overlap:pert:concentration}}
\label{sec:overlap:pert}
Let $\tilde{H}'_j$ be an i.i.d. copy of $\tilde{H}_j$. Let also $\underline{\tilde{H}}^j$ be a vector same as $\underline{\tilde{H}}$ except the $j$-th component is replaced by $\tilde{H}'_j$.
By the Efron-Stein inequality we have
\begin{align}
\mathbb{E}_{\underline{\tilde{H}}} \big [ \big ( \<Q_1\>_{\underline{\tilde{H}}} - \mathbb{E}_{\underline{\tilde{H}}} \<Q_1\>_{\underline{\tilde{H}}} \big )^2 \big ]
	& \leq \frac{1}{2} \sum_{j=1}^{n} \mathbb{E}_{\underline{\tilde{H}}} \mathbb{E}_{\tilde{H}'_j} \big [ \big ( \<Q_1\>_{\underline{\tilde{H}}^j} - \<Q_1\>_{\underline{\tilde{H}}} \big )^2 \big ] 
	\nonumber \\
	& = \frac{1}{2} \sum_{j=1}^{n} \mathbb{E}_{\underline{\tilde{H}}} \mathbb{E}_{\tilde{H}'_j} \big [ \big ( \<Q_1\>_{\underline{\tilde{H}}^j} - \<Q_1\>_{\underline{\tilde{H}}} \big )^2 \mathbb{I}(\tilde{H}'_j 
	= \infty) \mathbb{I}(\tilde{H}_j = 0) \big ] \nonumber \\
	& \hspace{0.5cm} + \frac{1}{2} \sum_{j=1}^{n} \mathbb{E}_{\underline{\tilde{H}}} 
	\mathbb{E}_{\tilde{H}'_j} \big [ \big ( \<Q_1\>_{\underline{\tilde{H}}^j} - \<Q_1\>_{\underline{\tilde{H}}} \big )^2 \mathbb{I}(\tilde{H}'_j = 0) \mathbb{I}(\tilde{H}_j = \infty) \big ] 
	\nonumber \\
	& = \sum_{j=1}^{n} \mathbb{E}_{\underline{\tilde{H}}} \mathbb{E}_{\tilde{H}'_j} \big [ \big ( \<Q_1\>_{\underline{\tilde{H}}^j} - \<Q_1\>_{\underline{\tilde{H}}} \big )^2 \mathbb{I}(\tilde{H}'_j 
	= \infty) \mathbb{I}(\tilde{H}_j = 0) \big ]. \label{eq:overlap:pert:ES2}
\end{align}
To see the last equality we can exchange $\tilde{H}_j$ and $\tilde{H}_j'$ in the r.h.s of the second equality to see that that the two terms are equal. 
This symmetry allows us to simplify the expression to \eqref{eq:overlap:pert:ES2}.
The GKS inequalities \eqref{eq:GKS1}, \eqref{eq:GKS2} imply $0 \leq \<Q_1\>_{\underline{\tilde{H}}^j} - \<Q_1\>_{\underline{\tilde{H}}} \leq 1$. This allows us to relax \eqref{eq:overlap:pert:ES2} to 
\begin{align}
\mathbb{E}_{\underline{\tilde{H}}} \big [ \big ( \<Q_1\>_{\underline{\tilde{H}}} - \mathbb{E}_{\underline{\tilde{H}}} \<Q_1\>_{\underline{\tilde{H}}} \big )^2 \big ]
	& \leq \sum_{j=1}^{n} \mathbb{E}_{\underline{\tilde{H}}} \mathbb{E}_{\tilde{H}'_j} \big [ \big ( \<Q_1\>_{\underline{\tilde{H}}^j} - \<Q_1\>_{\underline{\tilde{H}}} \big ) \mathbb{I}(\tilde{H}'_j = \infty) \mathbb{I}(\tilde{H}_j = 0) \big ]. \label{eq:overlap:pert:ES3}
\end{align}

When $\tilde{H}'_j = \infty$ and $\tilde{H}_j = 0$, we have
\begin{align}
\< Q_1 \>_{\underline{\tilde{H}}^j} - \< Q_1 \>_{\underline{\tilde{H}}}
	& = \frac{1}{n} \sum_{i=1}^{n} \big ( \< \sigma_i \>_{\underline{\tilde{H}}^j} - \< \sigma_i \>_{\underline{\tilde{H}}} \big )
	\- = \frac{1}{n} \sum_{i=1}^{n} \big ( \frac{\< \sigma_i e^{\tilde{H}'_j (\sigma_j - 1)} \>_{\underline{\tilde{H}}}}{\< e^{\tilde{H}'_j (\sigma_j - 1)} \>_{\underline{\tilde{H}}}} - \< \sigma_i \>_{\underline{\tilde{H}}} \big ) \nonumber \\
	& = \frac{1}{n} \sum_{i=1}^{n} \big ( \frac{\< \sigma_i (1+\sigma_j) \>_{\underline{\tilde{H}}}}{1+\< \sigma_j \>_{\underline{\tilde{H}}}} - \< \sigma_i \>_{\underline{\tilde{H}}} \big ) \label{eq:overlap:pert:int1}
	\- = \frac{1}{n} \sum_{i=1}^{n} \int_0^1 d\tau \frac{d}{d\tau} \frac{ \< \sigma_i \>_{\underline{\tilde{H}}} + \tau \< \sigma_i \sigma_j \>_{\underline{\tilde{H}}} }{1+\tau \< \sigma_j \>_{\underline{\tilde{H}}}} \\
	& = \frac{1}{n} \sum_{i=1}^{n} \int_0^1 d\tau \frac{ \< \sigma_i \sigma_j \>_{\underline{\tilde{H}}} - \< \sigma_i \>_{\underline{\tilde{H}}} \< \sigma_j \>_{\underline{\tilde{H}}} }{(1+\tau \< \sigma_j \>_{\underline{\tilde{H}}})^2} 
	\- \leq \frac{1}{n} \sum_{i=1}^{n} (\< \sigma_i \sigma_j \>_{\underline{\tilde{H}}} - \< \sigma_i \>_{\underline{\tilde{H}}} \< \sigma_j \>_{\underline{\tilde{H}}}) \label{eq:overlap:pert:int2}
\end{align}
where the first equality of \eqref{eq:overlap:pert:int1} follows from the identity $e^{\tilde{H}'_j \sigma_j} \equiv \cosh \tilde{H}'_j (1+\sigma_j \tanh \tilde{H}'_j)$, and the last bound 
in \eqref{eq:overlap:pert:int2} uses the first GKS inequality \eqref{eq:GKS1}.
Substituting \eqref{eq:overlap:pert:int2} into \eqref{eq:overlap:pert:ES3}, we get
\begin{align*}
	& \int_{\varepsilon_0}^{\varepsilon_1} d\epsilon \mathbb{E} \big [ \big ( \<Q_1\>_{\underline{\tilde{H}}} - \mathbb{E}_{\underline{\tilde{H}}} \<Q_1\>_{\underline{\tilde{H}}} \big )^2 \big ] \nonumber \\
	& \leq \int_{\varepsilon_0}^{\varepsilon_1} d\epsilon \frac{1}{n} \sum_{i,j=1}^{n} \mathbb{E} \big [ (\< \sigma_i \sigma_j \>_{\underline{\tilde{H}}} - \< \sigma_i \>_{\underline{\tilde{H}}} \< \sigma_j \>_{\underline{\tilde{H}}}) \mathbb{I}(\tilde{H}'_j = \infty) \mathbb{I}(\tilde{H}_j = 0) \big ] \\
	& = \mathbb{P}(\tilde{H}'_j = \infty) \int_{\varepsilon_0}^{\varepsilon_1} d\epsilon \frac{1}{n} \sum_{i,j=1}^{n} \mathbb{E} \big [ (\< \sigma_i \sigma_j \>_{\underline{\tilde{H}}} - \< \sigma_i \>_{\underline{\tilde{H}}} \< \sigma_j \>_{\underline{\tilde{H}}}) \mathbb{I}(\tilde{H}_j = 0) \big ] \\
	& \leq \mathbb{P}(\tilde{H}'_j = \infty) \int_{\varepsilon_0}^{\varepsilon_1} d\epsilon \frac{1}{n} \sum_{i,j=1}^{n} \mathbb{E} \big [ (\< \sigma_i \sigma_j \>_{\underline{\tilde{H}}} - \< \sigma_i \>_{\underline{\tilde{H}}} \< \sigma_j \>_{\underline{\tilde{H}}}) \big ] \\
	& = \delta n^{1-\theta} \int_{\varepsilon_0}^{\varepsilon_1} d\epsilon \mathbb{E} \big [ \< (Q_1 - \< Q_1 \>_{\underline{\tilde{H}}} )^2 \>_{\underline{\tilde{H}}} \big ] \\
	& \leq \frac{3\delta}{n^{\theta}},
\end{align*}
using Lemma~\ref{thm:overlap:thermal:concentration} to yield the last inequality.
%
\subsection{Proof of Lemma~\ref{thm:overlap:graph:concentration}}
\label{sec:overlap:graph}
We write $H_{t,s; \epsilon}(\delta) \equiv H_{t,s; \epsilon,\delta}$ and $h_{t,s; \epsilon}(\delta) \equiv h_{t,s; \epsilon, \delta}$ to emphasize $\delta$ in this proof. For both quantities we 
have taken the expectation over $\underline{\tilde{H}}$ and therefore their derivatives w.r.t. $\delta$ are well-defined. \footnote{The proof 
here differs from the standard strategy in \cite{barbier2017phase,BarbierM17a} in the way that 
an extra parameter $\delta$ is required and $\sfx(\epsilon)$ has to be independent of $\delta$. This is because we need a 
well-defined derivative of free entropy such that we can obtain a controllable upper bound like in \eqref{eq:overlap:graph:rewrite-Q1}.}
From \eqref{eq:perturbed_f:derivative4a} and \eqref{eq:perturbed_f:derivative3a} we have
\begin{align}
| \mathbb{E}_{\underline{\tilde{H}}} \< Q_1 \>_{t,s; \epsilon, \delta} - \mathbb{E} \< Q_1 \>_{t,s; \epsilon, \delta} |
	& = \frac{n^{\theta}(1 - \frac{\delta}{n^{\theta}})}{\ln 2} \Big | \frac{d}{d\delta} H_{t,s; \epsilon}(\delta) - \frac{d}{d\delta} h_{t,s; \epsilon}(\delta) \Big | \\
	& \leq \frac{n^{\theta}}{\ln 2} \Big | \frac{d}{d\delta} H_{t,s; \epsilon}(\delta) - \frac{d}{d\delta} h_{t,s; \epsilon}(\delta) \Big |.
	\label{eq:overlap:graph:rewrite-Q1}
\end{align}
Recall that $H_{t,s; \epsilon}(\delta)$ and $h_{t,s; \epsilon}(\delta)$ are convex in $\delta$.
A standard lemma in Appendix~\ref{app:bound-by-convexity} then implies that for any $\xi>0$ we have
\begin{align}
\Big | \frac{d}{d\delta} H_{t,s; \epsilon}(\delta) - \frac{d}{d\delta} h_{t,s; \epsilon}(\delta) \Big |
\leq \xi^{-1} \sum_{u \in \{\delta - \xi, \xi, \delta + \xi\}} |H_{t,s;\epsilon}(u) - h_{t,s;\epsilon}(u)| + C_\xi^+(\delta) + C_\xi^-(\delta) \label{eq:overlap:graph:bd1}
\end{align}
where
\begin{align}
& C_\xi^+(\delta) \equiv \frac{d}{d \delta}h_{t,s; \epsilon}(\delta+\xi) - \frac{d}{d \delta}h_{t,s; \epsilon}(\delta) \geq 0,
&& C_\xi^-(\delta) \equiv \frac{d}{d \delta}h_{t,s; \epsilon}(\delta) - \frac{d}{d \delta}h_{t,s; \epsilon}(\delta-\xi) \geq 0.
\label{eq:C+-}
\end{align}
We substitute \eqref{eq:overlap:graph:bd1} into \eqref{eq:overlap:graph:rewrite-Q1}, then square both sides and apply $ ( \sum_{r=1}^{k} u_r )^2 \leq k \sum_{r=1}^{k} u_r^2$. The resulting inequality upon full expectation is written as
\begin{align}
\mathbb{E} \big [ ( \mathbb{E}_{\underline{\tilde{H}}} \< Q_1 \>_{t,s; \epsilon, \delta} - \mathbb{E} \< Q_1 \>_{t,s; \epsilon, \delta} )^2 \big ]
	& \leq \frac{5 n^{2\theta}}{( \xi \ln 2)^2} \sum_{u \in \{\delta - \xi, \xi, \delta + \xi\}} \mathbb{E} [ (H_{t,s;\epsilon}(u) - h_{t,s;\epsilon}(u))^2 ] \nonumber \\
	& \hspace{0.5cm} + \frac{5 n^{2\theta}}{(\ln 2)^2} \big ( (C_\xi^+(\delta))^2 + (C_\xi^-(\delta))^2 \big ). \label{eq:overlap:graph:bd2}
\end{align}

We now make use of a concentration result for the interpolated free entropy. 
In Appendix \ref{appendix:free-energy} we prove:
\begin{lemma}[Free entropy concentration]\label{thm:concentration:free_energy}
For any $s$ in $[0,1]$ and $t=1, \dots, T$ there is a constant $C > 0$ such that
\begin{align}
\mathbb{E} \big [ ({H}_{t,s;\epsilon,\delta} - {h}_{t,s;\epsilon,\delta} )^2 \big ] \leq \frac{{C}}{n}.
\label{eq:concentration:free_energy}
\end{align}
\end{lemma}

Using Lemma~\ref{thm:concentration:free_energy}, the first term on the r.h.s is found to be smaller than $15{C}/((\ln 2)^2 n^{1-2\theta}\xi^2)$. 
Next, using $\frac{\ln 2}{n^{\theta}} \leq \frac{dh_{t,s;\epsilon}(\delta)}{d\delta} \leq 0$
allows us to assert from \eqref{eq:C+-} that $|C_\xi^{\pm}(\delta)| \leq \frac{\ln 2}{n^{\theta}}$. Then using $C_\xi^{\pm}(\delta) \geq 0$
\begin{align}
\int_{\delta_0}^{\delta_1} d\delta \big ( C^{+}_\xi(\delta)^2 + C^{-}_\xi(\delta)^2 \big )
	& \leq \frac{\ln 2}{n^{\theta}} \int_{\delta_0}^{\delta_1} d\delta \big ( C^{+}_\xi(\delta) + C^{-}_\xi(\delta) \big ) 
	\label{nottight}
	\\
	& = \frac{\ln 2}{n^{\theta}} \big [ \big(h_{t,s;\epsilon}(\delta_1+\xi) - h_{t,s;\epsilon}(\delta_1-\xi) \big) \nonumber \\
	& \hspace{1cm} + \big(h_{t,s;\epsilon}(\delta_0-\xi) - h_{t,s;\epsilon}(\delta_0+\xi) \big) \big ] \\
	& \leq \frac{4 (\ln 2)^2  \xi}{n^{2\theta}}
\end{align}
where the mean value theorem has been used to get the last inequality. Thus when \eqref{eq:overlap:graph:bd2} is integrated over $\delta$ we obtain
\begin{align}
\int_{\delta_0}^{\delta_1} d\delta\, \mathbb{E} \big [ ( \mathbb{E}_{\underline{\tilde{H}}} \< Q_1 \>_{t,s; \epsilon, \delta} - \mathbb{E} \< Q_1 \>_{t,s; \epsilon, \delta} )^2 \big ]
	\leq \frac{15{C}(\delta_1-\delta_0)}{(\ln 2)^2 n^{1-2\theta}\xi^2} + 4 \xi.
\label{eq:overlap:graph:bd3}
\end{align}
The proof is ended by choosing $\xi$ such that $1 / (n^{1-2\theta}\xi^2) = \xi$, i.e., $\xi = n^{-(1-2\theta)/3}$ and $\theta\in (0, 1/2)$.

\section{Proofs of technical lemmas}
\label{sec:proof}

\subsection{Proof of Lemma \ref{pert-free}}
\label{sec:proofepsilon-free}

Here we consider the effect of removing the perturbation, so we shall assume $\frac{dx^{(t)}}{d\epsilon} = 0$. From formula \eqref{eq:perturbed_f:derivative1a} proved in section \ref{sec:usefulderivatives} we have
\begin{align}
\frac{d}{d\epsilon} h_{t,s;\epsilon,\delta} = \frac{d\mathbb{E}[\bar{\epsilon}^{(t,s)}]}{d\epsilon} \frac{d}{d\mathbb{E}[\bar{\epsilon}^{(t,s)}]} h_{t,s;\epsilon,\delta}
	= - \frac{\ln 2}{n(1-\epsilon)} \sum_{i=1}^n (1 - \mathbb{E}\<\sigma_i\>_{t,s;\epsilon,\delta})
	= - \frac{\ln 2}{n} \sum_{i=1}^{n}  ( 1 - \mathbb{E} \< \sigma_i \>_{t, s;\epsilon,\delta;\sim H_i} ).
\end{align}
where $\< \sigma_i \>_{t, s;\epsilon,\delta;\sim H_i}$ is the Gibbs expectation with fixed $H_i=0$.
Thus $\vert \frac{d}{d\epsilon} h_{t,s;\epsilon,\delta} \vert \leq \ln 2$.
We also remarked below equation \eqref{eq:perturbed_f:derivative3a} that $\vert \frac{d}{d\delta} h_{t,s;\epsilon,\delta}\vert \leq \ln (2) /n^{\theta}$. 
Thus 
by the mean value theorem
\begin{align}
 \vert h_{t,s;\epsilon,\delta} - h_{t,s;\epsilon=0,\delta}\vert & \leq \epsilon \ln 2\, , \\
 \vert h_{t,s;\epsilon,\delta} - h_{t,s;\epsilon, \delta=0}\vert & \leq \frac{\ln 2}{n^{\theta}}\,.
\end{align}
By triangle inequality we get \eqref{eq:pert-free:1}.
Note that $\tilde{h}_{\epsilon,\delta}(\underline{\sfx})$ is in the form $h_{T,1;\epsilon,\delta}(\underline{\sfx}) + g(\underline{\sfx})$, 
therefore 
$$
\tilde{h}_{\epsilon,\delta}(\underline{\sfx}) - \tilde{h}_{\epsilon=0,\delta=0}(\underline{\sfx}) = h_{T,1;\epsilon,\delta}(\underline{\sfx}) - h_{T,1;\epsilon=0,\delta=0}(\underline{\sfx})\,.
$$
Consquently \eqref{eq:pert-free:2} follows immediately from \eqref{eq:pert-free:1}.

\subsection{Proof of Lemma \ref{thm:BEC:overlap}}
\label{sec:proof:BEC:overlap}
The first GKS inequality \eqref{eq:GKS1} implies
\begin{align}
\< \sigma_A \>_{t,s; \epsilon,\delta} ( 1 - \< \sigma_A \>_{t,s; \epsilon;\delta} ) \geq 0. \label{eq:GKS:implication}
\end{align}
Moreover, Nishimori's identity \eqref{eq:Nishimori} implies
\begin{align}
\mathbb{E}  \< \sigma_A \>_{t,s; \epsilon,\delta}  = \mathbb{E} [ \< \sigma_A \>_{t,s; \epsilon,\delta}^2 ] \label{eq:Nishimori:case1}
\end{align}
which can be written as
\begin{align}
\mathbb{E} [ \< \sigma_A \>_{t,s; \epsilon,\delta} ( 1 - \< \sigma_A \>_{t,s; \epsilon,\delta} ) ] = 0. \label{eq:Nishimori:case1:2}
\end{align}
As a result of \eqref{eq:GKS:implication} and \eqref{eq:Nishimori:case1:2}, we have $\< \sigma_A \>_{t,s; \epsilon,\delta}$ equal to either $0$ or $1$.

\subsection{Proof of Lemma~\ref{thm:bd_fixed_t}}
\label{sec:proof:bd_fixed_t}
Using the fundamental theorem of calculus, the desired difference has an integral form 
\begin{align}
\mathbb{E} \< Q_1^k \>_{t,s'; \epsilon,\delta}  - &\mathbb{E} \< Q_1^k \>_{t,0; \epsilon,\delta} 
	 = \frac{1}{n^k} \int_{0}^{s'} ds \sum_{i_1,\dots,i_k=1}^{n} \frac{d}{ds}\mathbb{E} \< \sigma_{i_1} \cdots \sigma_{i_k} \>_{t,s;\epsilon_n}
	 \nonumber \\
	& = \frac{1}{n^k} \int_{0}^{s'} ds \sum_{i_1,\dots,i_k=1}^{n}\Big\{ \frac{\alpha K}{T} \sum_{j=1}^{n}\big( \mathbb{E}\< \sigma_{i_1} \cdots \sigma_{i_k} \>_{e_{j,s}^{(t)} + 1} - \mathbb{E}\< \sigma_{i_1} \cdots \sigma_{i_k} \>_{e_{j,s}^{(t)}} \big) \nonumber \\
	& - \frac{\alpha n}{T}\big( \mathbb{E}\< \sigma_{i_1} \cdots \sigma_{i_k} \>_{m_s^{(t)} + 1} - \mathbb{E}\< \sigma_{i_1} \cdots \sigma_{i_k} \>_{m_s^{(t)}}\big)\Big\} 
	\label{eq:bd_fixed_t:1}
\end{align}
where~\eqref{eq:bd_fixed_t:1} follows from the Poisson property \eqref{eq:dPoissonMean}.
Since $| \< \sigma_{i_1} \cdots \sigma_{i_k} \> |\le 1$ we see that the absolute value of \eqref{eq:bd_fixed_t:1} is bounded by $2 \alpha (K+1) n / T$.

\subsection{Proof of Lemma~\ref{thm:rmv-delta}}
\label{sec:proof:rmv-delta}
As indices $t,s,\epsilon$ are fixed in this proof, we omit them for concision.
Recall that~\eqref{eq:Q1:delta-der} and \eqref{eq:GKS2} imply $\mathbb{E}\<Q_1\>_{\delta}$ is increasing in $\delta$ and therefore $\mathbb{E}[\<Q_1\>_{\delta}] \geq \mathbb{E}[\<Q_1\>_{0}]$. Using also
$|\<Q_1\>| \leq 1$ we obtain the inequality
\begin{align*}
\mathbb{E}[\<Q_1\>_{\delta}]^K - \mathbb{E}[\<Q_1\>_{0}]^K 
	= (\mathbb{E}[\<Q_1\>_{\delta}] - \mathbb{E}[\<Q_1\>_{0}]) \sum_{k=0}^{K-1} \mathbb{E}[\<Q_1\>_{\delta}]^{K-k} \mathbb{E}[\<Q_1\>_{\delta}]^{k}
	\leq K(\mathbb{E}[\<Q_1\>_{\delta}] - \mathbb{E}[\<Q_1\>_{0}]).
\end{align*}
This inequality, together with $q_1^{(t)} \equiv \mathbb{E} \tanh V^{(t)} \in [0,1]$ and $\mathbb{E}[\<Q_1\>_{\delta}] \geq \mathbb{E}[\<Q_1\>_{0}]$, gives
\begin{align}
& \int_{\delta_0}^{\delta_1} d\delta \int_{\varepsilon_0}^{\varepsilon_1} d\epsilon \Big | \mathbb{E}[\<Q_1\>_{\delta}]^K - \mathbb{E}[\<Q_1\>_{0}]^K
- K (q_1^{(t)})^{K-1} ( \mathbb{E}[\<Q_1\>_{\delta}] - \mathbb{E}[\<Q_1\>_{0}] ) \Big | \nonumber \\
& \leq 2 K \int_{\delta_0}^{\delta_1} d\delta \int_{\varepsilon_0}^{\varepsilon_1} d\epsilon (\mathbb{E}[\<Q_1\>_{\delta}] - \mathbb{E}[\<Q_1\>_{0}]) \nonumber \\
& = 2 K \int_{\delta_0}^{\delta_1} d\delta \big ( \int_{\varepsilon_0}^{\varepsilon_1} d\epsilon \mathbb{E}[\<Q_1\>_{\delta}] - \int_{\varepsilon_0}^{\varepsilon_1} d\epsilon \mathbb{E}[\<Q_1\>_{0}] \big ). 
\label{93}
\end{align}
We use the mean value theorem to upper bound \eqref{93} as
\begin{align*}
2 K \int_{\delta_0}^{\delta_1} d\delta \ \delta \max_{\delta' \in [0,\delta]} \big ( \frac{d}{d\delta} \int_{\varepsilon_0}^{\varepsilon_1} d\epsilon \mathbb{E}[\<Q_1\>_{\delta}] \big ) \Big |_{\delta=\delta'}
& = 2 K \int_{\delta_0}^{\delta_1} d\delta \ \delta \max_{\delta' \in [0,\delta]} \big ( \int_{\varepsilon_0}^{\varepsilon_1} d\epsilon \frac{d}{d\delta} \mathbb{E}[\<Q_1\>_{\delta}] \big ) \Big |_{\delta=\delta'}
\end{align*}
where the equality follows from the fact that $\underline{\sfx}(\epsilon)$ is independent of $\delta$ and therefore we can exchange the order of derivative and integral. Using \eqref{eq:Q1:delta-der} the last equation equals
\begin{align}
& 2 K \int_{\delta_0}^{\delta_1} d\delta \ \delta \max_{\delta' \in [0,\delta]} \big ( \int_{\varepsilon_0}^{\varepsilon_1} d\epsilon \frac{1}{n^{1+\theta}(1-\delta/n^{\theta})} \sum_{i, j=1}^{n} \mathbb{E} [\< \sigma_i \sigma_j \>_{\delta} 
- \< \sigma_i \>_{\delta} \< \sigma_j \>_{\delta}] \big ) \Big |_{\delta=\delta'} \nonumber \\
& = 2 Kn^{1-\theta} \int_{\delta_0}^{\delta_1} d\delta \frac{\delta}{1-\delta/n^{\theta}} \max_{\delta' \in [0,\delta]} \big ( \int_{\varepsilon_0}^{\varepsilon_1} d\epsilon \mathbb{E}[\< ( Q_1-\<Q_1\>_{\delta} )^2 \>_{\delta}] \big ) \Big |_{\delta=\delta'} \nonumber \\
& \leq \frac{6K}{n^{\theta}} \int_{\delta_0}^{\delta_1} d\delta \frac{\delta}{1-\delta/n^{\theta}} \label{eq:rmv-delta:3} \\
& \leq \frac{3K(\delta_1^2-\delta_0^2)}{n^{\theta}(1-\delta_1/n^{\theta})} \nonumber
\end{align}
where \eqref{eq:rmv-delta:3} follows from Lemma~\ref{thm:overlap:thermal:concentration}.

\subsection{Proof of Lemma \ref{thm:momentMatching:solution}}
\label{sec:proof:momentMatching:solution}

For each $n$, we seek distributions $\sfx^{(t)}\in \mathcal{B}$ for $V^{(t)}$, $t=1,\dots, T$ which solve equation \eqref{eq:momentMatching}. 
By symmetry between vertices
$\mathbb{E}\langle Q_1 \rangle_{t,0; \epsilon,0} = \mathbb{E}\langle \sigma_1 \rangle_{t,0; \epsilon,0}$ so the equation becomes 
\begin{align}
 \mathbb{E}\tanh V^{(t)} = \mathbb{E}\langle \sigma_1 \rangle_{t,0; \epsilon,0}.
\end{align}
Recall that in our interpolation scheme the right hand side depends only on  $\{\sfx^{(t')}\}_{t' <t}$ and is thus independent of $\sfx^{(t)}$. Thus it suffices to choose 
$V^{(t)}$ for $t=1, \dots, T$ as follows: $V^{(t)} = +\infty$ with probability $\mathbb{E}\langle \sigma_1 \rangle_{t,0; \epsilon,0}$ and $V^{(t)}= 0$ with 
probability $1-\mathbb{E}\langle \sigma_1 \rangle_{t,0; \epsilon,0}$. These are the distributions $\hat \sfx^{(t)}_n \in \mathcal{B}$ of the Lemma. It is clear that this solution is unique and $\hat{x}_n^{(t)} = \mathbb{E} \tanh V^{(t)}$.

Finally, we verify $d \hat{x}_n^{(t)} / d\epsilon \geq 0$. To simplify the notation we use $\< - \> \equiv \< - \>_{t,0;\epsilon,0}$ and $\< - \>_{\sim \bar{H}_j^{(t,0)}}$ denotes $\< - \>_{t,0;\epsilon,0}$ with $\bar{H}_j^{(t,0)}$ set to 0. 
Recall the definition of $\mathbb{E}[\bar{\epsilon}^{(t,s)}]$ in section \ref{sec:usefulderivatives}. By the chain rule we have
\begin{align}
\frac{d\hat{x}^{(t)}}{d\epsilon}
	= \frac{d\mathbb{E}[\bar{\epsilon}^{(t,0)}]}{d\epsilon} \frac{d\mathbb{E}\<\sigma_1\>_{t,0;\epsilon,0}}{d \mathbb{E}[\bar{\epsilon}^{(t,0)}]}.
\end{align}
To compute the last derivative, first we use the identity $e^{\pm x} = (1 \pm \tanh x)\cosh x $ to write
\begin{align*}
\mathbb{E}\<\sigma_1\>_{t,0;\epsilon,0} 
	& = \mathbb{E} \Biggl[ \frac{ \<e^{\bar{H}^{(t,0)}_j (\sigma_j-1)} \sigma_1 \>_{\sim \bar{H}_j^{(t,0)}} }{ \< e^{\bar{H}^{(t,0)}_j (\sigma_j-1)} \>_{\sim \bar{H}_j^{(t,0)}} } \Biggr]
	 = \mathbb{E}\Biggl[ \frac{ \< \sigma_1 \>_{\sim \bar{H}_j^{(t,0)}} + \< \sigma_1\sigma_j \>_{\sim \bar{H}_j^{(t,0)}} \tanh \bar{H}_j^{(t,0)} }{ 1 + \< \sigma_j \>_{\sim \bar{H}_j^{(t,0)}} \tanh \bar{H}_j^{(t,0)}} \Biggr] \\
	& = \mathbb{E}[\bar{\epsilon}_j^{(t,0)}] \cdot \mathbb{E} 
	\Biggl[ \frac{ \< \sigma_1 \>_{\sim \bar{H}_j^{(t,0)}} + \< \sigma_1\sigma_j \>_{\sim \bar{H}_j^{(t,0)}} }{ 1 + \< \sigma_j \>_{\sim \bar{H}_j^{(t,0)}} } \Biggr]
	+ (1 - \mathbb{E}[\bar{\epsilon}_j^{(t,0)}]) \mathbb{E} \< \sigma_1 \>_{\sim \bar{H}_j^{(t,0)}}.
\end{align*}
Then it is straightforward to compute 
\begin{align}
\frac{d\mathbb{E}\<\sigma_1\>_{t,0;\epsilon,0}}{d \mathbb{E}[\bar{\epsilon}^{(t,0)}]}
	& = \sum_{j=1}^{n} \frac{d\mathbb{E}\<\sigma_1\>_{t,0;\epsilon,0}}{d \mathbb{E}[\bar{\epsilon}_j^{(t,0)}]}
	\- = \sum_{j=1}^{n} \mathbb{E} \Biggl[ \frac{ \< \sigma_1 \sigma_j \>_{\sim \bar{H}_j^{(t,0)}} - \< \sigma_1 \>_{\sim \bar{H}_j^{(t,0)}} \< \sigma_j \>_{\sim \bar{H}_j^{(t,0)}} }{ 1+\< \sigma_j \>_{\sim \bar{H}_j^{(t,0)}} } \Biggr]. \label{eq:der-Gibbs-barEp}
\end{align}
The second GKS inequality \eqref{eq:GKS2} ensures \eqref{eq:der-Gibbs-barEp} non-negative, leaving the sign of $dx^{(t)}/d\epsilon$ determined by $d\mathbb{E}[\bar{\epsilon}^{(t,0)}] / d\epsilon$. Using \eqref{eq:BP-update:BEC} and \eqref{eq:E-epsilon-bar}, 
\begin{align*}
& \frac{\mathbb{E}[\bar{\epsilon}^{(1,0)}]}{d\epsilon} = 1, 
&& \frac{\mathbb{E}[\bar{\epsilon}^{(t,0)}]}{d\epsilon} = \frac{1-\mathbb{E}[\bar{\epsilon}^{(t,0)}]}{1-\epsilon} + (1-\mathbb{E}[\bar{\epsilon}^{(t,0)}]) \frac{\alpha K(K-1)(1-q)}{T} \sum_{t'=1}^{t-1} \hat{x}^{(t') K-2} \frac{d\hat{x}^{(t')}}{d\epsilon}.
\end{align*}
This equation implies that the claim $d \hat{x}^{(t)} / d\epsilon \geq 0$ is true for $t=1$ by direct calculation.  Then we also get the claim for $t \geq 2$ by induction.
%
%

\subsection{Proof of Lemma \ref{equivalence}}
\label{proof:h_RS_coupled:equivalence}


We first note that the generalized entropy functionals can easily be shown to be upper bounded and are defined on a closed convex set of probability measures. Hence we can replace the supremum in 
the lemma by a maximum. For the BEC $\underline\sfx\in \mathcal{B}^T$ and $\tilde{h}_{\epsilon=0,\delta=0}(\underline{\sfx})$ becomes a function of $\underline x\in [0,1]^T$, therefore the proof of the lemma 
can be carried out directly by elementary real analysis calculations. 

Here we give an analysis that applies more generally to functionals over $\underline\sfx\in \mathcal{X}^T$ in the general case of symmetric channels.
Let us outline the strategy of the proof: (i) We first show 
that the stationarity condition for $\tilde{h}_{\epsilon=0,\delta=0}(\underline{\sfx})$ implies that all $\sfx^{(t)}$ are equal for $t=1, \dots, T$; (ii) 
We then show that a maximum of  $\tilde{h}_{\epsilon=0,\delta=0}(\underline{\sfx})$ 
is necessarily a stationary point.  


Before carrying out point (i) it is convenient to
 express $\tilde{h}_{\epsilon=0,\delta=0}(\underline \sfx)$ more explicitly in terms of the distribution $\underline{\sfx}$ thanks to a
formalism from coding theory (see e.g. \cite{KumYMP:2014,MCT:2008}). We define an entropy functional\footnote{The notation $H$ for the entropy should not be confused with the notation $H$ for the 
perturbation field in the model.} $H: \mathcal{X} \rightarrow \mathbb{R}$ as
\begin{align}
H(\sfx)\equiv \int \ln (1+e^{-2 a}) \sfx(da) = \ln 2 - \int \ln (1+ \tanh a) \sfx(da). \label{eq:entropy}
\end{align}
The argument $a$ is to be interpreted as a half-log-likelihood ratio. 
Two convolution operators $\circledast, \boxast: \sX \times \sX \rightarrow \sX$ are 
defined for $a_1 \sim \sfx_1, a_2 \sim \sfx_2$ such that $\sfx_1 \circledast \sfx_2$ is 
the distribution of $a_1+a_2$ and $\sfx_1 \boxast \sfx_2$ is the distribution of $\tanh^{-1} (\tanh a_1 \tanh a_2)$. 
Therefore, the entropies of convolutions are
\begin{align}
H(\circledast_{i=1}^k \sfx_i ) &  = \int \ln (1+e^{-2 \sum_{i=1}^k a_i}) \prod_{i=1}^k \sfx_i (da_i), \label{eq:entropy-conv1} \\
H(\boxast_{i=1}^k \sfx_i ) & = \ln 2 - \int \ln \Big( 1 + \prod_{i=1}^k \tanh a_i \Big ) \prod_{i=1}^k \sfx_i (da_i). \label{eq:entropy-conv2}
\end{align}
We define $\sfx^{\circledast 0} \equiv \Delta_0$, where $\Delta_0$ is the identity of $\circledast$ and it is a distribution with 
solely a point mass at 0. We also define $\Lambda^{\circledast}(\sfx) \equiv \sum_{l=0}^{\infty} \Lambda_l \sfx^{\circledast l}$, 
where $\Lambda_l = \frac{(\alpha K)^l}{l!} e^{-\alpha K}$ denotes the probability that a variable node has degree $l$, 
and $\lambda^{\circledast}(\sfx) \equiv  \sum_{l=1} \lambda_l \sfx^{\circledast (l-1)}$, 
where $\lambda_l = \frac{l \Lambda_l}{\Lambda'(1)} = \frac{(\alpha K)^{l-1}}{(l-1)!} e^{-\alpha K}$ denotes the probability that an edge is connected to a variable node of degree $l\geq 1$. 
One can check that (see Appendix~\ref{proof:RSFreeEntropy_coupled2})
\begin{align}
\tilde{h}_{\epsilon=0,\delta=0} \big ( \underline{\sfx} \big ) 
	& = - \alpha K H \Big( \frac{1}{T} \sum_{t=1}^{T} \sfc \boxast (\sfx^{(t)})^{\boxast (K-1)} \Big ) 
	+ H \Big ( \Lambda^{\circledast} \Big ( \frac{1}{T} \sum_{t=1}^{T} \sfc \boxast (\sfx^{(t)})^{\boxast (K-1)} \Big ) \Big ) \nonumber \\
	& \hspace{1cm} + \alpha H(\sfc) + \frac{\alpha(K-1)}{T} \sum_{t=1}^{T} H \big (\sfc \boxast (\sfx^{(t)})^{\boxast K} \big ).
\label{eq:RSFreeEntropy_coupled2}
\end{align}
We will need differentiation rules for functionals. The directional (or Gateaux) derivative of a functional\footnote{To have well defined directional derivatives it is 
understood that we extend the space $\sX$ to the Banach space of signed probability measures over $\bar{\mathbb{R}}$.}
$F: \sfx\in \mathcal{X}\to \mathbb{R}$ at point $\sfx\in \mathcal{X}$ in 
the direction $\eta = \sfx_2 -\sfx_1$ 
where $\sfx_1,\sfx_2\in \mathcal{X}$ is by definition the following linear functional of $\eta$:
\begin{align*}
dF(\sfx) [\eta] \equiv \lim_{\gamma \rightarrow 0} \frac{F(\sfx + \gamma \eta) - F(\sfx)}{\gamma}.
\end{align*}
We employ the following computational rules that are easily proved for linear functionals $F$:
\begin{lemma}[{\cite[Propositions 14 and 15]{KumYMP:2014}}]
Let $F: \sX \rightarrow \mathbb{R}$ be a linear functional, and $*$ be 
either $\circledast$ or $\boxast$. Then for $k\geq 1$ integer, $\sfx, \sfx_1, \sfx_2 \in \sX$, and setting $\eta = \sfx_2 - \sfx_1$, we have
\begin{align*}
d F(\sfx^{* k})[\eta] = k F(\sfx^{* (k-1)} * \eta).
\end{align*}
For any polynomials $p,q$, we have
\begin{align*}
d F \big (p^{\circledast}(q^{\boxast}(\sfx)) \big )[\eta] = F \big ( p'^{\circledast}(q^{\boxast}(\sfx) \circledast (q'^{\boxast}(\sfx) \boxast \eta )\big ).
\end{align*}
where $p'$ and $q'$ are the derivatives of the polynomials.
\label{thm:entropyFunctional:derivative}
\end{lemma}

\begin{lemma}[{\cite[Theorem 4.41]{MCT:2008}}]
For any $\sfx_1, \sfx_2, \sfx_3 \in \sX$, we have
\begin{align*}
H((\sfx_1 - \sfx_2)\circledast \sfx_3) + H((\sfx_1 - \sfx_2) \boxast \sfx_3) = H(\sfx_1 - \sfx_2).
\end{align*}
\label{thm:entropyFunctional:duality}
\end{lemma}

We can now proceed to prove (i). 
Fix $t\in \{1, \cdots, T\}$. Consider the functional $\tilde h_{\epsilon=0, \delta=0}(\underline \sfx)$ as functional w.r.t its $t$-th component only. We denote
$d_{t} \tilde h_{\epsilon=0, \delta=0}(\underline\sfx) [\eta^{(t)}]$, the directional derivative of this functional
 at the point $\sfx$ with respect to its $t$-th component (only) in the direction 
$\eta^{(t)}$. This a linear functional of $\eta^{(t)}$ and corresponds to a ``partial'' Gateaux derivative as indicated by the notation $d_t$.
Let $$
\sfT(\sfc, \underline{\sfx}) \equiv \lambda^{\circledast} \big( \frac{1}{T} \sum_{t=1}^{T} \sfc \boxast \sfx^{(t)\boxast (K-1)} \big)\,.
$$
Using Lemma~\ref{thm:entropyFunctional:derivative}, $d_{t} \tilde{h}_{\epsilon=0,\delta=0}(\underline{\sfx})[\eta^{(t)}]$ is the sum of the following three terms:
\begin{align}
- \alpha K d_{t}  H \Big ( \frac{1}{T} \sum_{t=1}^{T} \sfc \boxast \sfx^{(t)\boxast (K-1)} \Big) [\eta^{(t)}]&= 
- \frac{\alpha K(K-1)}{T} H \big ( \sfc \boxast \sfx^{(t)\boxast (K-2)} \boxast \eta^{(t)} \big ), \label{eq:dir_derivative:term1}\\
	d_{t} H \Big (\Lambda^{\circledast}\Big ( \frac{1}{T} \sum_{t=1}^{T} 
	\sfc \boxast \sfx^{(t)\boxast (K-1)} \Big) \Big ) [\eta^{(t)}] &= 
	\frac{\alpha K(K-1)}{T} H \big ( \sfT(\sfc, \underline{\sfx}) 
	\circledast \big ( \sfc \boxast \sfx^{(t)\boxast (K-2)} \boxast \eta^{(t)} \big ) \big ), \label{eq:dir_derivative:term2}\\
- \alpha(K-1)d_{t}   H \Big (\frac{1}{T}\sum_{t=1}^{T}\sfc \boxast (\sfx^{(t)\boxast K} \Big ) [\eta^{(t)}]
	& = - \frac{\alpha K(K-1)}{T} H \big( \sfc \boxast \sfx^{(t)\boxast (K-1)} \boxast \eta^{(t)} \big ). \label{eq:dir_derivative:term3}
\end{align}
In addition, we use Lemma~\ref{thm:entropyFunctional:duality} to rewrite \eqref{eq:dir_derivative:term2} as
\begin{align}
\frac{\alpha K(K-1)}{T} \left \{ H \left ( \sfc \boxast \sfx^{(t)\boxast (K-2)} \boxast \eta^{(t)} \right ) - 
H \left ( \sfT(\sfc, \underline{\sfx}) \boxast \big ( \sfc \boxast \sfx^{(t)\boxast (K-2)} \boxast \eta^{(t)} \big ) \right )  \right \}. \label{eq:dir_derivative:term2a}
\end{align}
Putting \eqref{eq:dir_derivative:term1}, \eqref{eq:dir_derivative:term3} and \eqref{eq:dir_derivative:term2a} together, we have
\begin{align}
d_{t} \tilde{h}_{\epsilon=0,\delta=0}(\underline{\sfx})[\eta^{(t)}]
	& = \frac{\alpha K(K-1)}{T} H \left ( \big(\sfx^{(t)} - \sfT(\sfc, \underline{\sfx}) \big) \boxast \big( \sfc \boxast (\sfx^{(t)\boxast (K-2)} \boxast \eta^{(t)} \big) \right ),
	\label{eq:dhcoupled_dxt}
\end{align}
which implies that $\underline \sfx$ is a stationnary point of $\tilde{h}_{\epsilon=0,\delta=0}(\underline{\sfx})[\eta^{(t)}]$ if and only if it satisfies the equation 
\begin{align}\label{stat-point-equ}
 \sfx^{(t)} = \sfT(\sfc, \underline{\sfx}), \qquad t=1, \dots, T.
\end{align}
In particular we have $\sfx^{(1)} = \dots = \sfx^{(T)}$ as claimed.

Now we prove (ii). We proceed by contradiction and show that: if $\underline \sfx$ is not a stationary point then it cannot be a maximum.
From the Taylor expansion of the logarithm and \eqref{eq:channelSymmetry} we find for any $\sfx\in \mathcal{X}$ 
\begin{align}
H(\sfx) 
	& = \ln 2 - \sum_{p=1}^{\infty} \frac{(-1)^{p+1}}{p} \int \sfx(da) (\tanh a)^p 
	\- = \ln 2 - \sum_{p=1}^{\infty} \frac{1}{2p(2p-1)} \int \sfx(da) (\tanh a)^{2p} \, .\label{eq:entropy:expansion}
\end{align}
Let $\sfx_1, \sfx_2, \sfx_3, \sfx_4 \in \sX$. From \eqref{eq:entropy-conv2} and \eqref{eq:entropy:expansion} 
\begin{align}
H((\sfx_1-\sfx_2) \boxast (\sfx_3-\sfx_4)) = - \sum_{p=1}^{\infty} \frac{1}{2p(2p-1)} \bigg\{ \int (\sfx_1-\sfx_2)(da) (\tanh a)^{2p} \bigg\} \bigg\{ \int (\sfx_3-\sfx_4)(da) (\tanh a)^{2p} \bigg\}
\end{align}
which implies that \eqref{eq:dhcoupled_dxt} can be written as
\begin{align}
d_{t} \tilde{h}_{\epsilon=0,\delta=0}(\underline{\sfx})[\eta^{(t)}]
	& = - \frac{\alpha K(K-1)}{T} \sum_{p=1}^{\infty} \frac{1}{2p(2p-1)} \bigg\{ \int ( \sfx^{(t)} - \sfT(\sfc, \underline{\sfx}) )(da) (\tanh a)^{2p} \bigg\} \nonumber \\
	& \hspace{1cm} \bigg\{ \int ( \sfc \boxast \sfx^{(t)\boxast (K-2)} )(da) (\tanh a)^{2p} \bigg\} \bigg\{ \int \eta^{(t)}(da) (\tanh a)^{2p} \bigg\}. \label{eq:dhcoupled_dxt:expansion}
\end{align}
Now, take an $\underline\sfx$ that is not a stationary point. Then there must exist an $t^*$ such that $\sfx^{(t^*)} \neq \sfT(\sfc, \underline{\sfx})$. Hence we can look at the directional derivative in the 
non-trivial direction
$\eta^{(t^*)} = \sfx^{(t^*)} - \sfT(\sfc, \underline{\sfx})$. From \eqref{eq:dhcoupled_dxt:expansion} we see that
\begin{align*}
d_{t^{*}} \tilde{h}_{\epsilon=0,\delta=0}(\underline{\sfx})[\eta^{(t^*)}]
	= &\frac{\alpha K(K-1)}{T} \sum_{p=1}^{\infty} \frac{1}{2p(2p-1)} \bigg\{ \int ( \sfc \boxast \sfx^{(t)\boxast (K-2)})(da) (\tanh a)^{2p} \bigg\} 
	\nonumber \\ &
	\times
	\bigg\{ \int ( \sfx^{(t^*)} - \sfT(\sfh, \sfc, \underline{\sfx}) )(da) (\tanh a)^{2p} \bigg\}^2
\end{align*}
so the directional derivative is strictly positive. Hence $\underline\sfx$ cannot be a maximum since there exists one direction in which the functional increases.

\appendix
%


\section{Direct proof of identity \eqref{eq:channelSymmetry} for symmetric distributions}
\label{proof:equi_channel_symmetry}
If $\sfx(-dh) = e^{-2h} \sfx(dh)$ holds, then we have
\begin{align*}
&\int_{-\infty}^{\infty} (\tanh h)^{2k-1} \, \sfx(dh) 
	 = \int_{0}^{\infty} (\tanh h)^{2k-1} \, \sfx(dh) - \int_{0}^{\infty} (\tanh h)^{2k-1} \, \sfx(-dh)  \\
	 = \,&\int_{0}^{\infty} (\tanh h)^{2k-1} (1-e^{-2h}) \, \sfx(dh)   = \int_{0}^{\infty} (\tanh h)^{2k} (1+e^{-2h}) \, \sfx(dh)  \\
	 = \,&\int_{0}^{\infty} (\tanh h)^{2k} \, \sfx(dh) + \int_{0}^{\infty} (\tanh h)^{2k} \, \sfx(-dh)  = \int_{-\infty}^{\infty} (\tanh h)^{2k} \, \sfx(dh) .
\end{align*}

%
%

\section{Rewriting the replica formula: Proof of \eqref{eq:RSFreeEntropy_coupled2}}
\label{proof:RSFreeEntropy_coupled2}
We copy again 
\begin{align*}
\tilde{h}_{\epsilon,\delta} \big ( \underline{\sfx} \big ) 
	 = \,&\mathbb{E} \Big [ \ln \Big( \prod_{t=1}^{T} \prod_{b=1}^{l} (1 + \tanh U_b^{(t)} ) + e^{-2(H+\tilde{H})} \prod_{t=1}^{T} \prod_{b=1}^{l} (1 - \tanh U_b^{(t)} ) \Big) \nonumber \\
	& \qquad- \frac{\alpha(K-1)}{T} \sum_{t=1}^{T} \ln \Big( 1 + \tanh \tilde{J} \prod_{i=1}^{K} \tanh V_i^{(t)} \Big ) 
	- \alpha  \ln ( 1 + \tanh \tilde{J}  )\Big] .
\end{align*}
The first term can be rewritten as
\begin{align*}
	& \mathbb{E} \ln \Big( \prod_{t=1}^{T} \prod_{b=1}^{l} (1 + \tanh U_b^{(t)} ) + e^{-2(H+\tilde{H})} \prod_{t=1}^{T} \prod_{b=1}^{l} (1 - \tanh U_b^{(t)} ) \Big) \\
	=\,& \mathbb{E} \ln \Big( \prod_{t=1}^{T} \prod_{b=1}^{l} (1 + \tanh U_b^{(t)} ) \Big) + \mathbb{E} \ln \Big( 1 + e^{-2(H+\tilde{H})} \prod_{t=1}^{T} \prod_{b=1}^{l} \frac{1 - \tanh U_b^{(t)} }{1 + \tanh U_b^{(t)} } \Big) \\
	=\,& \mathbb{E} \ln \Big( \prod_{t=1}^{T} \prod_{b=1}^{l} (1 + \tanh U_b^{(t)} ) \Big) + \mathbb{E} \ln \Big(1 + e^{-2  ( \sum_{t=1}^{T} \sum_{b=1}^{l} U_b^{(t)} + H + \tilde{H}) } \Big) \\
	=\,& - \alpha K H \Big( \frac{1}{T} \sum_{t=1}^{T} \sfc \boxast (\sfx^{(t)})^{\boxast (K-1)} \Big) + \alpha K \ln 2 + H \Big( \sfh \circledast \Lambda^{\circledast} \Big( \frac{1}{T} \sum_{t=1}^{T} \sfc \boxast (\sfx^{(t)})^{\boxast (K-1)} \Big ) \Big ).
\end{align*}
The second term can be easily seen to be equal to
\begin{align*}
- \frac{\alpha(K-1)}{T} \sum_{t=1}^{T} \ln ( 1 + \tanh \tilde{J} \prod_{i=1}^{K} \tanh V_i^{(t)} )
	 =&\, \frac{\alpha(K-1)}{T} \sum_{t=1}^{T} H \big (\sfc \boxast (\sfx^{(t)})^{\boxast K} \big ) - \alpha(K-1) \ln 2.
\end{align*}
The remaining term is
\begin{align*}
- \alpha \mathbb{E} \ln ( 1 + \tanh \tilde{J} )
	= \alpha (H(\sfc) -  \ln 2).
\end{align*}

%

\section{Derivatives of the conditional entropy: Proof of \eqref{eq:perturbed_f:derivative1a}--\eqref{eq:perturbed_f:derivative3b}}
\label{app:entropy-derivatives}
A large part of this appendix is an adaptation of \cite{Mac:2007,KudM:2009}. We recall that $h_{t,s;\epsilon,\delta} = n^{-1}\mathbb{E}\ln \bar{\mathcal{Z}}_{t, s; \bar{\mathcal{E}}}$ where 
$\bar{\mathcal{Z}}_{t, s; \bar{\mathcal{E}}}$ is the partition function associated to the hamiltonian 
 \eqref{eq:Htse-abstract}. Therefore, as explained in section \eqref{sec:thermal-and-quenched}, the free entropy only depends on $(\epsilon, \delta)$ through the combination \eqref{eq:E-epsilon-bar}, 
 with an explicit dependence as well as (possibly) an implicit one through the choice of $\underline \sfx$. 
 To alleviate the notations in this appendix we  drop the subscripts $t,s;\bar{\mathcal{E}}$ in the Gibbs brackets.

\subsection{Proof of \eqref{eq:perturbed_f:derivative1a}}
\label{app:entropy-derivatives:1a}
Let $\sH_{t,s; \bar{\mathcal{E}}}^{\sim i}(\underline{\sigma}, \underline{\tilde{J}}, \underline{\bar{H}})$ be the Hamiltonian
$\sH_{t,s; \bar{\mathcal{E}}}(\underline{\sigma}, \underline{\tilde{J}}, \underline{\bar{H}})$ with $\bar{H}_i^{(t,s)}=0$. 
Let ${\cal Z}_{t,s; \bar{\mathcal{E}}}^{\sim i}$ and $\< - \>_{\sim i}$ be the partition function and the Gibbs expectation 
associated with $\sH_{t,s; \bar{\mathcal{E}}}^{\sim i}(\underline{\sigma}, \underline{\tilde{J}}, \underline{\bar{H}})$.
The identities
\begin{align}
\ln \biggl\{\frac{{\cal Z}_{t,s; \bar{\mathcal{E}}}}{{\cal Z}_{t,s; \bar{\mathcal{E}}}^{\sim i}} \biggr\}&= \ln \< e^{\bar{H}_i^{(t,s)} (\sigma_i-1)} \>_{\sim i}, \nonumber\\
e^{\bar{H}_i^{(t,s)} (\sigma_i-1)} & = \frac{1 + \sigma_i\tanh \bar{H}_i^{(t,s)} }{1+\tanh \bar{H}_i^{(t,s)}},
\label{eq:entropy-derivatives1}
\end{align}
imply
\begin{align}
h_{t,s; \epsilon,\delta}
	& = \frac{1}{n} \mathbb{E} \ln {\cal Z}_{t,s; \bar{\mathcal{E}}}^{\sim i} 
	+ \frac{1}{n} \mathbb{E}  \ln \biggl\{\frac{1+\< \sigma_i \>_{\sim i}\tanh \bar{H}_i^{(t,s)}}{1+\tanh \bar{H}_i^{(t,s)}}\biggr\}.  \label{eq:entropy-derivatives2}
\end{align}
As $\tanh \bar{H}_i^{(t,s)}$ and $\< \sigma_i \>_{\sim i}$ equal either $0$ or $1$, \eqref{eq:entropy-derivatives2} simplifies to
\begin{align}
h_{t,s; \epsilon,\delta}
	& = \frac{1}{n} \mathbb{E} \ln {\cal Z}_{t,s; \bar{\mathcal{E}}}^{\sim i} - \frac{1}{n} \mathbb{E} [ \bar{\epsilon}_i^{(t,s)} ] \ln 2 \left ( 1 - \mathbb{E} \< \sigma_i \>_{\sim i} \right ).
\label{eq:entropy-derivatives2a}
\end{align}
Therefore, we have
\begin{align}
\frac{d}{d\mathbb{E}[\bar{\epsilon}^{(t,s)}]} h_{t,s; \epsilon, \delta} 
	& = \sum_{i=1}^{n} \frac{d}{d\mathbb{E}[\bar{\epsilon}_i^{(t,s)}]} h_{t,s; \bar{\mathcal{E}}}\biggr\vert_{ \mathbb{E}[\bar{\epsilon}_1^{(t,s)}]=\dots=\mathbb{E}[\bar{\epsilon}_n^{(t,s)}]= \mathbb{E}[\bar{\epsilon}^{(t,s)}]}
	= -\frac{\ln 2}{n} \sum_{i=1}^{n} \left ( 1 - \mathbb{E} \< \sigma_i \>_{\sim i} \right ), \label{1_25}
\end{align}
which is the first equality in \eqref{eq:perturbed_f:derivative1a}. 

To obtain the second equality, simply notice that as $1 - \< \sigma_i \>_{t,s;\epsilon, \delta} = 0$ when $\bar{H}_i^{(t,s)}=+\infty$ 
(which happens with probability $\bar{\epsilon}_i^{(t,s)}$). Performing the expectation over $\bar{H}_i^{(t,s)}$ in the following expression we get 
\begin{align}
1-\mathbb{E}\<\sigma_i\>_{t,s;\epsilon,\delta} 
	& = \mathbb{E} [ 1 - \<\sigma_i\>_{t,s;\epsilon,\delta} ] \nonumber \\
	& = \mathbb{E} [ (1-\bar{\epsilon}_i^{(t,s)})(1-\mathbb{E}\<\sigma_i\>_{\sim i}) + \bar{\epsilon}_i^{(t,s)} (1 - \mathbb{E}\<\sigma_i\>_{\bar{H}_i^{(t,s)}=\infty}) ] \nonumber \\
	& = (1-\mathbb{E}[\bar{\epsilon}^{(t,s)}])(1-\mathbb{E}\<\sigma_i\>_{\sim i}). \label{eq:recover-original-gibbs}
\end{align}
Replacing in \eqref{1_25} yields the second equality in \eqref{eq:perturbed_f:derivative1a}.

\subsection{Proof of \eqref{eq:perturbed_f:derivative1b}}
\label{app:entropy-derivatives:1b}
Let $\sH_{t,s; \bar{\mathcal{E}}}^{\sim i,j}(\underline{\sigma}, \underline{\tilde{J}}, \underline{\bar{H}})$ be the Hamiltonian
$\sH_{t,s; \bar{\mathcal{E}}}(\underline{\sigma}, \underline{\tilde{J}}, \underline{\bar{H}})$ with $\bar{H}_i^{(t,s)}=\bar{H}_j^{(t,s)}=0$. Let ${\cal Z}_{t,s; \bar{\mathcal{E}}}^{\sim i,j}$ and $\< - \>_{\sim i,j}$ be the partition function and the Gibbs expectation associated with $\sH_{t,s; \bar{\mathcal{E}}}^{\sim i,j}(\underline{\sigma}, \underline{\tilde{J}}, \underline{\bar{H}})$.
Using again \eqref{eq:entropy-derivatives1} on the identity
\begin{align*}
\ln \biggl\{\frac{{\cal Z}_{t,s;\epsilon,\delta}}{{\cal Z}_{t,s; \bar{\mathcal{E}}}^{\sim i,j}} \biggr\}
	& = \ln \<e^{\bar{H}_i^{(t,s)}(\sigma_i-1) + \bar{H}_j^{(t,s)}(\sigma_j-1)} \>_{\sim i, j},
\end{align*}
we have
\begin{align}
h_{t,s; \bar{\mathcal{E}}} 
	& = \frac{1}{n} \mathbb{E} \ln {\cal Z}_{t,s; \bar{\mathcal{E}}}^{\sim i,j} 
\nonumber\\
&\qquad+ \frac{1}{n} \mathbb{E}  \ln \biggl\{\frac{1 + \< \sigma_i \>_{\sim i,j}\tanh \bar{H}_i^{(t,s)}  + \< \sigma_j \>_{\sim i,j}\tanh \bar{H}_j^{(t,s)}  
+ \< \sigma_i \sigma_j \>_{\sim i,j}\tanh \bar{H}_i^{(t,s)} \tanh \bar{H}_j^{(t,s)} }{1 + \tanh \bar{H}_i^{(t,s)} + \tanh \bar{H}_j^{(t,s)} + \tanh \bar{H}_i^{(t,s)} \tanh \bar{H}_j^{(t,s)}}\biggr\} \nonumber \\
	& = \frac{1}{n} \mathbb{E} \ln {\cal Z}_{t,s; \underline{\epsilon}}^{\sim i,j} 
+ \frac{\mathbb{E}[\bar{\epsilon}_i^{(t,s)} ] \mathbb{E}[\bar{\epsilon}_j^{(t,s)} ]}{n} \mathbb{E} \ln \biggl\{\frac{1 + \< \sigma_i \>_{\sim i,j} + \< \sigma_j \>_{\sim i,j} + \< \sigma_i \sigma_j \>_{\sim i,j}}{4}\biggr\}
\nonumber \\
	& \qquad + \frac{\mathbb{E}[\bar{\epsilon}_i^{(t,s)} ] (1-\mathbb{E}[\bar{\epsilon}_j^{(t,s)} ])}{n} \mathbb{E} \ln \biggl\{\frac{1 + \< \sigma_i \>_{\sim i,j}}{2}\biggr\}
	+ \frac{(1-\mathbb{E}[\bar{\epsilon}_i^{(t,s)} ]) \mathbb{E}[\bar{\epsilon}_j^{(t,s)} ]}{n} \mathbb{E} \ln \biggl\{\frac{1 + \< \sigma_j \>_{\sim i,j}}{2}\biggr\},
\label{eq:entropy-derivatives3}
\end{align}
where \eqref{eq:entropy-derivatives3} follows from taking the expectation over $\bar{H}_i^{(t,s)}$ and $\bar{H}_j^{(t,s)}$. From \eqref{eq:entropy-derivatives2a} 
one can deduce that $\frac{d^2}{d \mathbb{E}[\bar{\epsilon}_i^{(t,s)}]^2} h_{t,s; \epsilon,\delta} = 0$. Therefore
\begin{align*}
  \frac{d^2}{d \mathbb{E}[\bar{\epsilon}^{(t,s)}]^2} h_{t,s; \epsilon, \delta}
	& =  \sum_{i,j=1}^{n} \frac{d^2}{d \mathbb{E}[\bar{\epsilon}_j^{(t,s)}] d \mathbb{E}[\bar{\epsilon}_i^{(t,s)}]} h_{t,s; \bar{\mathcal{E}}}\biggr\vert_{\epsilon_1=\dots=\epsilon_n=\epsilon} 
	\nonumber \\ &
	=  \sum_{i \neq j} \frac{d^2}{d\mathbb{E}[\bar{\epsilon}_j^{(t,s)}] d\mathbb{E}[\bar{\epsilon}_i^{(t,s)}]} 
	h_{t,s; \bar{\mathcal{E}}}\biggr\vert_{\mathbb{E}[\bar{\epsilon}_i^{(t,s)}]=\dots=\mathbb{E}[\bar{\epsilon}_n^{(t,s)}] = \mathbb{E}[\bar{\epsilon}^{(t,s)}]}.
\end{align*}
The derivatives $\frac{d^2}{d \mathbb{E}[\bar{\epsilon}_j^{(t,s)}] d \mathbb{E}[\bar{\epsilon}_i^{(t,s)}]} h_{t,s; \epsilon,\delta}$ 
can be readily obtained from \eqref{eq:entropy-derivatives3}. This provides
\begin{align}
\frac{d^2}{d \mathbb{E}[\bar{\epsilon}^{(t,s)}]^2} h_{t,s; \epsilon, \delta}
	& = \frac{1}{n} \sum_{i \neq j} 
	\mathbb{E}\ln \biggl\{\frac{1 + \< \sigma_i \>_{\sim i,j} + \< \sigma_j \>_{\sim i,j} + \< \sigma_i \sigma_j \>_{\sim i,j}}{1 + \< \sigma_i \>_{\sim i,j} 
	+ \< \sigma_j \>_{\sim i,j} + \< \sigma_i \>_{\sim i,j} \< \sigma_j \>_{\sim i,j}}\biggr\} . \label{eq:entropy-derivatives6a}
\end{align} 
We now simplify each term in the sum \eqref{eq:entropy-derivatives6a}.
Given that $\< \sigma_S \>_{\sim i,j}$ equals either $0$ or $1$ for any subsets $S \subset \{ 1 \dots n\}$, one can verify that the numerator 
and denominator of \eqref{eq:entropy-derivatives6a} can be written as
\begin{align}
	& \hspace{-0.5cm} \ln \bigl( 1 + \< \sigma_i \>_{\sim i,j} + \< \sigma_j \>_{\sim i,j} + \< \sigma_i \sigma_j \>_{\sim i,j} \bigr)
	=  \bigl(\< \sigma_i \>_{\sim i,j} + \< \sigma_j \>_{\sim i,j} + \< \sigma_i \sigma_j \>_{\sim i,j}\bigr) \ln 2
	\nonumber \\
	& +\bigl(\< \sigma_i \>_{\sim i,j}\< \sigma_j \>_{\sim i,j} + \< \sigma_i \>_{\sim i,j} \< \sigma_i \sigma_j \>_{\sim i,j} + \< \sigma_j \>_{\sim i,j} \< \sigma_i \sigma_j \>_{\sim i,j}\bigr) 
	\bigl(\ln 3 - 2\ln 2\bigr)
	\nonumber \\
	& 
	+  \< \sigma_i \>_{\sim i,j} \< \sigma_j \>_{\sim i,j} \< \sigma_i \sigma_j \>_{\sim i,j} \bigl(5\ln2 - 3\ln 3\bigr),
	\label{eq:entropy-derivatives7}
\end{align}
and
\begin{align*}
\ln \bigl( 1 + \< \sigma_i \>_{\sim i,j} + \< \sigma_j \>_{\sim i,j} + \< \sigma_i \>_{\sim i,j} \< \sigma_j \>_{\sim i,j} \bigr)
	& =  \bigl(\< \sigma_i \>_{\sim i,j} + \< \sigma_j \>_{\sim i,j}\bigr) \ln 2.
\end{align*}
Special cases of the Nishimori identities \eqref{eq:Nishimori}, 
\begin{align*}
\mathbb{E}[\< \sigma_i \>_{\sim i,j} \< \sigma_j \>_{\sim i,j}] 
	& = \mathbb{E}[\< \sigma_i \>_{\sim i,j} \< \sigma_j \>_{\sim i,j} \< \sigma_i \sigma_j \>_{\sim i,j}], \\
\mathbb{E}[\< \sigma_i \>_{\sim i,j} \< \sigma_i \sigma_j \>_{\sim i,j}]
	& = \mathbb{E}[\< \sigma_i \>_{\sim i,j} \< \sigma_j \>_{\sim i,j} \< \sigma_i \sigma_j \>_{\sim i,j}], \\
\mathbb{E}[\< \sigma_j \>_{\sim i,j} \< \sigma_i \sigma_j \>_{\sim i,j}]
	& = \mathbb{E}[\< \sigma_i \>_{\sim i,j} \< \sigma_j \>_{\sim i,j} \< \sigma_i \sigma_j \>_{\sim i,j}],
\end{align*}
can now be used to simplify \eqref{eq:entropy-derivatives7} so that each term in the sum \eqref{eq:entropy-derivatives6a} becomes
\begin{align}
\ln (2) \, \mathbb{E} [\< \sigma_i \sigma_j \>_{\sim i,j} - \< \sigma_i \>_{\sim i,j} \< \sigma_j \>_{\sim i,j} ].
\label{eq:entropy-derivatives8}
\end{align}
Moreover, as $\< \sigma_i \sigma_j \> - \< \sigma_i \> \< \sigma_j \> = 0$ when $\bar{H}_i^{(t,s)}$ and/or $\bar{H}_j^{(t,s)}$ equal $+\infty$, we obtain
\begin{align}
\mathbb{E} [\< \sigma_i \sigma_j \> - \< \sigma_i \> \< \sigma_j \> ]
	& = (1- \mathbb{E}[\bar{\epsilon}_i^{(t,s)}]) (1-\mathbb{E}[\bar{\epsilon}_j^{(t,s)}]) \mathbb{E} [\< \sigma_i \sigma_j \>_{\sim i,j} - \< \sigma_i \>_{\sim i,j} \< \sigma_j \>_{\sim i,j} ].
	\label{eq:entropy-derivatives9}
\end{align}
Finally, from \eqref{eq:entropy-derivatives6a}, \eqref{eq:entropy-derivatives8} and \eqref{eq:entropy-derivatives9} we obtain \eqref{eq:perturbed_f:derivative1b}.

\subsection{Derivation of \eqref{eq:perturbed_f:derivative4a} and \eqref{eq:perturbed_f:derivative4b}}
The derivation of \eqref{eq:perturbed_f:derivative4a} is the same as Sec.~\ref{app:entropy-derivatives:1a} except that the steps should be done on $\tilde{H}_i$ instead of $\bar{H}_i^{(t,s)}$. The derivation of \eqref{eq:perturbed_f:derivative4a} is the same as Sec.~\ref{app:entropy-derivatives:1b} except that the steps should be done on $\tilde{H}_i, \tilde{H}_j$ instead of $\bar{H}_i^{(t,s)}, \bar{H}_j^{(t,s)}$.

\subsection{Proof of \eqref{eq:perturbed_f:derivative3a} and \eqref{eq:perturbed_f:derivative3b}}
For $\underline{\sfx}(\epsilon)$ independent of $\delta$, from \eqref{eq:E-epsilon-bar} we have
\begin{align*}
\frac{d \mathbb{E} [\bar{\epsilon}^{(t,s)}]}{d\delta} = \frac{1}{n^{\theta}}(1-\epsilon) e^{-\frac{K}{RT}(s \tilde{x}^{t} 
+ \sum_{t'=1}^{t-1} \tilde{x}^{t'})} = \frac{1-\mathbb{E} [\bar{\epsilon}^{(t,s)}]}{n^{\theta}-\delta}
\quad\quad{\rm and}\quad\quad\frac{d^2 \mathbb{E} [\bar{\epsilon}^{(t,s)}]}{d\delta^2} = 0. 
\end{align*}
Together with \eqref{eq:perturbed_f:derivative1a} and \eqref{eq:perturbed_f:derivative1b} we can immediately derive
\begin{align*}
\frac{d}{d \delta} h_{t,s; \epsilon, \delta} 
	& = \frac{d \mathbb{E} [\bar{\epsilon}^{(t,s)}]}{d\delta} \frac{d h_{t,s; \epsilon, \delta} }{d \mathbb{E} [\bar{\epsilon}^{(t,s)}]} 
	\- = - \frac{\ln 2}{n^{1+\theta}(1-\delta/n^{\theta})} \sum_{i=1}^n (1 - \mathbb{E}\<\sigma_i\>_{t,s;\epsilon,\delta}) \\
\frac{d^2}{d \delta^2} h_{t,s; \epsilon} 
	& = \frac{d}{d \delta} \bigg( \frac{d \mathbb{E} [\bar{\epsilon}^{(t,s)}]}{d\delta} \frac{d h_{t,s; \epsilon, \delta} }{d \mathbb{E} [\bar{\epsilon}^{(t,s)}]} \bigg)
	\- = \bigg( \frac{d \mathbb{E} [\bar{\epsilon}^{(t,s)}]}{d\delta} \bigg)^2 \frac{d^2 h_{t,s; \epsilon,\delta}}{d \mathbb{E} [\bar{\epsilon}^{(t,s)}]^2} + \frac{d^2 \mathbb{E} [\bar{\epsilon}^{(t,s)}]}{d\delta^2} \frac{d h_{t,s; \epsilon, \delta}}{d \mathbb{E} [\bar{\epsilon}^{(t,s)}]} \\
	& = \frac{\ln 2}{n^{1+2\theta}(1-\delta/n^{\theta})^2} \sum_{i \neq j} \mathbb{E} [\< \sigma_i \sigma_j \>_{t,s;\epsilon,\delta} 
- \< \sigma_i \>_{t,s;\epsilon,\delta} \< \sigma_j \>_{t, s;\epsilon,\delta}].
\end{align*}
The first equality of \eqref{eq:perturbed_f:derivative3a} follows from applying the same argument in \eqref{eq:recover-original-gibbs} to $\tilde{H}_i$.

\subsection{Proof of \eqref{eq:Q1:delta-der}}
We rearrange \eqref{eq:perturbed_f:derivative3a} to obtain
\begin{align}
\frac{1}{n} \sum_{i=1}^{n} \mathbb{E}\<\sigma_i\>_{t,s;\epsilon,\delta} = \frac{n^{\theta}(1-\delta/n^{\theta})}{\ln 2} \frac{d}{d\delta} h_{t,s;\epsilon,\delta}.
\end{align}
Then using \eqref{eq:perturbed_f:derivative3a} and \eqref{eq:perturbed_f:derivative3b} we have
\begin{align*}
	& \frac{d}{d\delta} \big ( \frac{1}{n} \sum_{i=1}^{n} \mathbb{E}\<\sigma_i\>_{t,s;\epsilon,\delta} \big )
	\- = \frac{n^{\theta}(1-\delta/n^{\theta})}{\ln 2} \frac{d^2}{d\delta^2} h_{t,s;\epsilon,\delta} - \frac{1}{\ln 2} \frac{d}{d\delta} h_{t,s;\epsilon,\delta} \\
	& = \frac{1}{n^{1+\theta}(1-\delta/n^{\theta})} \big ( \sum_{i \neq j} \mathbb{E} [\< \sigma_i \sigma_j \>_{t,s;\epsilon,\delta} 
- \< \sigma_i \>_{t,s;\epsilon,\delta} \< \sigma_j \>_{t, s;\epsilon,\delta}] - \sum_{i=1}^{n} (1 - \mathbb{E}\<\sigma_i\>_{t,s;\epsilon,\delta}) \big ) \\
	& = \frac{1}{n^{1+\theta}(1-\delta/n^{\theta})} \sum_{i, j=1}^{n} \mathbb{E} [\< \sigma_i \sigma_j \>_{t,s;\epsilon,\delta} 
- \< \sigma_i \>_{t,s;\epsilon,\delta} \< \sigma_j \>_{t, s;\epsilon,\delta}]
\end{align*}
where the last equality uses one of the Nishimori identities \eqref{eq:Nishimori}, namely $\mathbb{E} \< \sigma_i \> = \mathbb{E} [ \< \sigma_i \>^2 ]$.

%

\section{A bound on differences of derivatives due to convexity}
\label{app:bound-by-convexity}

Let $G(x)$ and $g(x)$ be two functions convex in $x$. Convexity implies that for any $\xi>0$ we have
\begin{align*}
G'(x) - g'(x)
 	& \leq \frac{G(x+\xi) - G(x)}{\xi} - g'(x) \nonumber \\
	& \leq \frac{G(x+\xi) - G(x)}{\xi} - g'(x) + g'(x+\xi) - \frac{g(x+\xi) - g(x)}{\xi} \nonumber \\
	& = \frac{G(x+\xi) - g(x+\xi)}{\xi} - \frac{G(x) - g(x)}{\xi} + g'(x+\xi) - C_\xi^+(x) \,,  \\
G'(x) - g'(x)
	& \geq \frac{G(x) - G(x-\xi)}{\xi} - g'(x) + g'(x-\xi) - \frac{g(x) - g(x-\xi)}{\xi} \nonumber \\
	& = \frac{G(x) - g(x)}{\xi} - \frac{G(x-\xi) - g(x-\xi)}{\xi} - C_\xi^-(x) \,.
\end{align*}
where $C_\xi^+(x) \equiv g'(x+\xi) - g'(x) \geq 0$ and $C_\xi^-(x) \equiv g'(x) - g'(x-\xi) \geq 0$.

The combined result of the above two inequalities is
\begin{align}
|G'(x) - g'(x)| \leq \xi^{-1} \sum_{u \in \{x-\xi, x, x+\xi\}} |G(u)-g(u)| + C_\xi^+(x) + C_\xi^-(x).
\end{align}

%

\section{Concentration of free entropy}
\label{appendix:free-energy}

%

Let $\mathcal{J}$ collect both the realization of $\underline{\tilde{J}}$ and the graph realization of all the factor nodes carrying elements in $\underline{\tilde{J}}$. Let $\mathcal{U}$ collect both the realization of $\underline{U}$ and the graph realization of all the half edges carrying elements in $\underline{U}$. The proof of Lemma~\ref{thm:concentration:free_energy} can be decomposed into the following three lemmas. 
We stress that the three Lemmas \ref{thm:free-energy:concentration0}, \ref{thm:free-energy:concentration1} and \ref{thm:free-energy:concentration2} are valid under the condition that $\underline{\tilde{J}}, \underline{U}$ are non-negative such that we can make use of the consequence $\< \sigma_S \>_{t,s; \epsilon, \delta} \geq 0$ where $S$ is any subset of $\{1, \dots, n\}$. Finally recall definitions \eqref{eq:htse} and \eqref{eq:Htse}. 

\begin{lemma}[Concentration w.r.t. $\underline{H}$]
For any $s, \epsilon, \delta$ all in $[0,1]$, $t=1, \dots, T$, $\nu>0$ and any realization $\underline{H}$ we have 
\begin{align}
\mathbb{P}(| {H}_{t,s; \epsilon,\delta} - \mathbb{E}_{\underline{H}}{H}_{t,s; \epsilon,\delta} | \geq \nu/3)
	& \leq 2\exp \Big( - \frac{2 n \nu^2}{(3\ln 2)^2} \Big).
	\label{eq:free-energy:concentration0}
\end{align}
\label{thm:free-energy:concentration0}
\end{lemma}
\begin{lemma}[Concentration w.r.t. $\mathcal{J}$]
For any $s, \epsilon, \delta$ all in $[0,1]$, $t=1, \dots, T$, $\nu>0$ and any realization $\mathcal{J}$ there exists a constant $C_1>0$ such that  
\begin{align}
\mathbb{P}( | \mathbb{E}_{\underline{H}}{H}_{t,s;\epsilon,\delta} - \mathbb{E}_{\underline{H}, \mathcal{J}}{H}_{t,s;\epsilon,\delta} | \geq \nu/3 ) 
	& \leq 3 \exp(- n \nu^2 C_1).
	\label{eq:free-energy:concentration1}
\end{align}
\label{thm:free-energy:concentration1}
\end{lemma}

\begin{lemma}[Concentration w.r.t. $\mathcal{U}$]
For any $s, \epsilon, \delta$ all in $[0,1]$, $t=1, \dots, T$, $\nu>0$ and any realization $\mathcal{U}$ there exists a constant $C_2>0$ such that  
\begin{align}
\mathbb{P}( | \mathbb{E}_{\underline{H}, \mathcal{J}} {H}_{t,s;\epsilon,\delta}-{h}_{t,s;\epsilon,\delta} | \geq \nu/3 ) 
	& \leq 3 \exp(- n \nu^2 C_2).
	\label{eq:free-energy:concentration2}
\end{align}
\label{thm:free-energy:concentration2}
\end{lemma}

Lemmas~\ref{thm:free-energy:concentration0} to \ref{thm:free-energy:concentration2} are consequences of McDiarmid's inequality, which states that if $X_1, \dots, X_N$ are independent variables and $g$ is a function satisfying the bounded difference property
\begin{align*}
|g(x_1, \dots, x_i, \dots, x_N) - g(x_1, \dots, x'_i, \dots, x_N)| \leq d_i \qquad \forall \ i = 1, \dots, N
\end{align*}
then for any $\nu>0$ we have
\begin{align*}
\mathbb{P}( |g(\underline{X})-\mathbb{E}_{\underline{X}}g(\underline{X})| \geq \nu ) \leq 2 \exp \Big ( - \frac{2\nu^2}{\sum_{i=1}^{N} d_i^2} \Big ).
\end{align*}
We provide the proof of those three lemmas at the end of this section. 

From the triangle inequality and the union bound we have
\begin{align}
\mathbb{P}(| {H}_{t,s; \epsilon,\delta} - {h}_{t,s; \epsilon,\delta} | \geq \nu)
	& \leq \mathbb{P}(| {H}_{t,s; \epsilon,\delta} - \mathbb{E}_{\underline{H}}{H}_{t,s; \epsilon,\delta} | \geq \nu/3) \nonumber\\
	&\qquad + \mathbb{P}( |\mathbb{E}_{\underline{H}}{H}_{t,s;\epsilon,\delta} - \mathbb{E}_{\underline{H}, \mathcal{J}}{H}_{t,s;\epsilon,\delta} | \geq \nu/3 ) \nonumber \\
	& \qquad + \mathbb{P}( |\mathbb{E}_{\underline{H}, \mathcal{J}} {H}_{t,s;\epsilon,\delta} - {h}_{t,s;\epsilon,\delta}] | \geq \nu/3 ).
	\label{eq:free-energy:ms1a}
\end{align}
From \eqref{eq:free-energy:concentration0}, \eqref{eq:free-energy:concentration1}, \eqref{eq:free-energy:concentration2}
\begin{align}
\mathbb{P}(| {H}_{t,s; \epsilon,\delta} - {h}_{t,s; \epsilon,\delta} | \geq \nu)
	& \leq 8 \exp(- n \nu^2 C_0).
	\label{eq:free-energy:ms1b}
\end{align}
where $C_0 \equiv \min \{ \frac{2}{(3\ln 2)^2}, C_1, C_2 \}$.
Let $D \equiv | {H}_{t,s; \epsilon,\delta} - {h}_{t,s; \epsilon,\delta} |$. We have
\begin{align}
\int_0^{\infty} d\nu \, \nu \mathbb{P}(D \geq \nu)
	& = \int_0^{\infty} d\nu \, \nu\, \mathbb{E}_{D} \mathbb{I}(D \geq \nu)  = \mathbb{E}_{D} \int_{0}^{\infty} d\nu \, \nu\, \mathbb{I}(D \geq \nu)  \nonumber\\
	 &= \mathbb{E}_{D}  \int_{0}^{D} d\nu \, \nu  = \frac{1}{2} \mathbb{E}_{D} [D^2]. \label{eq:reverse-markov}
\end{align}
Substituting \eqref{eq:free-energy:ms1b} into \eqref{eq:reverse-markov}, we have the required bound for Lemma~\ref{thm:concentration:free_energy} with $C = 8/C_0$:
\begin{align*}
\mathbb{E} \big [ ( {H}_{t,s;\epsilon,\delta} - {h}_{t,s;\epsilon,\delta})^2 \big ] = 2 \int_0^{\infty} d\nu \, \nu P(| {H}_{t,s; \epsilon,\delta} - {h}_{t,s; \epsilon,\delta} | \geq \nu)  \leq 16 \int_0^{\infty} d\nu \, \nu\, e^{- n \nu^2 C_0} = \frac{C}{n}.
\end{align*}
\subsection{Proof of Lemma~\ref{thm:free-energy:concentration0}}
Consider $g(H_1, \dots, H_n) \equiv {H}_{t,s; \epsilon,\delta}$ with $H_i \in \{0, \infty \}$ (note that ${H}_{t,s; \epsilon,\delta}$ given by \eqref{eq:Htse} 
is already averaged over $\underline{\widetilde H}$, but not over $\underline H$). As for all $i = 1, \dots, n$ the function $g$ satisfies 
\begin{align*}
| g(H_1, \dots, H_i, \dots, H_n) - &g(H_1, \dots, H'_i, \dots, H_n) |
	 = \Big | \frac{1}{n} \mathbb{E}_{\underline{\tilde{H}}} \ln \< e^{H_i (\sigma_i - 1)}\>_{t,s; \epsilon, \delta} \Big | \\
	& = \Big| \frac{1}{n} \mathbb{E}_{\underline{\tilde{H}}}  \ln (1 + \< \sigma_i \>_{t,s; \epsilon,\delta}\tanh H_i  )  - \frac{1}{n} \mathbb{E}_{\underline{\tilde{H}}}  \ln (1 + \tanh H_i)  \Big | \\
	& \leq \frac{\ln 2}{n}.
\end{align*}
McDiarmid's inequality immediately gives the lemma.
\subsection{Proof of Lemma~\ref{thm:free-energy:concentration1}}
Let $|\underline{\tilde{J}}|$ be the number of components of the vector $\underline{\tilde{J}}$. From the construction
of $\mathcal{G}_{t,s}$ in Sec. \ref{subsec:tsmodel}, we have $\mathbb{E}[|\underline{\tilde{J}}|] = \frac{\alpha n}{T}(T-t+1-s) \leq \alpha n$. 
Set $m_{\mathrm{max}} = (1+\gamma) \alpha n$ for $\gamma > 0$. 
The probability of the event $|\underline{\tilde{J}}| > m_{\mathrm{max}}$ can be bounded by a relaxed form of the Chernoff bound as follows.
\begin{lemma}[Chernoff bound, {\cite[Theorem 4.4]{Mitzenmacher:2005}}]
Let $X = \sum_{i=1}^{N} X_i$ where $X_i=1$ with probability $p_i$ and $X_i=0$ with probability $1-p_i$, and all $X_i$ 
are independent. Let $\mu = \mathbb{E}[X] = \sum_{i=1}^{N} p_i$. Then for all $\gamma > 0$ 
\begin{align*}
\mathbb{P}(X > (1+\gamma) \mu) \leq \exp\Big(- \frac\mu3\min \{\gamma, \gamma^2 \} \Big).
\end{align*}
\label{thm:Chernoff}
\end{lemma}
By the Chernoff bound we have
\begin{align}
\mathbb{P}( |\underline{\tilde{J}}| > m_{\mathrm{max}} ) \leq \exp \Big ( -\frac{\alpha n}{3} \min \{\gamma, \gamma^2 \} \Big ).
\label{eq:Chernoff1}
\end{align}
Conditioned on $|\underline{\tilde{J}}| \leq m_{\mathrm{max}}$, we can have the representation $\mathcal{J} = (c_1, \dots, c_{m_{\mathrm{max}}})$ where for $a = 1, \dots, m_{\mathrm{max}}$ 
the profile $c_a \equiv (A_a, \tilde{J}_a)$ encodes that a factor node with weight $\tilde{J}_a$ is connected to a $K$-tuple identified by $A_a$. For $m < a \leq m_{\mathrm{max}}$ we denote $c_a = (\emptyset,0)$. 

Now consider $g(c_1, \dots, c_{m_{\mathrm{max}}}) \equiv \mathbb{E}_{\underline{H}}  {H}_{t,s;\epsilon,\delta} $ and pick a $c_a$ for a given $a$. 
Let $c'_a \equiv \big ( A'_a, \tilde{J}'_{a} \big )$ be a new profile with either $A_a \neq A'_a$ or $\tilde{J}_a \neq \tilde{J}'_{a}$. Also let $c''_a \equiv (A_a, 0)$ and $c'''_a \equiv (A'_a, 0)$. Note that $g(c_1, \dots, c''_a, \dots, c_{m_{\mathrm{max}}}) = g(c_1, \dots, c'''_a, \dots, c_{m_{\mathrm{max}}})$. We then have
\begin{align*}
	& |g(c_1, \dots, c_a, \dots, c_{m_{\mathrm{max}}}) - g(c_1, \dots, c'_a, \dots, c_{m_{\mathrm{max}}}) | \\
	=\ & |g(c_1, \dots, c_a, \dots, c_{m_{\mathrm{max}}}) - g(c_1, \dots, c''_a, \dots, c_{m_{\mathrm{max}}}) \nonumber\\
	&\qquad+ g(c_1, \dots, c'''_a, \dots, c_{m_{\mathrm{max}}}) - g(c_1, \dots, c'_a, \dots, c_{m_{\mathrm{max}}}) | \\
	\leq\ & |g(c_1, \dots, c_a, \dots, c_{m_{\mathrm{max}}}) - g(c_1, \dots, c''_a, \dots, c_{m_{\mathrm{max}}}) | \nonumber\\
	&\qquad+ | g(c_1, \dots, c'''_a, \dots, c_{m_{\mathrm{max}}}) - g(c_1, \dots, c'_a, \dots, c_{m_{\mathrm{max}}}) | \\
	=\ & \Big | \frac{1}{n} \mathbb{E}_{\underline{\tilde{H}}, \underline{H}}  \ln \< e^{\tilde{J}_a (\sigma_{A_a} - 1)}\>_{t,s; \epsilon, \delta}  \Big | + \Big| \frac{1}{n} \mathbb{E}_{\underline{\tilde{H}}, \underline{H}}  \ln \< e^{\tilde{J}'_{a} (\sigma_{A'_a} - 1)}\>_{t,s; \epsilon, \delta}  \Big | \\
	=\ & \Big | \frac{1}{n} \mathbb{E}_{\underline{\tilde{H}}, \underline{H}}  \ln (1 + \< \sigma_{A_a} \>_{t,s; \epsilon, \delta}\tanh \tilde{J}_a  ) - \frac{1}{n} \mathbb{E}_{\underline{\tilde{H}}, \underline{H}}  \ln (1 + \tanh \tilde{J}_{a})  \Big | \\
	& \qquad + \Big| \frac{1}{n} \mathbb{E}_{\underline{\tilde{H}}, \underline{H}}  \ln (1 +  \< \sigma_{A'_a} \>_{t,s; \epsilon, \delta} \tanh \tilde{J}'_{a}) - \frac{1}{n} \mathbb{E}_{\underline{\tilde{H}}, \underline{H}}  \ln (1 + \tanh \tilde{J}'_{a})  \Big | \\
	\leq\ & \frac{2\ln 2}{n}
\end{align*}
This allows the use of McDiarmid's inequality to obtain
\begin{align}
\mathbb{P}( |\mathbb{E}_{\underline{H}}{H}_{t,s;\epsilon,\delta} - \mathbb{E}_{\underline{H}, \mathcal{J}}{H}_{t,s;\epsilon,\delta} | \geq \nu/3 \  | \ |\underline{\tilde{J}}| \leq m_{\mathrm{max}}) \leq 2 \exp \Big ( - \frac{n \nu^2 }{18\alpha (\ln 2)^2} \Big ).
\label{eq:McDiarmid2}
\end{align}
Finally, we take the union bound based on \eqref{eq:Chernoff1} and \eqref{eq:McDiarmid2}:
\begin{align*}
\mathbb{P}( |\mathbb{E}_{\underline{H}}{H}_{t,s;\epsilon,\delta} - \mathbb{E}_{\underline{H}, \mathcal{J}}{H}_{t,s;\epsilon,\delta} | \geq \nu/3 ) 
	& \leq 2 \exp \Big( - \frac{n \nu^2 }{18 \alpha (\ln 2)^2} \Big ) + \exp \Big ( -\frac{\alpha n}{3} \min \{\gamma, \gamma^2 \} \Big ).
\end{align*}
Choosing $\nu^2 = \min \{\gamma, \gamma^2 \}$ and $C_1 = \min \{ \frac{1}{18 \alpha (\ln 2)^2}, \frac{\alpha}{3} \}$, we obtain the lemma.
\subsection{Proof of Lemma~\ref{thm:free-energy:concentration2}}
This proof can adopt the same presentation as in the proof of Lemma~\ref{thm:free-energy:concentration1} by noting that in the construction of $\sG_{t,s}$ the Poisson process of adding half edges with weight $U^{(t')}_{a \rightarrow i}$ can be rephrased as follows:
\begin{enumerate}
\item ({\it Create all the messages without specifying their location}): We draw the random numbers $e_{i}^{(t')}$, $e_{i,s}^{(t)}$ and create the associated number of copies of $U^{(t')}$ for $t'=1,\dots,t$. We collect all $U^{(t')}$ to form a set $\{ U_k \}_{k=1}^w$, where $w$ follows a Poisson distribution with mean $\frac{n \alpha K}{T}(t-1+s) \leq n \alpha K$.
\item ({\it Specify the location of the messages}): Given the number $w$ and the set $\{ U_k \}$, we attach each $U_{k}$ to variable node $i$ chosen randomly and uniformly. 
\end{enumerate}  

Let $w_{\mathrm{max}} = (1+\gamma) n \alpha K$. The Chernoff bound (Lemma~\ref{thm:Chernoff}) provides that
\begin{align}
\mathbb{P}(w > w_{\mathrm{max}}) \leq \exp \Big ( - \frac{n \alpha K}{3} \min \{ \gamma, \gamma^2 \} \Big ).
\label{eq:Chernoff2}
\end{align}
Conditioned on $w \leq w_{\mathrm{max}}$, we have the representation $\mathcal{U} = (u_1, \dots, u_{w_\mathrm{max}})$ where for $k = 1, \dots, w_{\mathrm{max}}$ the profile $u_k = (i_k, U_k)$ represents that a half edge with weight $U_k$ is connected to variable node $i_k$. For $w < k \leq w_{\mathrm{\max}}$ we denote $u_k = (\emptyset,0)$.

Now consider $g(u_1, \dots, u_{m_{\mathrm{max}}}) \equiv \mathbb{E}_{\underline{H}, \mathcal{J}}  {h}_{t,s;\epsilon,\delta}$ and pick any $u_k$. Let $u'_k \equiv (i'_k, U_k)$ be a new profile with either $i_k \neq i'_k$ or $U_k \neq U'_k$. Also let $u''_k = (i_k, 0)$ and $u''_k = (i'_k, 0)$. Note that $g(u_1, \dots, u''_k, \dots, u_{m_{\mathrm{max}}}) = g(u_1, \dots, u'''_k, \dots, u_{w_{\mathrm{max}}})$. We then have 
\begin{align*}
	& |g(u_1, \dots, u_k, \dots, u_{w_{\mathrm{max}}}) - g(u_1, \dots, u'_k, \dots, u_{w_{\mathrm{max}}}) | \\
	=\ & |g(u_1, \dots, u_k, \dots, u_{w_{\mathrm{max}}}) - g(u_1, \dots, u''_k, \dots, u_{w_{\mathrm{max}}}) \nonumber\\
	&\qquad+ g(u_1, \dots, u'''_k, \dots, u_{w_{\mathrm{max}}}) - g(u_1, \dots, u'_k, \dots, u_{w_{\mathrm{max}}}) | \\
	\leq\ & |g(u_1, \dots, u_k, \dots, u_{w_{\mathrm{max}}}) - g(u_1, \dots, u''_k, \dots, u_{w_{\mathrm{max}}})| \nonumber\\
	&\qquad+ |g(u_1, \dots, u'''_k, \dots, u_{w_{\mathrm{max}}}) - g(u_1, \dots, u'_k, \dots, u_{w_{\mathrm{max}}}) | \\
	=\ & \Big | \frac{1}{n} \mathbb{E}_{\underline{\tilde{H}}, \underline{H}, \mathcal{J}}  \ln \< e^{U_k (\sigma_{i_k} - 1)}\>_{t,s; \epsilon, \delta}  \Big | + \Big | \frac{1}{n} \mathbb{E}_{\underline{\tilde{H}}, \underline{H}, \mathcal{J}}  \ln \< e^{U'_k (\sigma_{i'_k} - 1)}\>_{t,s; \epsilon, \delta}  \Big | \\
	=\ & \Big | \frac{1}{n} \mathbb{E}_{\underline{\tilde{H}}, \underline{H}, \mathcal{J}}  \ln (1 + \< \sigma_{i_k} \>_{t,s; \epsilon, \delta}\tanh U_k  )  - \frac{1}{n} \mathbb{E}_{\underline{\tilde{H}}, \underline{H}, \mathcal{J}}  \ln (1 + \tanh U_k)  \Big| \\
	& \qquad + \Big | \frac{1}{n} \mathbb{E}_{\underline{\tilde{H}}, \underline{H}, \mathcal{J}}  \ln (1 + \< \sigma_{i'_k} \>_{t,s; \epsilon, \delta} \tanh U'_k  )  - \frac{1}{n} \mathbb{E}_{\underline{\tilde{H}}, \underline{H}, \mathcal{J}} \ln (1 + \tanh U'_k)  \Big | \\
	\leq\ & \frac{2\ln 2}{n}.
\end{align*}
McDiarmid's inequality is then used to obtain
\begin{align}
\mathbb{P}( |\mathbb{E}_{\underline{H}, \mathcal{J}}{H}_{t,s;\epsilon,\delta} - {h}_{t,s;\epsilon,\delta} | \geq \nu/3 \ \big | \ w \leq w_{\mathrm{max}}) \leq 2 \exp \Big ( - \frac{n \nu^2}{18(\ln 2)^2 \alpha K} \Big).
\label{eq:McDiarmid3}
\end{align}
Finally, we take the union bound based on \eqref{eq:Chernoff2} and \eqref{eq:McDiarmid3}:
\begin{align*}
\mathbb{P}( |\mathbb{E}_{\underline{H}, \mathcal{J}} {H}_{t,s;\epsilon,\delta} - {h}_{t,s;\epsilon,\delta}] | \geq \nu/3 ) 
	& \leq 2 \exp \Big ( - \frac{n \nu^2}{18(\ln 2)^2 \alpha K} \Big ) + \exp \Big ( -\frac{n \alpha K}{3} \min \{\gamma, \gamma^2 \} \Big ).
\end{align*}
Choosing $\nu^2 = \min \{\gamma, \gamma^2 \}$ and $C_2 = \min \{ \frac{R}{18(\ln 2)^2 K}, \frac{K}{3R} \}$, we obtain the lemma.
%

\section{Illustration of the replica formula}\label{app:plot-BEC}
Recall that the distribution of $V$ is denoted by $\sfx = x \Delta_0 + (1-x)\Delta_{\infty}, x \in [0,1]$. From \eqref{BPfp} the distribution of U is $\tilde{\sfx} = \tilde{x} \Delta_0 + (1-\tilde{x}) \Delta_{\infty}$ where $\tilde{x} = 1-(1-q)(1-x)^{K-1} \in [0,1]$. The first term of \eqref{eq:RSFreeEntropy1} can be simplified as
\begin{align*}
	& \mathbb{E} \ln \Big ( \prod_{B=1}^{l} (1 + \tanh U_B ) + \prod_{B=1}^{l} (1 - \tanh U_B ) \Big ) \nonumber \\
	=\ & \mathbb{E}_{l} \mathbb{E}_{U} \Big [ l \ln (1 + \tanh U) \Big ] + \mathbb{E}_{l} \mathbb{E}_{\underline{U}} \ln \Big ( 1 +  \prod_{B=1}^{l} \frac{1-\tanh U_B}{1+\tanh U_B} \Big ) \\
	=\ & \mathbb{E}_{l} \big [ l \big ] (1-\tilde{x}) \ln 2 + \mathbb{E}_{l} \big [ \tilde{x}^l \big ] \ln 2 \\
	=\ & \alpha K (1-q)(1-x)^{K-1} \ln 2 + e^{- \alpha K (1-q)(1-x)^{K-1}} \ln 2
\end{align*}
The remaining terms of \eqref{eq:RSFreeEntropy1} can also be simplified straightforwardly. Eventually, for BEC we can write \eqref{eq:RSFreeEntropy1} with a scalar expression:
\begin{align*}
{h}_{\mathrm{RS}}(x) = (\ln 2) \big [ e^{- \alpha K (1-q)(1-x)^{K-1}} + \alpha K (1-q) (1-x)^{K-1} - \alpha(K-1)(1-q)(1-x)^{K} - \alpha (1-q) \big ]
\end{align*}
We illustrate ${h}_{\mathrm{RS}}(x)$ with $K=3$ and $\alpha = 1/5$ in Fig.~\ref{fig:plot-BEC}.

\begin{figure}
    \centering
    \begin{subfigure}[b]{0.4\textwidth}
        \includegraphics[width=6cm]{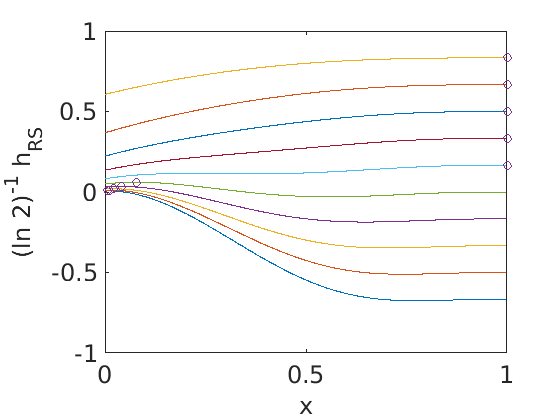}
        \label{fig:plot-BEC:h_RS}
    \end{subfigure}
    \begin{subfigure}[b]{0.4\textwidth}
        \includegraphics[width=6cm]{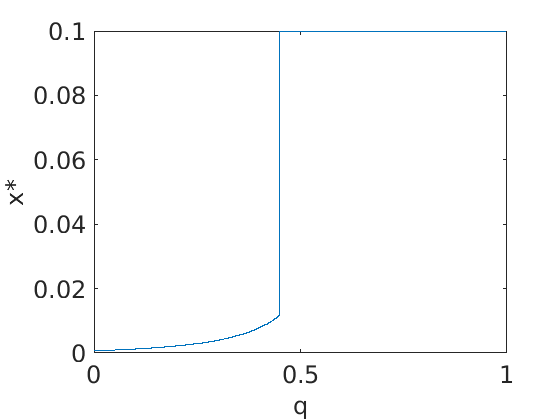}
        \label{fig:plot-BEC:x_star}
    \end{subfigure}
    \caption{Illustration of ${h}_{\mathrm{RS}}(x)$ with $K=3$ and $\alpha = 1/5$. (Left) ${h}_{\mathrm{RS}}(x)$ as a function of $x$ for $q=0, 0.1, 0.2, \dots, 0.9$. ${h}_{\mathrm{RS}}(x)$ increases with $q$ when $x$ is fixed. Circles locate the maximum of ${h}_{\mathrm{RS}}(x)$ for every $q$. (Right) The first order phase transition for $x^*(q) = \arg \max_{x} {h}_{\mathrm{RS}} ( x )$ as a function of $q$}
    \label{fig:plot-BEC}
\end{figure}

\section*{Acknowledgments}
Jean Barbier and Chun Lam Chan acknowledge the SNSF grant no. 200021-156672.
\bibliographystyle{unsrt_abbvr}

%
\end{document}